\documentclass[12pt]{article}
\usepackage{amsmath}
\usepackage{graphicx}
\usepackage{enumerate}
\usepackage{natbib}
\usepackage{url} 
\usepackage{bigints}
\usepackage[normalem]{ulem}

\newcommand{\blind}{1}

\usepackage[dvipsnames]{xcolor} 
\usepackage[T1]{fontenc}
\usepackage[english]{babel}
\usepackage{graphicx}			
\usepackage{authblk}
\usepackage{amsmath} 			
\usepackage{amssymb} 			
\usepackage{calligra}			
\usepackage{calrsfs}
\usepackage{bm}					
\usepackage{bbm}
\usepackage{amsthm} 			
\usepackage{listings}			
\usepackage{multicol}
\usepackage{subcaption}         
\usepackage{multirow} 
\usepackage{enumitem}			
\usepackage{newtxmath}			
\usepackage[fleqn]{mathtools}
\usepackage{tabularx,ragged2e}
\usepackage[ruled,vlined]{algorithm2e} 
\usepackage{adjustbox}
\usepackage{lineno} 
\usepackage{hyperref}

\usepackage{array}

\hypersetup{
    colorlinks  = true,
    linkcolor   = black,
    urlcolor    = black,
    citecolor   = blue
}

\usepackage[nameinlink]{cleveref} 
\newtheorem{theorem}{Theorem}

\newtheorem{lemma}{Lemma}
\newtheorem{corollary}{Corollary}

\newcommand{\iid}{\stackrel{\mathrm{iid}}{\sim}}
\newcommand{\ind}{\stackrel{\mathrm{ind}}{\sim}}

\newcommand{\myequal}[1]{\mathrel{\overset{\makebox[0pt]{\mbox{\normalfont\tiny\sffamily #1}}}{=}}}

\newcommand{\mydiff}{\mathrm{d}}

\newcommand{\indic}{\mathbb{I}}
\newcommand{\E}{\mathbb{E}}

\newcommand{\G}{\mathbb{G}}
\renewcommand{\P}{\mathbb{P}}
\newcommand{\X}{\mathbb{X}}

\newcommand{\R}{\mathbb{R}}

\newcommand{\bmc}{\bm{c}}

\newcommand{\bmi}{\bm{i}}
\newcommand{\bmn}{\bm{n}}
\newcommand{\bmr}{\bm{r}}
\newcommand{\bms}{\bm{s}}
\newcommand{\bmt}{\bm{t}}
\newcommand{\bmu}{\bm{u}}
\newcommand{\bmw}{\bm{w}}
\newcommand{\bmx}{\bm{x}}
\newcommand{\bmy}{\bm{y}}

\newcommand{\bmpsi}{\bm{\psi}}
\newcommand{\bmbeta}{\bm{\beta}}
\newcommand{\bmgamma}{\bm{\gamma}}

\newcommand{\bmmu}{\bm{\mu}}
\newcommand{\bmtheta}{\bm{\theta}}
\newcommand{\bmtau}{\bm{\tau}}
\newcommand{\bmxi}{\bm{\xi}}

\newcommand{\bmS}{\bm{S}}
\newcommand{\bmU}{\bm{U}}

\newcommand{\widetrho}{\widetilde{\rho}}

\begin{document}

\def\spacingset#1{\renewcommand{\baselinestretch}%
{#1}\small\normalsize} \spacingset{1}


\if1\blind
{
  \title{\bf Hierarchical Mixture of Finite Mixtures}
  \author{
  \normalsize
  Alessandro Colombi\thanks{We would like to thank Mario Beraha for the helpful discussion.}\hspace{.2cm}\\
    Department of Economics, Management and Statistics, Università degli studi di Milano-Bicocca, Italy\\
    and \\
    Raffaele Argiento \\
    Department of Economics, Università degli studi di Bergamo, Italy\\
    and \\
    Federico Camerlenghi \\
    Department of Economics, Management and Statistics, Università degli studi di Milano-Bicocca, Italy\\
    and \\
    Lucia Paci\\
    Department of Statistical Sciences, Università Cattolica del Sacro Cuore, Milan, Italy}
    \date{\empty}
  \maketitle
} \fi

\if0\blind
{
  \bigskip
  \bigskip
  \bigskip
  \begin{center}
    {\LARGE\bf Mixture modeling via vectors of normalized independent finite point processes}
\end{center}
  \medskip
} \fi


\bigskip
\abstract{
	Statistical modelling in the presence of data organized in groups is a crucial task in Bayesian statistics.
	The present paper conceives a mixture model based on a novel family of Bayesian priors designed for multilevel data and obtained by normalizing a finite point process. In particular, the work extends the popular Mixture of Finite Mixture model to the hierarchical framework to capture heterogeneity within and between groups. A full distribution theory for this new family and the induced clustering is developed, including the marginal, posterior, and predictive distributions.  Efficient marginal and conditional Gibbs samplers are designed to provide posterior inference. The proposed mixture model overcomes the Hierarchical Dirichlet Process, the utmost tool for handling multilevel data, in terms of analytical feasibility, clustering discovery, and computational time. The motivating application comes from the analysis of shot put data, which contains performance measurements of athletes across different seasons. In this setting, the proposed model is exploited to induce clustering of the observations across seasons and athletes. By linking clusters across seasons, similarities and differences in athletes' performances are identified. 
}
{\it Keywords:} Model-based clustering, Multilevel data, Partial exchangeability, Sports analytics, Vector of finite Dirichlet processes.

\vfill

\spacingset{1.5} 


\section{Introduction}
\label{section:intro}
Statistical modelling of population heterogeneity is a recurrent challenge in real-world applications. In this study, our focus shifts towards hierarchical data scenarios, where observations emerge from distinct groups (or levels) and our objective is to model heterogeneity both within and between these groups.
Specifically, heterogeneity within groups is handled via mixture modelling to get group-specific clustering of observations, as well as density estimation.
	Concurrently, between group heterogeneity can be addressed through two extreme modelling choices: (i) pooling all observations; (ii) conducting independent analyses for each of the $d$ groups. However, both alternatives pose limitations. In the first case, differences across groups are not accounted for, while in the second one, groups are not linked, preventing sharing of statistical strength. In this work, a tunable balance between the two alternatives is proposed to model heterogeneity across groups. 

In Bayesian formalism, the sharing of information is naturally achieved through hierarchical modelling; parameters are shared among groups, and the randomness of the parameters induces dependencies among the groups. 
In particular, in model-based clustering, such sharing 
leads to group-specific clusters that are common among the groups, i.e., 
clusters that have the same interpretation across groups allow to define a global clustering. 
In general, the number of clusters within each group, as well as the number of global clusters, are unknown and need to be inferred from the data. 
In the hierarchical landscape, several Bayesian nonparametric approaches have been proposed,  the pioneering work being the Hierarchical Dirichlet Process (HDP) mixture model \citep{teh2006hierarchical}. 
The HDP has been recently extended to encompass normalized random measures \citep{camerlenghi2019distribution,cremaschi2020} and species sampling models \citep{casarin2020}. 
Despite their mathematical elegance, these methods suffer of considerable computational complexity and costs, posing a barrier for practitioners.
%

This work targets these problems and proposes a finite mixture model for each group, where the random number of mixture components, as well as the mixture parameters, are shared among groups. 
Rather, the mixture weights are assumed to be group-specific in order to accommodate differences between groups.
The proposed approach bridges the gap between infinite and finite mixtures in the hierarchical setting by introducing a novel family of Bayesian priors, named Vector of Finite Dirichlet Processes (Vec-FDP) to capture heterogeneity within and between groups. In particular, we use the Vec-FDP prior to building a new class of hierarchical mixture models, named Hierarchical Mixture of Finite Mixtures (HMFM), that encompasses the popular Mixture of Finite Mixtures (MFM) model \citep{millerharrison, here2inf} as a special case when the number of groups $d=1$.   
Then, following the same approach of \cite{argiento2022annals}, we introduce a more general construction for the mixing weights. This involves the normalization of positive unnormalized weights distributed according to any probability distribution over the positive real numbers. By doing so, we define a broader family of Bayesian priors, referred to as Vector of Normalized Independent Finite Point Processes (Vec-NIFPP), which encompasses the Vec-FDP as a special case when the unnormalized weights are Gamma distributed.
Although we borrow Bayesian nonparametric tools to derive the distributional results and the clustering properties, the model employs a finite, yet random, number of mixture components, enhancing its accessibility to a broader audience, e.g., those interested in model-based clustering and mixture of experts modelling \citep{Jacobs91}.




Posterior inference poses computational challenges in the context of Bayesian nonparametrics, particularly when dealing with hierarchical data. Rather, 
leveraging the posterior characterization of the Vec-FDP, we do design an efficient  Markov chain Monte Carlo (MCMC) sampling strategy, which significantly improves the existing methods based on infinite dimensional processes, such as the HDP. 
Notably, given a fixed number of clusters and groups, the HDP's computational time scales quadratically with the volume of data. In contrast, our approach achieves linear scaling.

%
The reason behind such improvement is evident from the restaurant franchise-like representation of the Vec-FDP prior. 
In contrast to the HDP's complex distinction between tables and menus, our representation is straightforward due to the absence of stacked layers of infinite dimensional objects in the model structure.
This simplification enhances model interpretability, which is preserved when moving from a single group to multiple groups, unlike the HDP.
The restaurant franchise metaphor also illuminates the flexibility of 
the sharing of information mechanism induced by the proposed method.
In this regard, a comprehensive simulation study compares the proposed HMFM model with both the HDP mixture model and the MFM model 
assumed independently for each group. Experiments shed light on the advantages of employing joint modelling rather than independent analyses and demonstrate how the excessive sharing of information induced by the HDP may lead to misleading conclusions. Therefore, the HMFM strikes a balanced compromise between the HDP and the independent analyses,  offering a tunable approach that combines the strengths of both methods. 
The motivating example for this work comes from sports analytics. See \citep{page2013} 
for  
Bayesian nonparametric methods in this domain. 
Our methodology finds application in the analysis of data from shot put, a track and field discipline that involves propelling a heavy spherical ball, or \textit{shot}, over the greatest distance attainable. 
The dataset comprises measurements, specifically throw lengths or marks, recorded during professional shot put competitions from 1996 to 2016, for a total of $35,637$ measurements on $403$ athletes.
The data are organized by aligning marks by season to ensure equitable comparison among athletes. Athletes have varying durations, with a maximum of $15$ seasons, which correspond to the $d$ different groups. 
We employ the proposed HMFM model to
cluster athletes' performances within each season while preserving the interpretability of clusters across different seasons. 
This allows us to enrich the understanding of the evolution in athletes' performance trends.
A remarkable finding is that the estimated clusters are gender-free, thanks to the inclusion of an additional season-specific regression parameter. 
Notably, the model identifies a special cluster consisting of six exceptional women's performances achieved by Olympic or world champions;  
no men have ever been able to outperform their competitors in such a neat way.


Summing up, our methodological contribution is to propose a new family of dependent Bayesian nonparametric priors designed for analyzing hierarchical data in a mixture setting. The proposed model 
allows to identify random clusters of observations within and between groups. 
We deeply investigate the theoretical properties of the Vec-NIFPP, including the distribution of the random partition, the predictive distribution of the process, and its correlation structure. Additionally, we fully characterize the posterior distribution of the Vec-NIFPP. 
On the computational side, 
two efficient MCMC procedures based on a conditional and a marginal sampler provide inference on the number of clusters, and the clustering structure 
as well as group-dependent density estimation. 
Both algorithms show improved computational efficiency with respect to common strategies for posterior sampling in hierarchical mixture models.  
The flexibility of the proposed model is illustrated through an extensive simulation study and a real-world application in sport analytics. 

The rest of the paper is organized as follows. 
\Cref{section:HMFM} introduces the HMFM model and defines the local and global clustering.
\Cref{section:VecNIFFP} frames the HMFM model in the larger class of mixture models based on the Vec-NIFPP and provides the main distributional results. 
In \Cref{section:HMFM_properties}, we study the marginal, posterior and predictive distributions of the Vec-FDP, along with the hyperpriors choice and the computational strategies. 
The analysis of shot put data is illustrated in \Cref{section:application}.
A discussion concludes the paper and Supplementary materials complement it.

\section{Hierarchical mixture of finite mixture}
\label{section:HMFM}
Given $d$ groups (or levels), let $\bmy_j = \left(y_{j,1},\dots,y_{j,n_j}\right)$ denote the data collected over $n_j$ individuals in group $j=1,\dots,d$, where $y_{j,i}\in\mathbb{Y}$ and 
$\mathbb{Y}$ is the sampling space.  
We assume that the data in each group $j$ come from a finite mixture of $M$ components, that is
\begin{equation}
	y_{j,1}, \dots, y_{j,n_j}  \mid P_j  \ \iid \ \int_{\Theta}f( \, \cdot \, \mid \theta)P_j({\rm d}\theta) 
	\quad \text{for each } j=1,\dots,d,
	\label{eqn:mixture1}
\end{equation}
where $\{f( \, \cdot\,  \mid \theta), \ \theta\in\Theta \}$ is a parametric family of densities over $\mathbb{Y}$ and $(P_1, \ldots , P_d )$ is a vector of random probability measures over the parameter space $\Theta\subset\R^{s}$. We focus on random probability measures having almost surely discrete realizations. 
More specifically, we define a vector of finite dependent random probability measures $\left(P_1,\dots,P_d\right)$ as follows,
\begin{equation}
	P_j(\cdot) = \sum_{m=1}^{M}\frac{S_{j,m}}{T_j}\delta_{\tau_m}(\cdot),
	\label{eqn:Pj_def}
\end{equation}
where $\delta_{\tau_m}$ stands for the delta-Dirac mass at $\tau_m$, and $T_j = \sum_{m=1}^{M} S_{j,m}$ is referred to as the total mass.\\
As a prior for $M$, we place a $1$-shifted Poisson distribution with parameter $\Lambda$, denoted by $\operatorname{Pois}_1(\Lambda)$, so that we are sure that there always exists at least one mixture component.
Then, conditionally to $M$, $(\tau_1,\dots,\tau_M)$ are common random atoms across the $d$ random probability measures, which are assumed independent and identically distributed (i.i.d.) with common distribution $P_0$, that is a diffuse probability measure on $\Theta$.
Moreover, given $M$, the unnormalized weights $S_{j,m}$ are conditionally independent both within and between the groups. 
In particular, we assume the components of $\bmS_j = \left(S_{j,1},\dots,S_{j,M}\right)$ to be i.i.d. from $\operatorname{Gamma}\left(\gamma_j,1\right)$, independently with respect to $j=1,\dots,d$. Throughout this work, we always refer to a shape-rate parametrization of the gamma distribution.
It is easy to show that the induced prior on the normalized mixture weights $\left(\pi_{j,1}, \ldots , \pi_{j,M}\right)$ is
\[
\left(\pi_{j,1}, \ldots , \pi_{j,M}\right) = \left( \frac{S_{j,1}}{T_j}, \ldots , \frac{S_{j,M}}{T_j} \right) \sim 
{\rm Dir}_M \left(\gamma_j, \ldots , \gamma_j\right), \quad \text{for each }j=1,\dots,d,
\]
where ${\rm Dir}_M \left(\gamma_j, \ldots , \gamma_j\right)$ denotes the $M$-dimensional symmetric Dirichlet distribution with parameter $\gamma_j$.
The mixing measure $P_j$ is obtained by normalization:
\begin{equation}
	P_j(\cdot) = \frac{\mu_j(\cdot)}{\mu_j(\Theta)},
	\label{eqn:normalized_def}
\end{equation}
where $\mu_j(\cdot) = \sum_{m=1}^M S_{j,m}\delta_{\tau_m}(\cdot)$.
The seminal contribution of \citet{regazzini2003distributional} has spurred the construction of random probability measures via the normalization approach, which turns out to be a convenient framework to face posterior inference; see, e.g., \citet{LijNip(14)} \citet{camerlenghi2019distribution}, \citet{cremaschi2020} and \citet{argiento2022annals} for allied contributions.
We point out that, marginally, each component $P_j$ is a finite Dirichlet process as the one described in  \citet{argiento2022annals}. These discrete random measures are
also known as Gibbs-type priors with
a negative parameter \citep{GnedinPitman2006,deblasi2015}.
Our model construction generalizes the work of  \citet{argiento2022annals} by allowing
for the sharing of information across groups thanks to shared atoms and a shared number of components.
Besides, the proposed model retains mathematical tractability thanks to the normalization approach and the conditional independence of the unnormalized weights.
The prior specification for the mixture parameters $\left(M,\bmS,\bmtau\right)$, where $\bmS = \left(S_{1},\dots,S_{d}\right)$ and $\bmtau = \left(\tau_{1},\dots,\tau_{M}\right)$, is equivalent to specify the joint law of the vector $\left(P_1,\dots,P_d\right)$, 
called \textit{Vector of Finite Dirichlet Process} (Vec-FDP) and denoted by 
\begin{equation}
	\label{eqn:Def_VecFDP}    (P_1,\dots,P_d)\sim\operatorname{Vec-FDP}\left(\Lambda,\bmgamma,P_0\right),
\end{equation}
where $\bmgamma = (\gamma_1,\dots,\gamma_d)$.
Summing up, the model can be formulated in the following hierarchical form,
\begin{equation}
	\label{eqn:HMFM_definition}
	\begin{split}
		\bmy_{j,1}, \dots, \bmy_{j,n_j} \mid \bmS_j, \bmtau, M 
		& \ \iid \
		\sum_{m=1}^M w_{j,m}
		f\left(\bmy_{j,i} \mid \tau_{m}\right)\\
		\tau_{1},\dots,\tau_{M} \mid M
		&\ \iid \ P_0(\cdot)\\
		w_{j,1},\dots,w_{j,M} \mid M,\gamma_j
		& \ \iid \ \operatorname{Dir}_M\left(\gamma_j,\dots,\gamma_j\right),
		\quad \text{for } j = 1,\dots,d\\
		M \mid \Lambda & \ \sim \ \operatorname{Pois}_1(\Lambda).
	\end{split}
\end{equation}
We notice that, when $d=1$, the model in Equation \eqref{eqn:HMFM_definition} coincides with the MFM model \citep{millerharrison,here2inf}. It follows that the proposed model is an extension of the MFM model to hierarchical data and so
we refer to it 
as the \textit{Hierarchical Mixture of Finite Mixture} (HMFM) model.

\subsection{Clustering}
\label{subsection:clustering}
It is worth noticing that the mixture model in Equation \eqref{eqn:mixture1} can be equivalently written as
\begin{equation*} \label{eqn:mixture_latent}
	y_{j,i}  \mid \theta_{j,i} \ind f (\, \cdot \, \mid  \theta_{j,i}), \qquad \theta_{j,i} \mid P_j \iid P_j
\end{equation*}
with $i=1, \ldots , n_j$ and $j=1, \ldots , d$. Under this formulation, we get rid of the integral in Equation \eqref{eqn:mixture1} by introducing a latent variable
$\theta_{j,i}$ for each observation $y_{j,i}$.
Conditionally on $\left(P_1,\dots,P_d\right)$, the latent variables $\theta_{j,i}$'s are i.i.d. within the same group and independent across groups. In other words, by virtue of the de Finetti representation theorem, $\bmtheta := \left(\bmtheta_1,\dots,\bmtheta_d\right)$, where $\bmtheta_j=(\theta_{j,1},\dots,\theta_{j,n_j})$, is a sample from a partially exchangeable array of latent variables. This is tantamount to saying that the distribution of the $\theta_{j,i}$'s is invariant under a specific class of permutations; see \citep{Kall(05)} and references therein.
The distributional properties of $\bmtheta$ play a pivotal role in mixture models both in defining the clustering and 
devising efficient computational procedures.
More precisely, since the $P_j$'s in Equation \eqref{eqn:Def_VecFDP} have discrete realizations, almost surely, the latent variables feature ties within and across groups; thus they naturally induce 
both a local 
and a global clustering of the observations.

\paragraph{\textbf{Local clustering}} For each group $j = 1, \ldots  , d$, $P_j$ is almost surely discrete, then ties are expected with positive probability among $\theta_{j,1},\dots,\theta_{j,n_j}$.
Let $K_j:=K_{j,(n_j)}$ be the random number of distinct values in this sample. 
Such group-specific distinct values are collected in the set $\mathcal{T}_j = \left\{\theta^*_{j,1},\dots,\theta^*_{j,K_j}\right\}$. 
Furthermore, let $\widetrho_j = \left\{C_{j,1},\dots,C_{j,K_j}\right\}$ be the random partition of $\{ 1,\dots,n_j\}$ induced by $\mathcal{T}_j$ through the following rule:
\[
\theta_{j,i}\in C_{j,k} \iff \exists k\in \{1,\dots,K_j\} \text{ such that } \theta_{j,i} = \theta^*_{j,k},
\]
for each $i = 1,\dots,n_j$.
Note that the random partition $\widetrho_j$ is exchangeable due to the exchangeability of $\theta_{j,1},\dots,\theta_{j,n_j}$ and it is called \textit{local clustering} of group $j$ (or group-specific clustering).

\paragraph{\textbf{Global clustering}} Since the $P_j$'s share the same support, we also expect ties between groups,
i.e., $\P\left(\theta^*_{j,k} = \theta^*_{j^\prime,k^\prime}\right)>0$, with $j\neq j^\prime$.
We define $\mathcal{T}= \left\{\theta^{**}_1,\dots,\theta^{**}_{K}\right\}$ the set of unique values among the $\theta^*_{j,k}$, $j=1,\dots,d$ and $k=1,\dots,K_j$; we also observe that
$
\mathcal{T} = 
\bigcup_{j=1}^d\mathcal{T}_j .
$
The corresponding random partition $\rho = \left\{C_1,\dots,C_K\right\}$ is induced by $\mathcal{T}$ as follows
\[
\theta_{j,i} \in C_k \iff \exists k\in \{1,\dots,K\} \text{ such that } \theta_{j,i} = \theta^{**}_{k},
\]
for each $j = 1,\dots,d$ and for each $i =  1,\dots,n_j$. 
The random partition $\rho$ is called \textit{global clustering} and the number of global clusters has been denoted by $K:=K_{(n_1,\dots,n_d)}$.
Since values are expected to be shared also across groups, then $K\leq \sum_{j=1}^d K_j$.

To shade light on the relationship between local and global clustering, we introduce
$\rho_j = \left\{C_{j,1},\dots,C_{j,K}\right\}$ 
so that $C_k = \bigcup_{j=1}^d C_{j,k}$. Note that $C_{j,k}$ can be empty for some $k = 1,\dots,K$.
We define 
$n_{j,k} = \left| C_{j,k} \right|$ the number of observations of the $j$-th group in the $k$-th cluster. 
Then, we let
$\bmn_j = \left(n_{j,1},\dots,n_{j,K}\right)$
and the following 
hold true
\begin{equation}
	\sum_{j=1}^{d}n_{j,k} >0, \text{ and }\sum_{k=1}^{K} n_{j,k} = n_j,
	\label{eqn:cluster_vincoli}
\end{equation}
for any $k= 1,\dots,K$ and for any $j=1,\dots,d$, respectively.

The random partition induced by the whole sample $\bm{\theta}$ of size $n:= n_1+\cdots +n_d$ may be described through a probabilistic object called \textit{partially Exchangeable Partition Probability Function} (pEPPF).
The pEPPF, denoted here as $\Pi_K^{(n)}\left(\bmn_1,\dots,\bmn_d\right)$ is the probability distribution of both the local and global clustering, where $\bmn_1,\dots,\bmn_d$ satisfy the constraints given in Equation \eqref{eqn:cluster_vincoli}.   The pEPPF is formally defined as follows:
\begin{equation*}
	\label{eqn:pEPPF_def}
	\Pi_K^{(n)}\left(\bmn_1,\dots,\bmn_d\right) :=\E \left[\bigintssss_{\Theta^K} \prod_{k=1}^K P_j^{n_{j,k}} \left({\rm d} \theta_k^{**}\right)\right].
\end{equation*}
We refer to \citet{camerlenghi2019distribution} for additional details.

Observe that, conditionally to $M$ and $\bmtau$, the unique values $\theta^{**}_1,\dots,\theta^{**}_K$ are such that the following properties hold: (i) $K\leq M$; (ii) there exists $m\in\{1,\dots,M\}$ such that $\theta^{**}_k = \tau_m$, for each $k=1,\dots,K$. 
Property (i) implies that a distinction between allocated mixture components (clusters) and non-allocated components is required \citep{Nobile2004, argiento2022annals}. From property (ii) follows that there are exactly $K$ allocated components collected in the following set:
$$
\mathcal{M}^{(a)} = 
\left\{
m\in\{1,\dots,M\} : \exists k \in\{1,\dots,K\} \text{ such that } \theta^{**}_{k} = \tau_m
\right\}.
$$

\section{Vector of Normalized Independent Finite Point Processes}
\label{section:VecNIFFP}

In this section, we frame the model introduced in \Cref{section:HMFM} into a more general and flexible class of Bayesian nonparametric models. 
Indeed, it turns out that the study of the distributional properties of the proposed prior, presented in \Cref{subsection:distributional} can be carried out under quite a general choice for the prior distribution of the number of components $M$ and the unnormalized weights $S_{j,m}$.

Here, $M$ is supposed to be a positive integer-valued random variable with distribution $q_M$ on $\{1,2,3,\dots\}$. Then, given $M$, we still assume the unnormalized weights $S_{j,m}$ to be conditionally independent both within and between the groups. However, we now assume the components of $\bmS_j = \left(S_{j,1},\dots,S_{j,M}\right)$ to be, given $M$, i.i.d. from $H_j$, a probability distribution over $\R^+$. For the sake of simplicity, we assume all $H_j$ to belong to the same parametric family having density $h(\cdot\mid\gamma_j)$, i.e., $H_j({\rm d}s) = h(s\mid\gamma_j) \mydiff s$, where $\gamma_j\in\R^+$ are group-specific parameters.
Finally, conditionally to $M$, the common random $(\tau_1,\dots,\tau_M)$ atoms are i.i.d. from $P_0$, as said in \Cref{section:HMFM}.
We point out that, marginally, for each group $j$, Equations \eqref{eqn:mixture1} and \eqref{eqn:Pj_def} define a mixture model where the mixing distribution is the Normalized Independent Finite Point Process as defined \citet{argiento2022annals}.

The specification of the probability measures $q_M$, $H_j$, as $j=1, \ldots , d$, and $P_0$ induces a prior distribution for the mixture parameters $\left(M,\bmS,\bmtau\right)$, of which the model in \Cref{eqn:HMFM_definition} is a special case.
In this general construction, the joint law of the vector $\left( P_1,\dots,P_d\right)$ is 
called \textit{Vector of Normalized Independent Finite Point Process}  and denoted by
\begin{equation*}
	( P_1,\dots,P_d )\ \sim \ \operatorname{Vec-NIFPP}\left(q_M,H,P_0\right),
	\label{eqn:joint_Pjs}
\end{equation*}
where $H = H_1 \times \dots \times H_d$. Since the $P_j$'s are defined through normalization, see Equation \eqref{eqn:normalized_def}, we also define the joint law of the $\left(\mu_1,\dots,\mu_d\right)$, called Vector of Independent Finite Point Process (Vec-IFPP).
See \Cref{app:appA} for a detailed construction of the $\operatorname{Vec-NIFPP}$ via point processes.

\subsection{Distributional results}
\label{subsection:distributional}
In this section, we derive all theoretical properties for the latent variables $\theta_{j,i}$ modelled as follows
\begin{equation}
	\label{eqn:latent_model}
	\begin{split}
		\theta_{j,i}\mid P_j \iid P_j,
		\quad
		( P_1,\dots,P_d ) & \sim  \operatorname{Vec-NIFPP}\left(q_M,H,P_0\right),
	\end{split}
\end{equation}
where $j=1, \ldots , d$ and $i = 1, \ldots , n_j$. In addition, we assume conditional independence across groups, i.e.,
$\bmtheta_j,\bmtheta_l\mid P_j,P_l$ are independent for $j\neq l$.
The distributional results derived for the model in Equation \eqref{eqn:latent_model} play a pivotal role in devising efficient marginal and conditional algorithms to perform posterior inference. The proofs of the theoretical properties are deferred to \Cref{section:appB} of the Supplementary materials.

The almost sure discreteness of the $P_j$, coupled with their common supports, entails that the hierarchical sample $\left(\bmtheta_1,\dots,\bmtheta_d\right)$ is equivalently characterized by $(\bmtheta^{**},\rho)$, previously defined in \Cref{subsection:clustering}. The following theorem specifies the distribution of the pEPPF for the model under investigation.
\begin{theorem}[pEPPF]
	\label{thm:peppf}
	The probability to observe a sample  $\bm{\theta}= (\bm{\theta}_1, \ldots , \bm \theta_d)$ of size $n$ from Equation \eqref{eqn:latent_model} exhibiting $K$ distinct values $\left(\theta^{**}_1, \dots, \theta^{**}_K\right)$ with respective counts $\bmn_1,\dots,\bmn_d$ is given by the following pEPPF 
	\begin{equation*}
		\Pi_K^{(n)}(\bm{n}_1,\dots,\bm{n}_d) = 
		\bigintsss_{(\R^+)^d}{
			\prod_{j=1}^d\left\{
			\frac{u_j^{n_j-1}}{\Gamma(n_j)}
			\prod_{k=1}^K \kappa_j(u_j,n_{j,k})
			\right\}}
		\Psi(K,\bmu)
		{\rm d}u_1,\dots,{\rm d}u_d
		\label{eqn:peppf}
	\end{equation*}
	where \[
	\Psi(K,\bmu) = \sum_{m=0}^\infty \ q_M(m+K)\frac{(m+K)!}{m!}
	\prod_{j=1}^d\left( \psi_j\left(u_j\right) \right)^m,
	\]
	$\psi_j\left(u_j\right)$ is the Laplace transform of a random variable $S_j\sim H_j$ and $\kappa_j\left(u_j,n_{j,k}\right)$ is its derivative. Namely,
	\[
	\psi_j\left(u_j\right) = \E\left[e^{-u_j S_j}\right] = \int_{0}^\infty e^{-u_j s} H_j({\rm d}s), 
	\qquad
	\kappa_j\left(u_j,n_{j,k}\right) = \int_{0}^\infty e^{-u_j s} s^{n_{j,k}} H_j({\rm d}s).
	\]
\end{theorem}
The predictive distributions, i.e., the distribution of $\theta_{j,n_j+1}$ given $\left(\bmtheta_1,\dots,\bmtheta_d\right)$ for each possible group $j$, easily follow from \Cref{thm:peppf}. 
We describe them in detail in \Cref{section:HMFM_properties} for a specific choice of $q_M$ and $H$,
while we refer to \Cref{app:predictive} for the general formulation.

The following theorem provides functionals of $\left(\mu_1,\dots,\mu_d\right)$, which are useful for prior elicitation and to shade light on the dependence structure introduced by our prior.

\begin{theorem}[Mixed moments]
	\label{thm:moments}
	Let $(P_1,\dots,P_d) \sim \operatorname{Vec-NIFPP}\left(q_M,H,P_0\right)$ be a vector of normalized random probability measures defined through normalization as in Equation \eqref{eqn:normalized_def}. 
	For any measurable set $A$ and for any $j\in \{1,\dots,d\}$, $\E\left[P_j(A)\right] = P_0(A)$. Additionally, 
	\begin{itemize}
		\item[(i)] for any measurable sets $A,B$ and for any $j,l \in \{1,\dots,d\}$,
		\begin{align}
			\label{eqn:mixedmom1}
			\E\left[P_j(A)P_l(B)\right]  
			= \mathbb{P}\left(K_{(1,1)} = 1\right) 
			\left( P_0(A\cap B) - P_0(A)P_0(B)\right);
		\end{align}
		\item[(ii)] for any measurable set $A$ and for any $j,l \in \{1,\dots,d\}$,
		\begin{align}
			\label{eqn:mixedmom2}
			\E\left[P_j(A)^{n_j}P_l(A)^{n_l}\right] = \E\left[P_0(A)^{K_{(n_j,n_l)}}\right] = 
			\sum_{k=1}^{n_j+n_l} P_0(A)^k\mathbb{P}\left(K_{(n_j,n_l)} = k\right),
		\end{align}
	\end{itemize}
	where $K_{(n_j,n_l)}$ is the global number of clusters across two groups with size $n_j$ and~$n_l$.
\end{theorem}

The prior probabilities for the number of clusters can be written in terms of the pEPPF, e.g.,
$\mathbb{P}\left(K_{(1,1)} = 1\right) = \Pi_1^{(2)}\left((1),(1)\right)$.
Furthermore, both Equation \eqref{eqn:mixedmom1} and \eqref{eqn:mixedmom2} can be extended to the case of more than two groups. 
As a byproduct of \Cref{thm:moments} we obtain a closed form expression for pairwise correlation between the components of $(P_1, \ldots , P_d)$ evaluated on specific sets. 
Let $A$ be a measurable set, then, for any $j,l \in \{1,\dots,d\}$:
\begin{equation}
	\label{eqn:corr}
	\begin{split}
		\text{corr} \left( P_j(A),P_l(A)\right)  \ = \ 
		\frac{\mathbb{P}\left(K_{(1,1)} = 1\right)}
		{\sqrt{\mathbb{P}\left(K_{j,(2)} = 1\right)
				\mathbb{P}\left(K_{l,(2)} = 1\right)} },
	\end{split}
\end{equation}
where $\mathbb{P}\left(K_{j,(2)} = 1\right)$ is the prior probability of having a single local cluster out of two observations in group $j$. 
The expression in Equation \eqref{eqn:corr} does not depend on the choice of the set $A$. Thus, it may be considered an overall measure of dependence between the two random probability measures. As a consequence, Equation \eqref{eqn:corr} will be used in  \Cref{section:HMFM_properties} to clarify
the sharing of information mechanism and select the hyperparameters in the noteworthy example of vectors of finite Dirichlet processes.
We also derive a more general formulation of Equation \eqref{eqn:corr}, when the random probability measures are evaluated on different sets, see \Cref{app:correlation}. Furthermore, we present in \Cref{app:coskewness} the analytical formula for the coskewness between $P_j(A)$ and $P_l(A)$.

We now aim at giving a posterior characterization for a vector $\left(P_1,\dots,P_d\right)$ distributed as in Equation \eqref{eqn:joint_Pjs}. Since $\left(P_1,\dots,P_d\right)$ is obtained via normalization of $\left(\mu_1,\dots,\mu_d\right)$, it is sufficient to provide a posterior characterization for the latter vector. 
In order to do this, we follow the same approach of
\citet{jlp2009}, \citet{camerlenghi2019distribution} and \citet{argiento2022annals}.
Thus, we introduce a vector of auxiliary variables $\bmU_n = \left(U_1,\dots,U_d\right)$ such that $U_j\mid T_j\ind \operatorname{gamma}\left(n_j,T_j\right)$, where $T_j = \mu_j(\X)$. 
This is possible since the marginal distribution of $\bmU_n$ does exist, see Section \ref{proof:posterior}.
Hence, conditionally to $\bmU_n$ and to $\left(\bm \theta_1,\dots, \bm \theta_d\right)$, $\left(\mu_1,\dots,\mu_d\right)$ is a superposition of two independent processes, one driving the non-allocated components and the other one driving the allocated components. 
\begin{theorem}[Posterior representation]\label{thm:posterior}
Let $\left(\bm \theta_1,\dots,\bm\theta_d\right)$  be a sample from the statistical model in Equation \eqref{eqn:latent_model}.
Then, the posterior distribution of $\left(\mu_1,\dots,\mu_d\right)$ 
is characterized as  the superposition of two independent 
processes on $(\R^+)^d \times \Theta$:
\[
\left(\mu_1,\dots,\mu_d\right) \mid \bm{\theta}_1, \ldots , \bm \theta_d , \bm{U}_n \myequal{d} \left(\mu^{(a)}_1,\dots,\mu^{(a)}_d\right) + \left(\mu^{(na)}_1,\dots,\mu^{(na)}_d\right),
\text{  \it where:}\]

\begin{itemize}
	\item[(i)] the process of allocated components $\left(\mu^{(a)}_1,\dots,\mu^{(a)}_d\right)$ equals
	\[
	\mu^{(a)}_j = \sum_{k=1}^K S^{(a)}_{j,k}\delta_{\theta^{**}_k}, \text{ as } j=1, \ldots , d,
	\]
	where the random variables $S^{(a)}_{j,k}$, for $j  \in \{ 1,\dots,d \}$ and $k\in\{1,\dots,K\}$, are independent with densities on $\R^+$ given by
	\[
	f_{S^{(a)}_{j,k}}(s) \propto e^{-u_j s}s^{n_{j,k}}h_j(s); 
	\]
	\item[(ii)] the process of non-allocated components $\left(\mu^{(na)}_1,\dots,\mu^{(na)}_d\right)$ is a $\operatorname{Vec-IFPP}\left(q^*_M,H^*,P_0\right)$
	, with
	$H^*({\rm d} \bms) =  H^*_1({\rm d}s_1) \times \dots \times H^*_d({\rm d}s_d)$ and
	\begin{equation*}
		H^*_j({\rm d} s) \propto e^{-u_j s}h_j(s){\rm d} s,  \qquad \quad
		q^*_M(m)\propto q_M(m+K)\frac{(m+K)!}{m!}\prod_{j=1}^d\left( \psi_j\left(u_j\right) \right)^m.
	\end{equation*}
\end{itemize}
\end{theorem}
Note that for the process of non-allocated component $q^*_M(0)>0$ even if $q_M(0)=0$. Hence, it is possible to have zero non-allocated components.

\section{Properties of the HMFM model and fitting details}
\label{section:HMFM_properties}

In the following, we specialize the theorems presented in \Cref{subsection:distributional} to the case of the $\operatorname{Vec-FDP}$ introduced in Section \ref{section:HMFM}. 
Note that, for $S_{j,m} \sim \operatorname{Gamma}\left(\gamma_j,1\right)$, the expressions of $\kappa_j\left(u_j,n_{j,k}\right)$ and $\psi_j\left(u_j\right)$ defined in \Cref{thm:peppf} are given by
\begin{equation}
\label{eqn:kappa_psi_def_VecFDP}
\psi_j\left(u_j\right) = \frac{1}{\left(1+u_j\right)^{\gamma_j}},
\quad
\kappa_j\left(u_j,n_{j,k}\right) = \frac{1}{\left(1+u_j\right)^{n_{j,k}+\gamma_j}} \frac{\Gamma\left(n_{j,k} + \gamma_j\right)}{\Gamma\left(\gamma_j\right)}.
\end{equation}
Moreover, if $q_M$ is a 1-shifted Poisson distribution, $\Psi(K,\bmu)$ defined in \Cref{thm:peppf} admits the following tractable analytical form,
\begin{equation}
\Psi(K,\bmu) = \Lambda^{K-1} 
\left(K+\Lambda \bmpsi(\bmu)\right) 
e^{-\Lambda\left(1-\bmpsi(\bmu)\right)},
\label{eqn:Psi_def_VecFDP}
\end{equation}
where $\bmpsi(\bmu) = \prod_{j=1}^d \psi_j\left(u_j\right)$.
Equation \eqref{eqn:Psi_def_VecFDP} is derived in \Cref{proof:Psi} of the Supplementary materials.

\subsection{Marginal and posterior representation}
\label{subsection:theorems_for_VecFDP}
Equations \eqref{eqn:kappa_psi_def_VecFDP} and \eqref{eqn:Psi_def_VecFDP}, plugged in \Cref{thm:peppf}, allow to derive the pEPPF induced by the $\operatorname{Vec-FDP}$ prior, that is
\begin{equation}
\Pi_K^{(n)}(\bm{n}_1,\dots,\bm{n}_d) = 
\text{V}(K;\bmgamma,\Lambda) \prod_{j=1}^d\prod_{k=1}^K
\frac{\Gamma\left(n_{j,k} + \gamma_j\right)}{\Gamma\left(\gamma_j\right)},
\label{eqn:peppf_FDP}
\end{equation}
where $\text{V}(K;\bmgamma,\Lambda)$ equals
\begin{equation*} 
\label{eqn:integralV_FDP}
\begin{split}      
	\text{V}(K;\bmgamma,\Lambda) =
	\bigintssss_{(\R^+)^d}
	\Psi(K,\bmu)
	\prod_{j=1}^d \frac{u_j^{n_j-1}}{\Gamma\left(n_j\right)} \frac{1}{\left(1 + u_j\right)^{n_j+K \gamma_j}}
	{\rm d}u_1\dots {\rm d}u_d.
\end{split}
\end{equation*}
Furthermore, this setting is particularly convenient for computing the prior distribution of the global number of clusters.
The following theorem states the result for two groups having cardinalities $n_1,n_2$, respectively. 
\begin{theorem}[Prior distribution of global number of clusters]\label{thm:priorK}
Consider $d=2$ groups of observations with size $n_1$ and $n_2$, respectively. Then, the prior distribution for the global number of clusters $K_{(n_1,n_2)}$ is
\begin{equation*}
	\label{eqn:Kprior_d2}
	\begin{split}
		&\P\left(K_{(n_1, n_2)} = K \right) \\
		& \qquad =
		V(K;\bmgamma,\Lambda)
		\sum_{r_1=0}^K  \sum_{r_2=0}^{K-r_1} 
		\binom{K-r_1}{r_2}\frac{(K-r_2)!}{r_1!}
		\prod_{j=1}^2  \left| C(n_j  , K-r_j ; -\gamma_j ) \right|
	\end{split}
\end{equation*}
where for any non-negative integers $n\geq 0$ and $0\leq K \leq n$,  $C(n, K; -\gamma_j )$ denotes the central generalized factorial coefficients.
See \citep{chara2002}.
\end{theorem}
Being able to derive an explicit expression for the prior distribution of the number of clusters allows us to evaluate the mixed moment formula given in Equation \eqref{eqn:mixedmom2}. Moreover, for a Vec-FDP,
we leverage on the explicit formulation for the pEPPF given in Equation \eqref{eqn:peppf_FDP} to further develop the computation of 
the correlation between $P_j(A)$ and $P_l(A)$ given in Equation \eqref{eqn:corr}. Then, the latter is equal to
\begin{equation}
\label{eqn:corr_Dir}
\text{corr}\left(P_j(A),P_l(A)\right) \ = \ 
\frac{1 - e^{-\Lambda}}
{\Lambda \left(\gamma_j+1\right)\left(\gamma_l+1\right)I\left(\gamma_j,\Lambda\right)I\left(\gamma_l,\Lambda\right)},
\end{equation}
where $I\left(\gamma_j,\Lambda\right) = \int_0^1\left(1+\Lambda x\right)
e^{-\Lambda(1-x)}(1-x^{1/\gamma_j}) {\rm d} x$. 
The limits of Equation \eqref{eqn:corr_Dir} when both $\gamma_j$ and $\gamma_l$ goes to $0$ and $+\infty$ equal to
\begin{equation}
\label{eqn:corr_limits}
\lim_{\gamma_j,\gamma_l \rightarrow 0} 
\text{corr}\left(P_j(A),P_l(A)\right) \ = \ \frac{1 - e^{-\Lambda}}{\Lambda},\quad
\lim_{\gamma_j,\gamma_l \rightarrow \infty} 
\text{corr}\left(P_j(A),P_l(A)\right) \ = \ 1.
\end{equation}
These limits are interesting because we see that, given $\Lambda$, decreasing $\gamma_j$ and $\gamma_l$, the correlation does not go to $0$ but reaches a lower bound that depends on $\Lambda$, which, in turn, goes to $0$ if $\Lambda$ increases.  
On the other hand, increasing values of $\gamma_j$ and $\gamma_l$ lead correlation equal to $1$, regardless of the value of $\Lambda$. See the left panel in \Cref{fig:cor_main}.
We refer to \Cref{app:corr_limiting} of the Supplementary materials for the proofs of Equations \eqref{eqn:corr_Dir} and \eqref{eqn:corr_limits}. Finally, \Cref{cor:posterior_FDP} provides the particular case of the posterior representation given in \Cref{thm:posterior}. 

\begin{figure}
\centering
\centering
\includegraphics[width=0.49\textwidth]{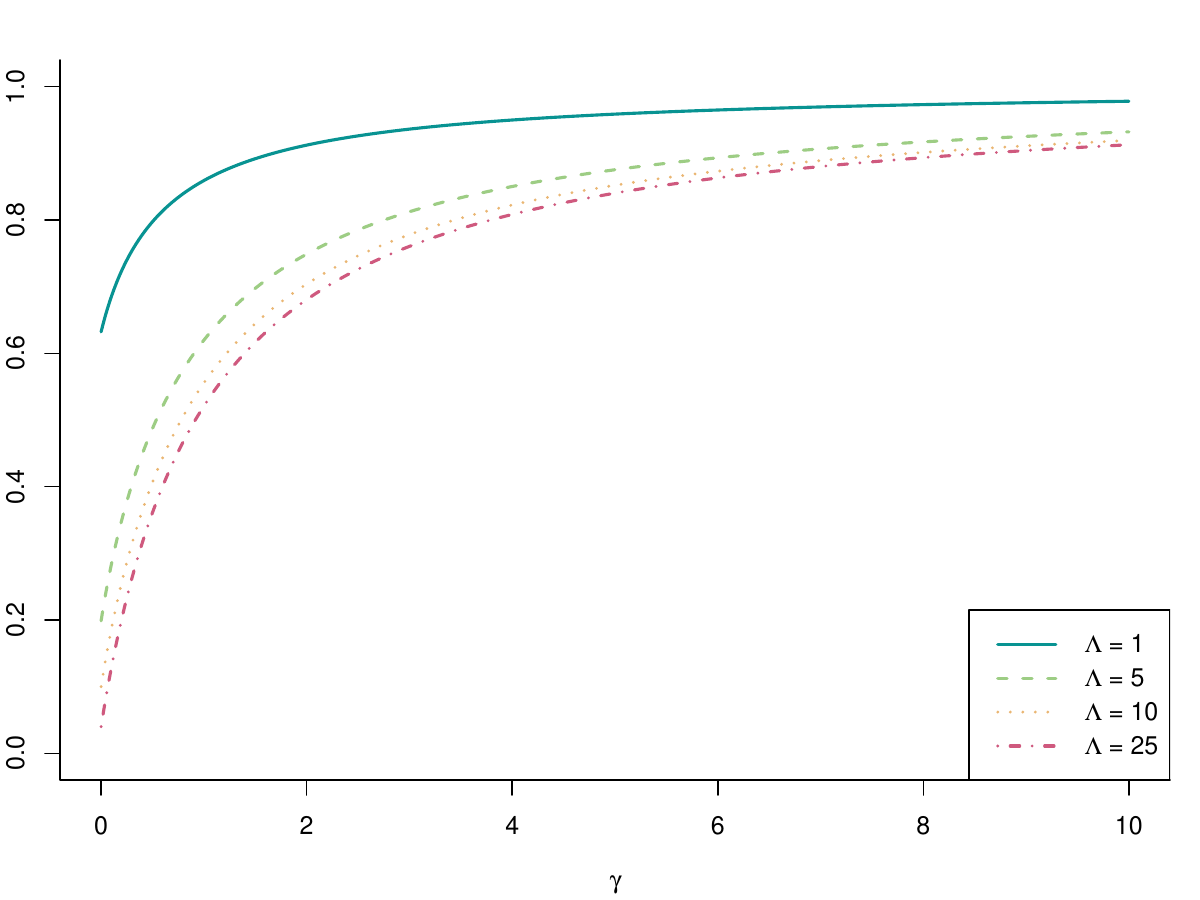}
\hfill
\includegraphics[width=0.49\textwidth]{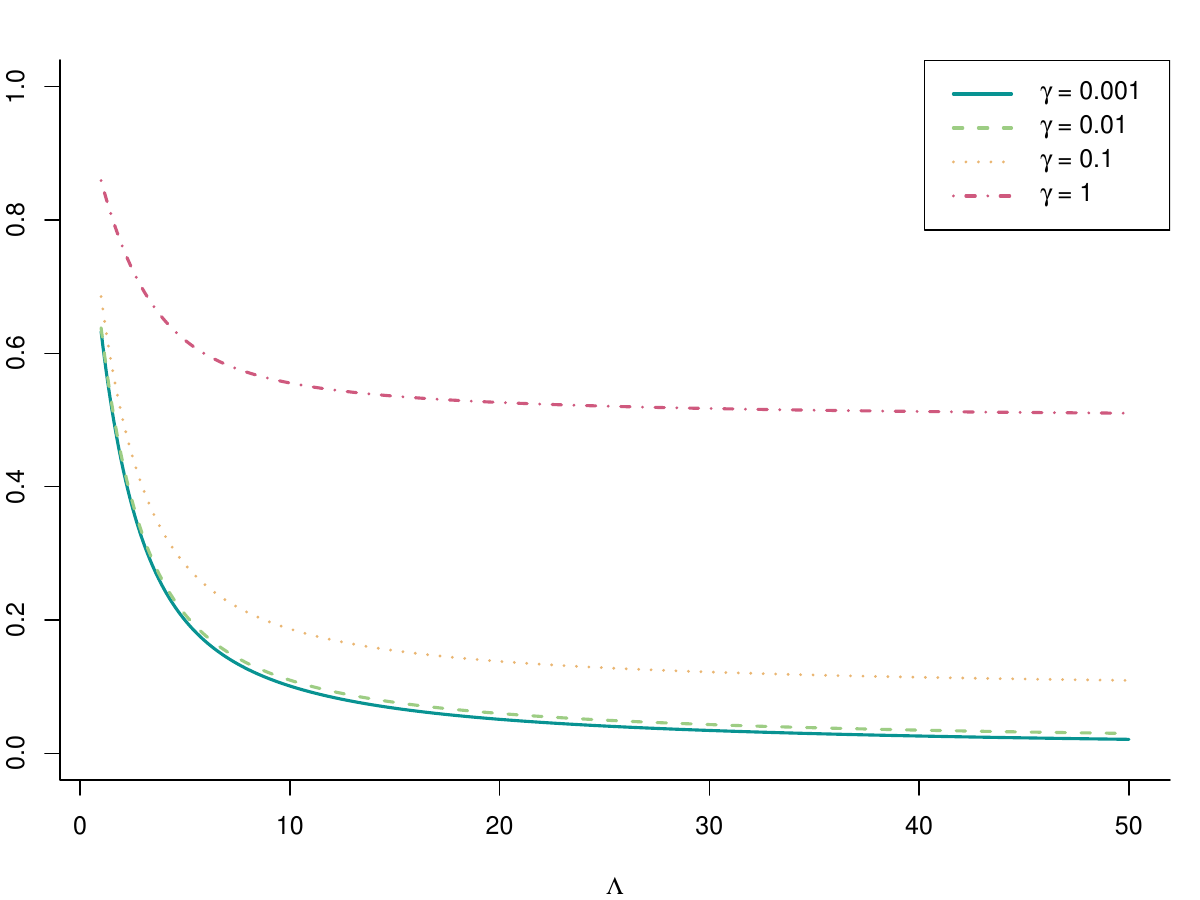}
\caption{Left panel: correlation for $\gamma_1 = \gamma_2$ varying over a grid of values. Each curve is obtained by fixing $\Lambda$ to the values reported in the legend. 
	Right panel: correlation for $\Lambda$ varying over a grid of values. Each curve is obtained by fixing $\gamma_1 = \gamma_2$ to the values reported in the legend. }
\label{fig:cor_main}
\end{figure}

\begin{corollary}
\label{cor:posterior_FDP}
If $(P_1,\dots,P_d) \sim \operatorname{Vec-FDP}\left(\Lambda,\bmgamma,P_0\right)$ the representation of \Cref{thm:posterior} holds true with the following specifications:  
\begin{itemize}
	\item[(i)] 
	the random variables 
	$S^{(a)}_{j,k} \ind \operatorname{Gamma}(n_{j,k} + \gamma_j, u_j + 1)$, for each $j \in \{1,\dots,d\}$ and $k \in \{ 1,\dots,K\}$;
	\item[(ii)] 
	the random variables 
	$S^{(na)}_{j,1},\dots,S^{(na)}_{j,M^{*}} \mid M^{*} \iid  \operatorname{Gamma}(\gamma_j, u_j + 1)$, for $j \in\{ 1,\dots,d\}$, 
	\item[(iii)] 
	$q^*_{M}( \, \cdot \,)$ is a distribution over non-negative integers $\{0,1,2,\dots\}$ given by a mixture of Poisson distributions, namely, $q^*_{M}(m) $ equals
	\begin{equation*}
		\label{eqn:qstar}
		\begin{split}  
			(1- w_K (\bm{u}))\operatorname{Pois}_1\left(\Lambda \prod_{j=1}^d \psi_j\left(u_j\right)\right) 
			+
			w_K (\bm{u}) \operatorname{Pois}\left(\Lambda \prod_{j=1}^d \psi_j\left(u_j\right)\right)
		\end{split}
	\end{equation*}
	where we have set $w_K (\bm{u}):= K/(\Lambda \prod_{j=1}^d \psi_j\left(u_j\right) + K )$.
\end{itemize}
\end{corollary}

\subsection{Predictive distribution and franchise metaphor}
\label{subsection:CRFP}
Further intuitions on the cluster mechanism described in \Cref{subsection:clustering} are available when considering the predictive distributions. 
Consider a realization $\left(\bmtheta_1,\dots,\bmtheta_d\right)$ with $K$ distinct values $\left(\theta^{**}_1,\dots,\theta^{**}_K\right)$ and a partition $\rho = \left\{C_1,\dots,C_K\right\}$ with counts $\left(\bmn_1,\dots,\bmn_d\right)$ satisfying the constraints in Equation \eqref{eqn:cluster_vincoli}.
Following the approach of \citet{jlp2009}, \citet{favaroteh2013} and \citet{argiento2022annals} we work conditionally to $\bmU_n = \bmu$. Then, for each group, say $j$, we have 
\begin{equation}
\label{eqn:predittiva_FDP}
\begin{split}     
	&\mathbb{P}\left(\theta_{j,n_j+1} \in \cdot \mid \bmtheta_1,\dots,\bmtheta_d, \bmu, \bmgamma, \Lambda\right) \\
	& \qquad \qquad  \propto 
	\sum_{k=1}^K \left(n_{j,k} + \gamma_j\right)\delta_{\theta^{**}_k}(\cdot) 
	+  
	\bmpsi(\bmu)
	\gamma_j\Lambda
	\frac{K+1+\Lambda\bmpsi(\bmu)}{K+\Lambda\bmpsi(\bmu)}
	P_0(\cdot) .
\end{split}
\end{equation}
Such a predictive distribution can be interpreted in terms of a restaurant franchise metaphor. Consider a franchise of $d$ Chinese restaurants each with possibly infinitely many tables. Here  $\theta_{j,i}$ represents the dish served to  customer $i$ in restaurant $j$, and each $\theta^{**}_k$ represents a dish. 
All customers sitting at the same table must eat the same dish.
The same dish can not be served at different tables in the same restaurant, but it can be served across different restaurants.  According to the predictive law, 
the first customer entering the first restaurant sits at the first table eating dish $\theta_{1,1}=\theta^{**}_1$, which is drawn from $P_0$. At the same time, an empty table serving dish $\theta^{**}_1$ must be prepared in all the other restaurants:  this step corresponds to the first cluster allocation, i.e.,  $C_1$. Then, the second customer of the first restaurant arrives and can either: (i) sit at the same table as the first customer, with probability proportional to $1+\gamma_1$ or (ii) sit at a new table with probability proportional to $\bmpsi(\bmu)\gamma_1\Lambda\frac{2+\Lambda\bmpsi(\bmu)}{1+\Lambda\bmpsi(\bmu)}$. In the latter case, the customer chooses a new dish $\theta^{**}_2$, drawn from $P_0$, and the number of clusters $K$ is increased by $1$; moreover, an empty table serving dish $\theta^{**}_2$ must be prepared in all the other restaurants of the franchise. Then, the process evolves according to Equation \eqref{eqn:predittiva_FDP}.
\Cref{fig:tavoli} in the Supplementary materials displays a graphical representation of the process.

Interestingly, our model is more parsimonious than the HDP by \citet{teh2006hierarchical} in sharing information across restaurants. 
While the HDP relies on the popularity of a dish throughout the entire franchise to influence a new customer's choice, 
in our model, such probability depends on the sample only through the dish's popularity within the specific restaurant the customer enters.
This distinctive feature proves to be appealing as it mitigates the excessive borrowing of information across groups that is induced by hierarchical processes. Subsequent numerical experiments highlight the advantage of our model by showing that the HDP can lead to misleading results in posterior inference.
An additional advantage of our model with respect to the HDP is that
different tables within the same restaurant cannot serve the same dish. This simplifies the local clustering structure, and it improves the computational efficiency in posterior sampling, as discussed in \Cref{subsection:computational}.

\subsection{Hyperpriors} 
\label{subsection:hyperpriorsHMFM}
We consider the following hyperpriors for the process hyperparameters $(\Lambda,\bmgamma)$, 
\begin{equation}
\pi(\Lambda,\bmgamma) = \pi(\bmgamma\mid\Lambda)\pi(\Lambda) = \prod_{j=1}^d
\operatorname{Gamma}(a_\gamma,\Lambda b_\gamma)\times\operatorname{Gamma}(a_\Lambda,b_\Lambda).
\label{eqn:prior_Lambda_gamma}
\end{equation}
The prior distribution in Equation \eqref{eqn:prior_Lambda_gamma} 
extends, to our setting, the prior choice introduced by \citet{here2inf} 
to encourage sparsity in the mixture, whose advantages have been studied both theoretically  \citep{Rousseau2011,Rousseau2015} and empirically \citep{sparsemix2016,mixofmix2017}. 
Furthermore, this prior formulation assumes the $\gamma_j$s to be conditionally independent given $\Lambda$, so tuning the sharing of information between groups,  
see also \Cref{fig:cor_main}. 
In particular, note that $\Lambda\mid \bmgamma$ is still $\operatorname{Gamma}$ distributed, i.e.,
$
\Lambda\mid \bmgamma \sim \operatorname{Gamma}\left(a_\Lambda + d a_\gamma, b_\Lambda + b_\gamma\sum_{j=1}^d \gamma_j\right),
$
which yields tractable posterior inference. 
We provide practical guidelines for setting the values of hyperparameters $(a_\gamma,b_\gamma,a_\Lambda,b_\Lambda)$ 
by exploiting the equivalent sample principle \citep{diaconis1979}. To this end, we  design a suitable reparametrization of the prior based on three quantities,
$\Lambda_0$, $\text{V}_\Lambda$ and $\gamma_0$. The first two quantities represent the prior expected value and variance of $\Lambda$, respectively, while $\gamma_0$ represent a common prior guess on $\gamma_j$. 
Thus, the new specification relies on hyperparameters that are easy to interpret and allows to elicit prior knowledge, when available.
Further details are presented in \Cref{app:hyperparam} of the Supplementary materials.

\subsection{Computational methods and simulation}
\label{subsection:computational}
As usual in hierarchical modelling, we 
introduce latent allocation vectors 
$\bmc_j = \left(c_{j,1},\dots,c_{j,n_j}\right)$ 
whose element $c_{j,i}\in\{1,\dots,M\}$ denotes to which component observation $y_{j,i}$ is assigned, for each $j = 1,\dots,d$. 
Setting $\theta_{j,i} = \tau_{c_{j,i}}$, we are able to link the mixture parameters and the observation-specific parameters. 
We suggest two MCMC strategies to carry out posterior inference for mixture modelling. The first one is a conditional algorithm
that provides full Bayesian inference on both the mixing parameters $(P_1,\dots,P_d)$ and the clustering structure $\rho$. 
Namely, we draw a  sample of the vector of random probability measures $\left(P_1,\dots,P_d\right)$ from its posterior distribution given in \Cref{cor:posterior_FDP} by sampling  from the joint posterior distribution of 
$\left(\bmS_1,\dots,\bmS_d,\bmtau,\bmc_1,\dots,\bmc_d,M \right)$.
Note that the global number of clusters $K$ is automatically deduced from the cluster allocation vectors $\left(\bmc_1,\dots,\bmc_d\right)$. 
To do so, we resort to auxiliary variables $\bmU_n$ and the hyperparameters $\left(\Lambda,\bmgamma\right)$. 
For the sake of brevity, we denote $\bmS = (\bmS_1,\dots,\bmS_d)$ and $\bmc = (\bmc_1, \dots, \bmc_d)$. We adopt a blocked Gibbs sampling strategy. In particular, let 
$\Delta = \left(\bmS,\bmtau,\bmc,M,\bmU,\bmgamma,\Lambda\right)$ be a vector collecting all the parameters and let $\bmy$ be the collection of all variables $y_{j,i}$, for each group $j$ and for each individual $i$. We partition $\Delta$ in two blocks $\Delta_1 = \left(\bmS,\bmtau,\bmc,\Lambda\right)$ and $\Delta_2 = \left(\bmU,M,\bmgamma\right)$.
The sampling scheme proceeds by iterating two steps: (i) sampling $\Delta_1$ conditionally to $\Delta_2$ and $\bmy$;  (ii) sampling $\Delta_2$ conditionally to $\Delta_1$ and $\bmy$. We refer to \Cref{app:MCMC_conditional} for a step-by-step description of the algorithm, where all full conditional distributions are detailed.

The second algorithm is a marginal sampler that simplifies the computation by integrating the mixture parameters while providing inference on the sole clustering structure. 
The algorithm is derived from the predictive distributions detailed in \Cref{subsection:CRFP}, and the full conditional distribution of $\bmU_n$, which is obtained as a byproduct of \Cref{thm:posterior}, see \Cref{app:full_U}. A detailed description of the marginal algorithm is given in \Cref{app:MCMC_marginal}.

In \Cref{app:simstudy} we present an extensive simulation study comparing the HMFM in Equation \eqref{eqn:HMFM_definition} with the HDP \citep{teh2006hierarchical} mixture model as well as with the MFM model, that is assumed independently for each group and fitted using the algorithm by \citet{argiento2022annals}.
Regarding the HMFM, we investigate the performance of both the conditional (HMFM-cond) and the marginal (HMFM-marg) sampler. 
In summary, the simulation study consists of three experiments: the first one clearly shows the advantage of the joint modelling approach against independent, group-specific analyses. 
The second experiment is an illustrative example that highlights the limitations of the HDP  in situations where borrowing information from other groups can lead to misleading conclusions. Rather, this issue is mitigated by the HMFM that borrows 
less information across the groups relative to the HDP. The third experiment generates data from $d=15$ groups and further evidences how the HMFM outperforms the HDP.
The results are given in \Cref{app:simstudy}.

Finally, we show that our model outperforms the HDP in terms of computational efficiency. Indeed, given a fixed number of clusters and groups, a single iteration of the direct allocation scheme for the HDP takes $O(n^2)$ time, limiting, eventually, its feasibility in practice.
In contrast, the HMFM offers comparable results at a reduced computational burden as it achieves linear scaling, taking $O(n(M+K))$ and $O(n(K+1))$ time for one iteration of the conditional and marginal sampler, respectively.
A theoretical derivation of the computational times and their empirical validation are deferred to \Cref{app:MCMC} and \Cref{app:computational_time}, respectively.
The R code implementing both MCMCM algorithms is available at \url{https://anonymous.4open.science/r/GDFMM-F23C/README.md}, along with the simulation study.

\section{Analysis of shot put data} 
\label{section:application}

Shot put is a track and field event in which athletes throw a heavy spherical ball, known as the shot, as far as possible. Our dataset comprises measurements, specifically the throw lengths or marks, recorded during professional shot put competitions from 1996 to 2016, for a total of $ 35,637$ measurements on $403$ athletes.
Each athlete's record includes the mark achieved, competition details, and personal information, namely, age, gender, and whether the event took place indoors or outdoors. The analyzed data are publicly available (\url{www.tilastopaja.eu}).
Our objective is to model the seasonal performance for each shot putter, interpreted as the mean and variance of his/her seasonal marks. In particular, the season number assigned to each observation corresponds to the number of seasons the athlete has participated in, excluding seasons where he/she did not compete. For example, season 1 represents the athlete's first active season. This grouping of observations into seasons reflects the athletes' years of experience.
\Cref{fig:application_traiettorie} in the Supplementary materials visually illustrates the performance evolution throughout the career of four randomly selected shot putters from the dataset. Each athlete has different participation in competitions, and the length and trajectory of their performance careers vary. 
While performance is expected to vary over the athlete's career, the figure evidences that the performance remains relatively consistent within each season. 
We characterized the seasonal performances as arising primarily from random fluctuations around a mean value. The values of this mean and the associated variability are unknown and are inferred from the data. Although it is a simplified representation, this captures the essential characteristics of athletes' careers.

In a previous study, \citet{montagnaSMA} employed a GARCH model to account for the volatility clustering of athletes' results over time. Rather, in this work, we frame the data into a hierarchical structure 
where each season represents a different group. Hence, we assume the HMFM model described in \Cref{section:HMFM} for analyzing the athletes' performance; the proposed model allows us to capture the variability among different seasons and clustering the performances both within the seasons and across them.

Let $n_{j}$ be the number of athletes competing in season $j$, with $j = 1,\ldots,d$. The longest career consists of $15$ seasons, which is then the total number of groups $d=15$. 
Each active athlete $i$ in season $j$, with $i = 1,\ldots,n_{j}$, takes part in $N_{j,i}$ events. At each event, indexed by $h = 1,\ldots,N_{j,i}$, the athlete's mark $y_{j,i,h}$ is measured. 
Moreover, $r$ event-specific covariates are available,
$\bmx_{j,i,h} \in \R^r$, and collected in 
the design matrices $X_{j,i}\in \R^{N_{j,i}\times r}$.

Assuming that observations are noisy measurements of an underlying athlete-specific function, the model we employ for these data is
$
y_{j,i,h} = \mu_{j,i} + X_{j,i}\bmbeta_j + \varepsilon_{j,i,h}, 
$
with $\varepsilon_{j,i,h} \iid \text{N}\left(0, \sigma^2_{j,i} \right)$, where 
$\mu_{j,i}$ is a season-specific random intercept,  $\bmbeta_j$ is a $r$-dimensional vector of regression parameters, shared among all the athletes in season $j$, and $\sigma^2_{j,i}$ denotes the error variance.  
Therefore, within each season $j$, the
athlete's observations $\bmy_{j,i} = (y_{j,i,1},\dots,y_{j,i,N_{j,i}})$ are distributed as 
\begin{equation}
\label{eqn:shotput_likelihood}
\bmy_{j,i} \mid \mu_{j,i}, \sigma^2_{j,i}, \bmbeta_j, X_{i,j} \ind \text{N}_{N_{j,i}}\left(\mu_{j,i}\mathbf{1}_{N_{j,i}} + X_{j,i} \bmbeta_j, \sigma^2_{j,i}\mathbf{I}_{N_{j,i}}\right),
\end{equation}
where $\text{N}_{N_{j,i}}$  denotes the $N_{j,i}$-dimensional normal distribution, $\mathbf{1}_{N_{j,i}}$ is a vector of length $N_{j,i}$ with all entries equal to $1$ and $\mathbf{I}_{N_{j,i}}$ is the identity matrix of size $N_{j,i}$. To ensure  identifiability, observations $\bmy_{j,i}$ have been centered within each season, i.e., $\sum_{i=1}^{n_j}\sum_{h=1}^{N_{j,i}} y_{j,i,h} = 0$ for each $j$. 

Letting $\theta_{j,i} = \left(\mu_{j,i}, \sigma^2_{j,i}\right)$, we place a $\operatorname{Vec-FDP}$ prior for $\theta_{j,i}$ so that 
a clustering of athletes' performances both within and across different seasons is achieved. 
We assume a multivariate normal prior distribution for the regression coefficients, whose prior mean is denoted by $\bmbeta_0$ and covariance matrix $\Sigma_0$.
We define $\bmy$ as the collection of all observations across seasons $j$ and athletes $i$.
Based on evidence from a previous analysis \citep{montagnaSMA} and for ease of interpretation, we use only gender as a covariate in our analysis. In particular, we use male athletes as reference baseline and set $\bmbeta_0 = -2\mathbf{1}_d$ and $\Sigma_0 = \mathbf{I}_d$ expecting male athletes to throw longer than females. 


We set the base probability measure $P_0$ for $\theta=\left(\mu,\sigma^2\right)$ to be a $\operatorname{Normal-InvGamma}$, exploiting the conjugacy with the likelihood in Equation \eqref{eqn:shotput_likelihood}. In particular, the $\operatorname{Normal-InvGamma}$ distribution is parametrized as is \citet{hoff}, i.e.,
\[
\left(\mu,\sigma^2\right) \ \sim \ 
\operatorname{InvGamma}(\mu_0,k_0,\nu_0,\sigma^2_0) = 
\text{N}\left(\mu_0,\frac{\sigma^2}{k_0}\right)\times
\operatorname{Gamma}\left(\frac{\nu_0}{2},\frac{\nu_0}{2}\sigma^2_0\right).
\]
Following the approach of \citet{Richardson1997} and \citet{reinforcementBNP2007}, we set 
$\mu_0 = 0$ and $k_0 = \frac{1}{\text{range}(\bmy)^2}$.
Then, we set $\nu_0 = 4$ and and $\sigma^2_0 = 10$ to have a vague $\operatorname{InvGamma}$ with infinite variance. 
For the process hyperparameters $\Lambda$ and $\bmgamma$, we set the hyperprior in Equation \eqref{eqn:prior_Lambda_gamma}.
To achieve sparsity in the mixture, we follow the approach in \Cref{app:hyperparam}, and set $\Lambda_0 = 25$, $\text{V}_\Lambda = 3$ and $\gamma_0 = 1/\sum_{j=1}^d n_j = 0.00027$, leading to $a_\gamma = 13.89$, $b_\gamma=2007.78$, $a_\Lambda=208.33$, $b_\Lambda=8.33$. The complete formulation of the hierarchical model can be found in \Cref{app:shotput} of the Supplementary materials.
The burn-in period has been set equal to $50,000$, then $200,000$ additional iterations were run with a thinning of $10$. The initial partition has been set using the k-means on the pooled dataset with 20 centres. Posterior analysis is not sensible for such a choice. 

Figure \ref{fig:betas} shows the posterior 95\% credible intervals of 
the regression coefficients. The posterior distribution is concentrated on negative values, meaning that the athletes' marks are, on average, higher for males than females.  
Also, the effect of gender on athletes' performance is significantly different across seasons, e.g., it is more evident in the first years and it reduces over career years, with the exception of the final seasons.

\begin{figure}[ht]
\centering
\includegraphics[width=0.66\textwidth]{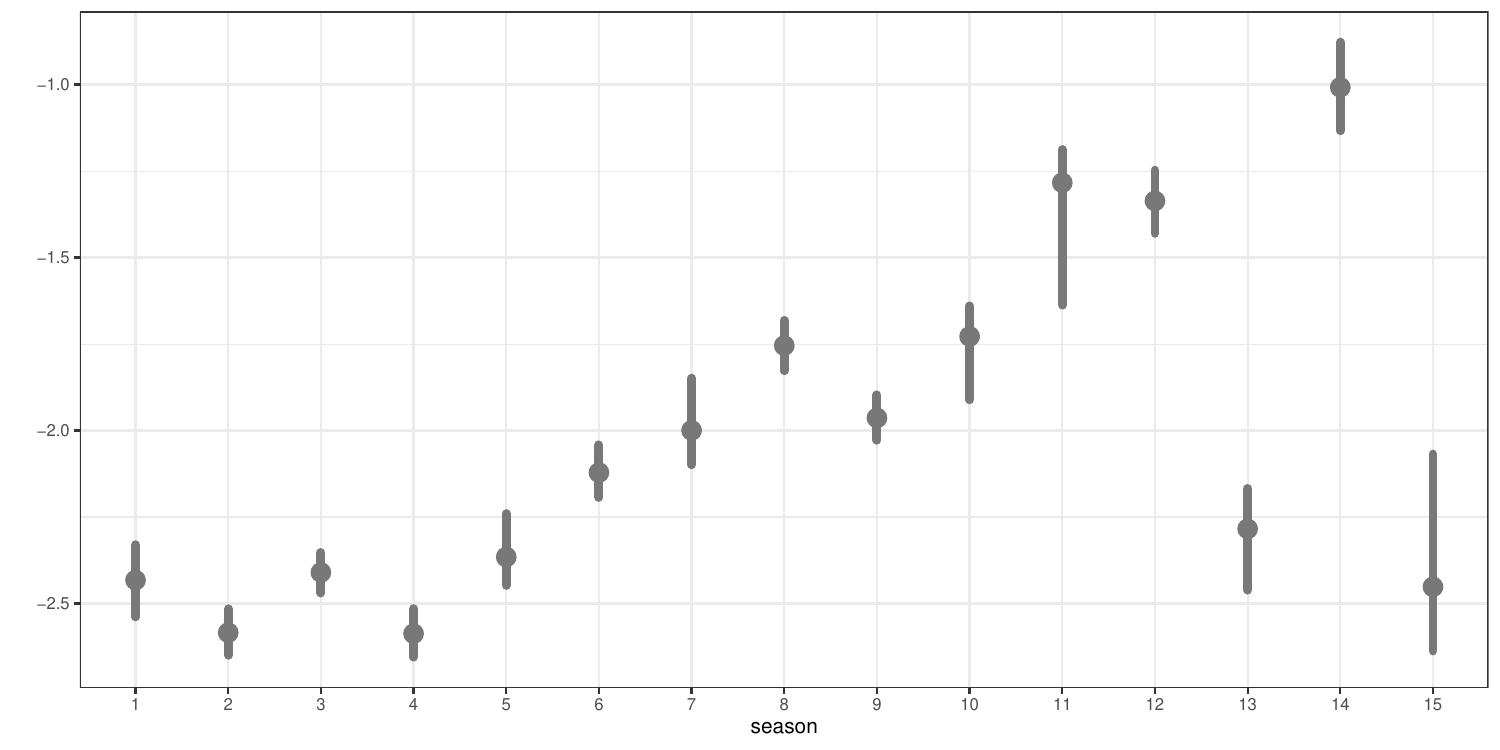}
\caption{$95\%$ posterior credible intervals 
of season-specific coefficients      $\bmbeta$'s.
}
\label{fig:betas}
\end{figure}

The final clustering has been obtained through minimization of the variation of information \citep{wade2018, dahl2022} loss function and it consists of $15$ clusters, which is also the posterior mode of the number of clusters $K$. Among these, we identify $12$ main clusters since three of them capture noisy observations and cannot be interpreted.
The estimated clusters have been relabeled so that the first cluster is the one with the highest mean value among the males’ observations. A remarkable finding is that the cluster interpretation does not depend on gender, whose effect has been filtered out by the season-specific parameter $\beta$ . 
In other words, our clustering does not trivially distinguish between males and females, but it models the athletes' performance regardless of their gender.
This claim is supported by the fact that, when ordering the clusters according to males' average performance, then females' average performances are ordered, too. Moreover, \Cref{fig:application_clust_all} reports the athletes' marks, coloured according to their cluster membership, for male and female players, respectively. The two plots are similar, highlighting that the cluster interpretation is gender-free.

\begin{figure}[!ht]
\centering
\includegraphics[width=0.49\textwidth]{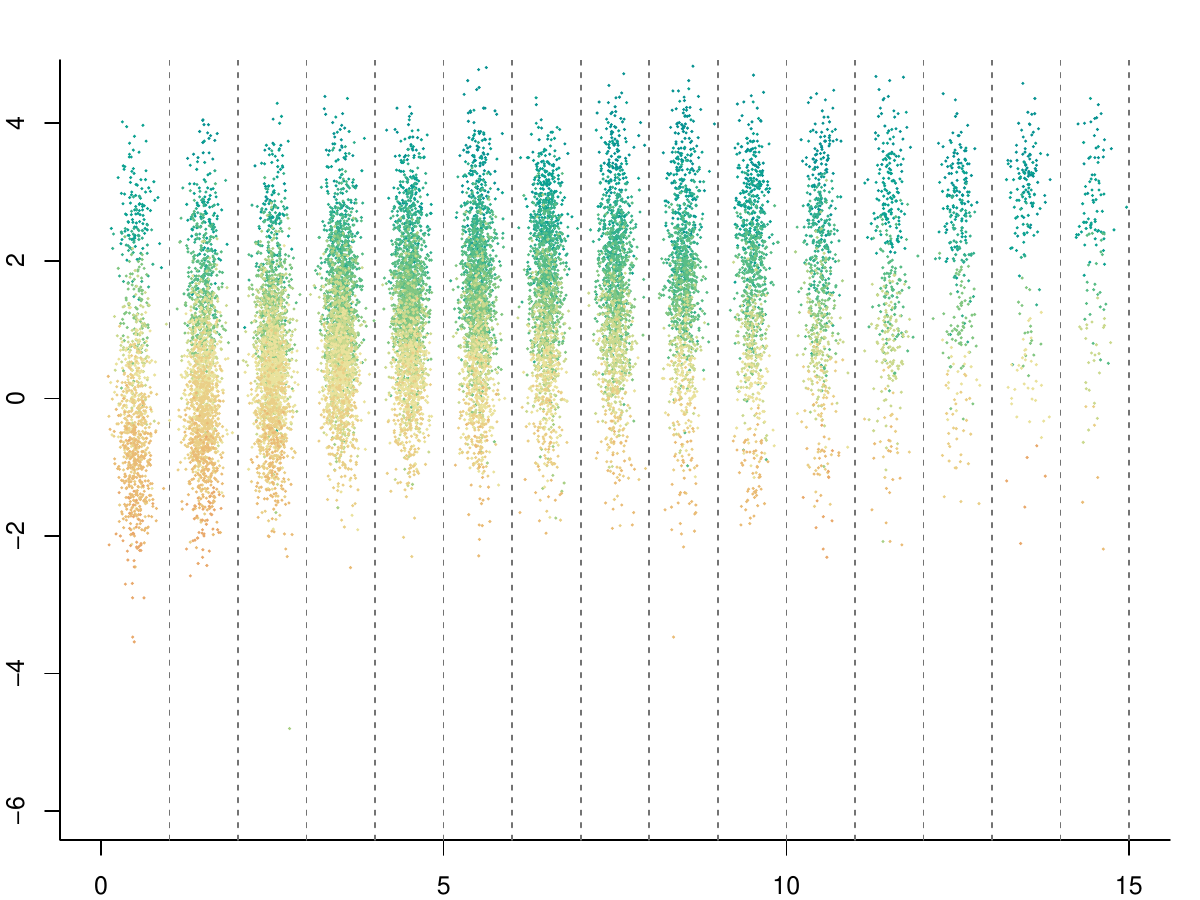}
\hfill
\includegraphics[width=0.49\textwidth]{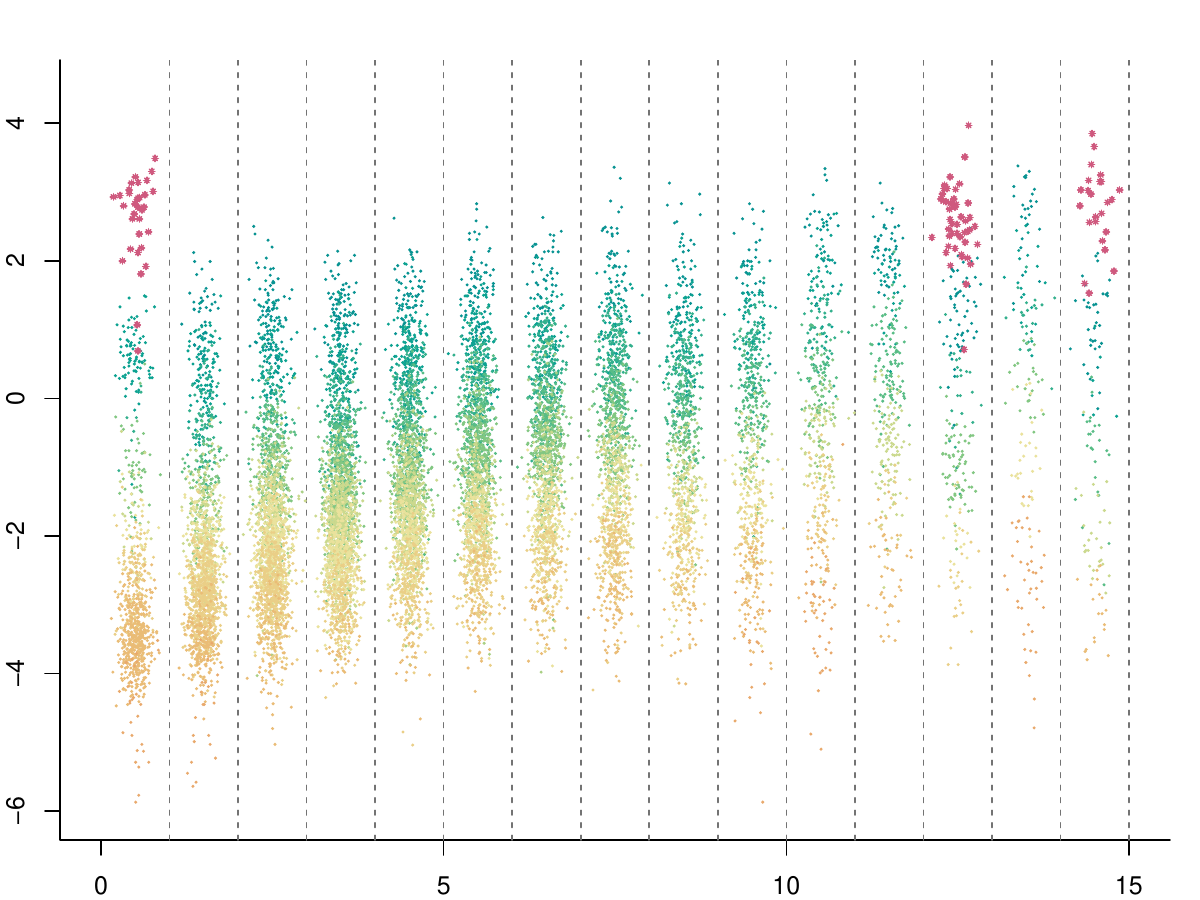}
\caption{Shot put marks for male athletes (left panel) and female athletes (right panel). 
Vertical dotted lines delimit seasons.
Dots are coloured according to the cluster membership.}
\label{fig:application_clust_all}

\end{figure}

Nevertheless, we are able to identify the presence of a particular cluster (points highlighted in the right panel of \Cref{fig:application_clust_all}) including six exceptional women performances, which are much above the average mean throw for female athletes. No man belongs to such a cluster, meaning that no one has ever been able to outperform competitors in such a neat way.
In particular, in this cluster, we find Astrid Kumbernuss, who is a three-champion and one-time Olympic champion; Christina Schwanitz, a one-time World champion; and Valerie Adams, who, during her outstanding career, won two Olympic Games and four World Championships.

We refer to \Cref{app:shotput} for tables reporting all season-specific cluster sizes and cluster summaries.
In summary, at the beginning of their careers, almost half of the athletes are grouped into cluster $10$, which is a low-ranked cluster. This finding aligns with the expectation that rookies tend to exhibit similar performances.  
However, around $6\%$ of the sample, 
immediately join clusters $2$ and $3$ (high-ranked clusters) in their very first season. Notably, these athletes have an average age of approximately $25$ years, which contrasts with the overall entrance age of less than $19$ years.
Further examination reveals that athletes typically reach the peak of their careers around seasons $7$ to $10$ when their average age is approximately $26$ years or slightly older. Analyzing cluster summaries provides additional insights. Male athletes in clusters $1$ and $2$ have average ages of $27.70$ and $27.52$, respectively, while female athletes in the same clusters have average ages of $26.08$ and $24.91$. This indicates that women tend to reach their peak performance at a younger age compared to men. A possible explanation for this disparity is that female bodies develop earlier than male bodies; this is supported by 
the fact that women begin their careers at an average age of $18.16$, while men start at an average age of $19.44$.
Cluster $12$ stands out from this trend as it comprises female athletes with a mean age of $28.05$, which is considerably older than the average peak age.
Lastly, it is worth noting that top-level clusters specifically clusters $1$ and $2$, have a higher proportion of female athletes compared to male athletes. This is despite the overall number of observations for women being smaller than that for men. This suggests that male competitions may be more balanced, making it more challenging for athletes to distinguish themselves from the average level.

\begin{figure}[h]
\centering
\includegraphics[width=0.48\linewidth]{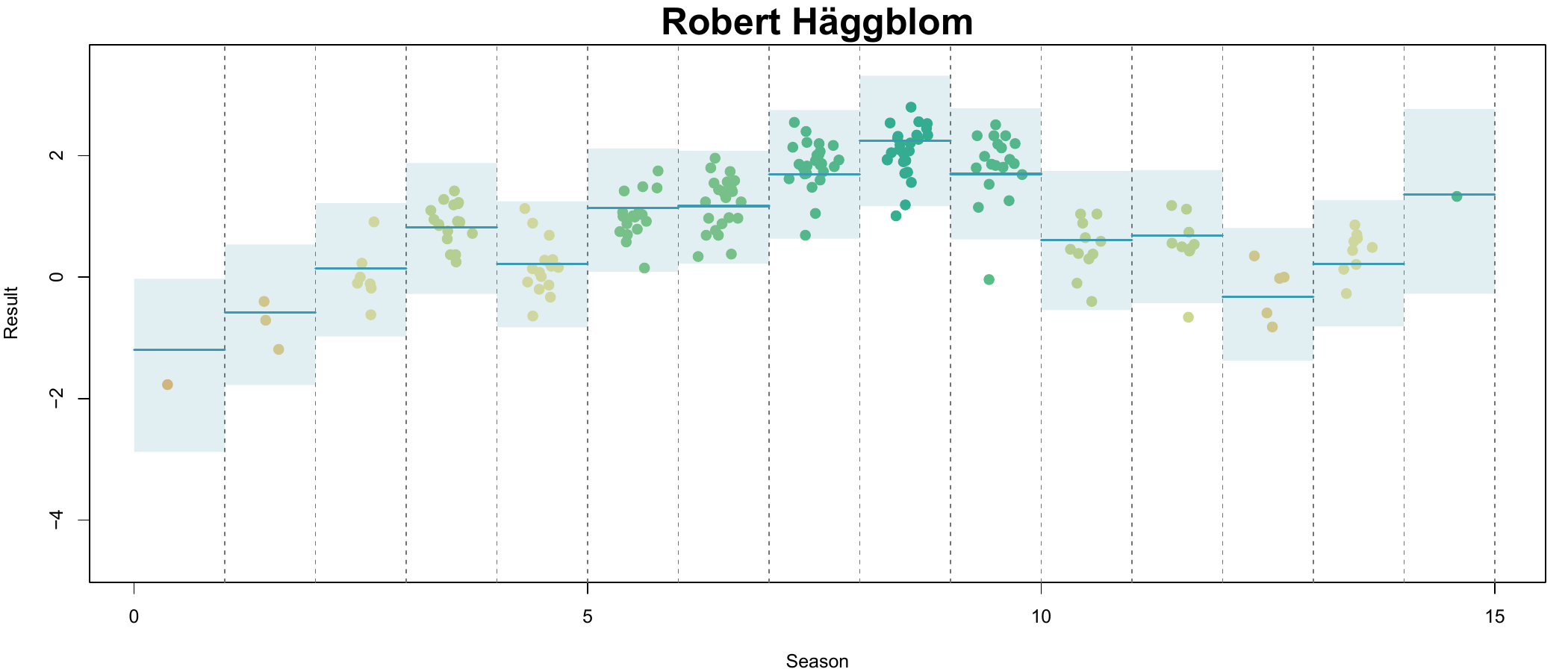}
\hfill
\includegraphics[width=0.48\linewidth]{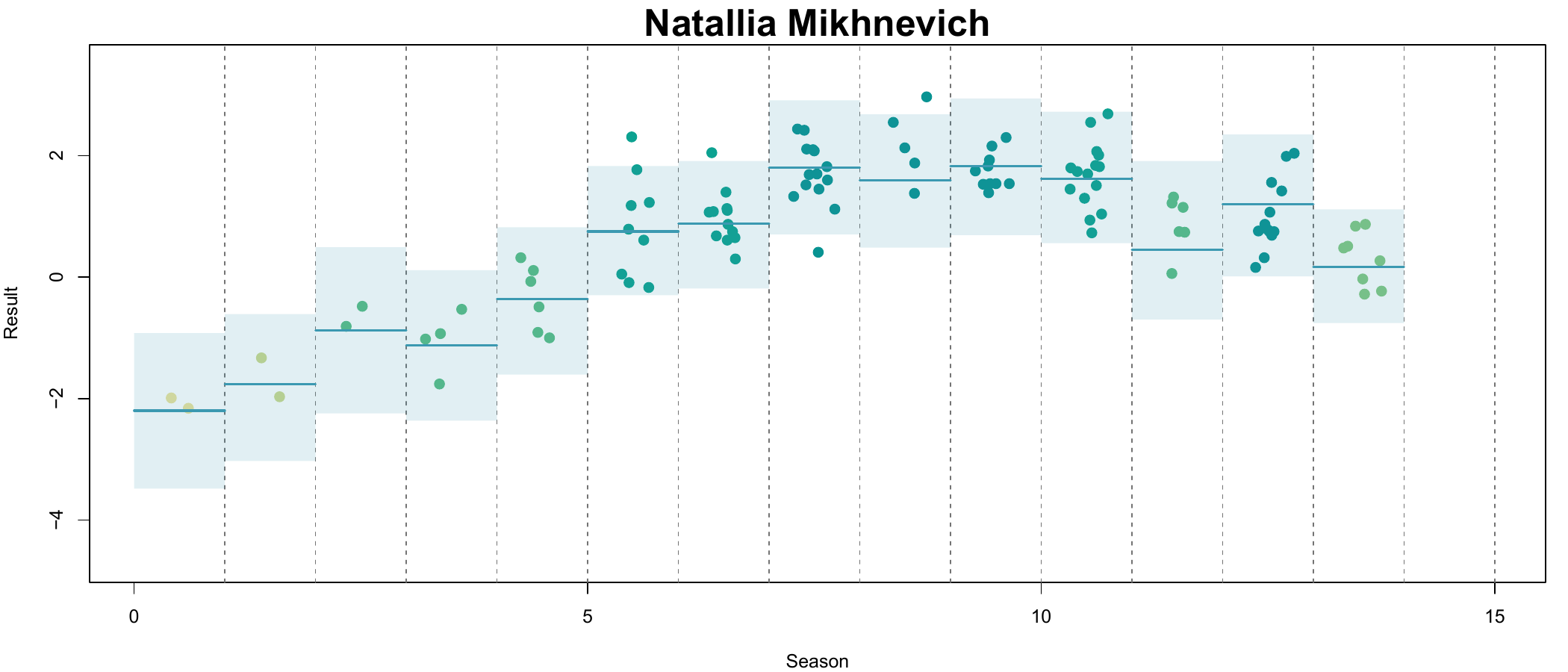}\\
\includegraphics[width=0.48\linewidth]{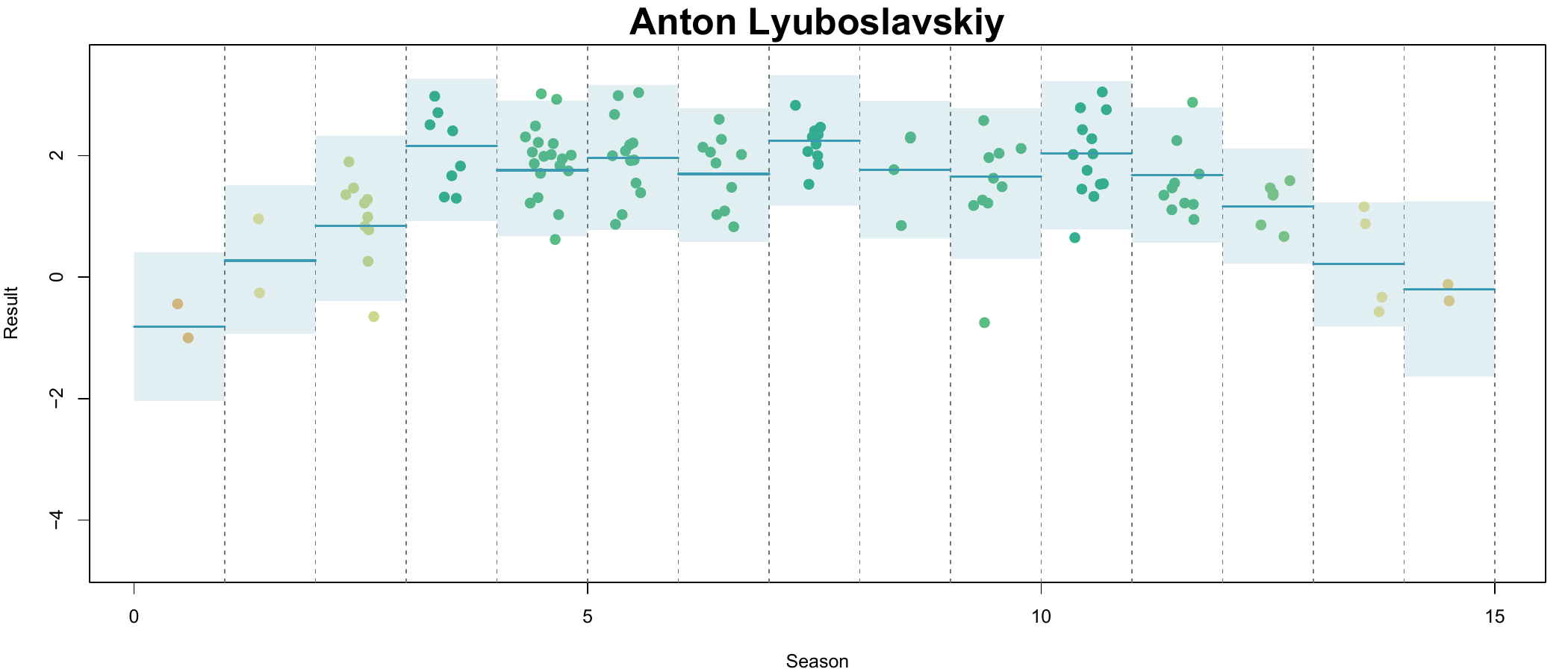}
\hfill
\includegraphics[width=0.48\linewidth]{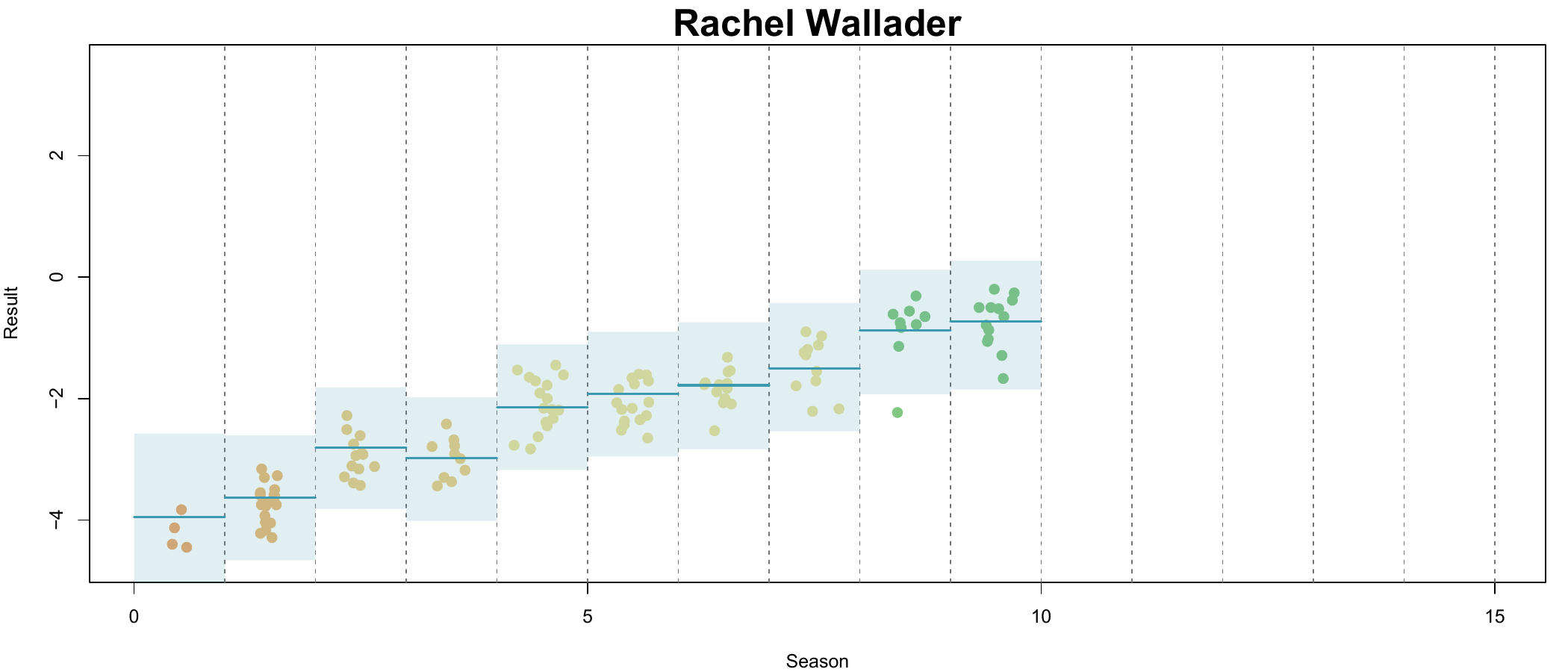}
\caption{ Shot put measurement for four randomly selected athletes. Points are coloured according to the cluster membership of the corresponding performance. Solid means represent the estimated cluster means. Shaded areas represent the 95\% credible bands.
}
\label{fig:application_traiettorie_post}
\end{figure}


A key feature of our analysis is the possibility of studying the evolution of season-specific cluster membership of each player.
This is exemplified in \Cref{fig:application_traiettorie_post}, which showcases the trajectories of four players (two men and two women). In the plot, marks representing each season are colour-coded according to their cluster memberships, while solid lines represent the estimated seasonal mean performances, and the shaded areas represent $95\%$ credible bands.
One notable pattern emerges from the trajectory of Robert H\"{a}ggblom, 
where a drop occurred immediately after the peak of his career, where he even participated in the 2008 Olympic games. A back injury conditioned the final part of its career. 
This sudden change would have been challenging to capture by a time-smoothing model. 
In contrast, Rachel Wallader 
demonstrates significant improvements throughout her career, starting from cluster $10$, which we recall being the predominant cluster among rookies, and eventually reaching the higher-performing cluster before retiring, also managing to win the British title.
Anton Lyuboslavskiy and Natallia Mikhnevich 
share remarkable career paths, showcasing exceptional performance not only for a single season but for extended periods of time, both around $9$ years. Indeed, both are Olympic level players.
Unlike Rachel Wallader, Anton Lyuboslavskiy continued to compete beyond his prime, maintaining high levels of performance but eventually transitioning to intermediate cluster levels.
Lastly, we highlight that Robert H\"{a}ggblom and Natallia Mikhnevich achieved comparable marks, but the performances of the second athlete, a woman, are assigned to higher-ranked clusters. This demonstrates our model's ability to recognize top players, regardless of their gender.
Indeed, Natallia Mikhnevich's career has been richer in success, as she won both gold and silver at the European Championships.

\section{Discussion}
\label{section:discussion}


We have introduced an innovative Bayesian nonparametric model for the analysis of grouped data, leveraging the normalization of finite dependent point processes as its foundation. Furthermore, we provided a comprehensive Bayesian analysis of this novel model class, delving into the examination of the pEPPF, posterior distributions and predictive distributions.
A special emphasis has been placed on vectors of finite Dirichlet processes,  which stand out as a noteworthy example in this context. Besides, we have also defined the HMFM as a natural extension of the work by \cite{millerharrison}. Based on our theory, marginal and conditional algorithms have been developed.
One significant benefit of HMFM is its ability to effectively capture dependence across groups.
Moreover, we have empirically shown that HMFM better calibrates the borrowing of information across groups than a traditional HDP and also performs better in terms of computational time.
In other words, the HMFM is a harmonious balance between HDP and independent analyses.
Finally, the analysis of the shot put data illuminated the flexibility of the model employed to infer athlete career trajectories and group them into clusters with a meaningful interpretation.

We now pinpoint several open problems related to our work, which are left for future research.
First, note that our construction allows a local 
and a global clustering. However, in numerous applications, one is also interested in clustering the different groups of observations. 
Nested structures, introduced by \citet{nested}, have gained popularity as valuable Bayesian tools for accomplishing this task.  Recent developments also include \citet{CAM2021} and \citet{dangelo2022}. We intend to explore the use of NIFPP to simplify the complexity of nested models, and to alleviate the computational burden associated with traditional nested structures.

Other future directions of research aim to enhance between-group 
dependence based on 
our approach. An intriguing extension we plan to investigate is to adopt a construction akin to compound random measures \citep{compound2017}.
Our idea is to replace  $S_{j,m}$, the unnormalized weight referring to group $j$ and atom $m$, with a product between a shared component across groups, i.e., depending only on $m$, and an idiosyncratic component, which depends on both group $j$ and the specific atom $m$. This modification would break the conditional independence, given $M$, of the unnormalized weights. However, it could be advantageous in situations where additional information sharing is desired,
as for the overlapping community in modular graphs, see \citep{Caron2023}.

Throughout the paper, we assumed that the atoms $\tau_1, \dots, \tau_M$ are i.i.d. according to a diffuse base measure $P_0$. However, recent research lines have explored the use of repulsive point processes as priors for mixture parameters \citep{PeRaDu12,marioJGCS2022} and applied them to model-based clustering for high-dimensional data \citep{lorenzo2023}. Furthermore, the new general theory presented by \cite{nRM2023} unifies various types of dependence on location, including independence, repulsiveness, and attractiveness. Palm calculus, which was a fundamental tool in our analysis, provides a mathematical framework for analysing these models. As a consequence, it is a promising avenue to extend our model to incorporate more sophisticated forms of dependence 
across the atoms $\tau_1, \ldots , \tau_M$.

The use of vectors of NIFPP is not limited to the mixture framework. Indeed, they can be helpful to face extrapolation problems when multiple populations of species are available.
To have a glimpse of this, consider two populations of animals composed of different species with unknown proportions. Given samples from the first and the second populations, extrapolation problems refer to
out-of-sample prediction. For instance, a typical question is: how many new and distinct species are shared across two additional samples from the populations? The seminal work of
\cite{Lij(07)} faced extrapolation problems in the simplified framework of a single population, but no results are available in the presence of multiple groups. We think that the use of vectors of NIFPP can help to face the multiple-sample setting, still unexplored in the Bayesian framework.




\newpage
\setcounter{page}{1}
\setcounter{equation}{0}
\setcounter{section}{0}
\setcounter{table}{0}
\setcounter{figure}{0}
\renewcommand\thesection{S\arabic{section}}
\renewcommand\thetable{S\arabic{table}}
\renewcommand\thefigure{S\arabic{figure}}
\renewcommand\theequation{S\arabic{equation}}

\begin{center}
	\LARGE Supplementary materials for:\\
	"Hierarchical Mixture of Finite Mixtures"
\end{center}

\section{Point Processes}
\label{app:appA}

In the present section, we remind the definition of independent finite point processes and some basic notions of Palm calculus, which is a fundamental tool for all our proofs.

\subsection{Independent Finite Point Processes}
\label{app:IFPP}
Consider a probability space 
$\left(\Omega, \mathcal{F}, \mathbb{P}\right)$, and let us denote by $\mathbb{G}$ a Polish space endowed with its corresponding Borel $\sigma$-algebra $\mathcal{G}$. Moreover, let $\mathbb{M}(\mathbb{G})$ be the space of locally finite measures on $\mathbb{G}$, and $\mathcal{M}(\mathbb{G})$ represents the corresponding  Borel $\sigma$-algebra.  
A point process $\Phi$ on $\mathbb{G}$ is a measurable map $\Phi: \mathbb{G}\rightarrow \mathbb{M}(\mathbb{G})$, defined as
\[
\Phi(B) = \sum_{m\geq 1}\delta_{\xi_m}(B), \quad B \in \mathcal{G}.
\] 
having denoted by $\delta_{\xi_m}$ the delta-Dirac mass at $\xi_m$. The sequence $(\xi_m)_{m\geq 1}$ is a random countable subset of $\mathbb{G}$ and its elements are called atoms. 

By \textit{Independent Finite Point Process} (IFPP) on $\mathbb{G}$ we mean a point process $\Phi$ whose atoms are $(\xi_m)_{m=1,\dots,M}$ where $M$ is a positive and integer-valued random variable distributed as $q_M$. The points $\xi_1,\dots,\xi_M$ are, conditionally to $M$, independent and identically distributed (i.i.d.) according to $\nu$, a probability measure on $(\mathbb{G}, \mathcal{G})$.
We write $\Phi \sim \operatorname{IFPP}(q_M,H)$.
Some texts refer to this formulation of IFPP as the \textit{mixed Binomial point process}. See \citet[Cp. 4.3]{baccelli2020} for generalization beyond i.i.d. atoms. 
Following the approach of \citet{argiento2022annals}, in order to define random probability measures, we consider an IFPP process defined on $\mathbb{G}=\R^+ \times \mathbb{X}$, for some Polish space $\mathbb{X}$ endowed with its Borel $\sigma$-algebra $\mathcal{X}$. We denote by 
$\xi = (s,\tau)$ a point of $\mathbb{G}$, and by $\xi_m =(S_m, \tau_m)$ a point of $\Phi$.
In particular, we only restrict our attention to processes such that $\nu({\rm d}s,{\rm d}x)=H({\rm d}s)\times P_0({\rm d}x)$, where $H$ is a probability distribution on $\R^+$, $P_0$ is a probability distribution on $\mathbb{X}$ and $\times$ is measure theoretical product. In other words, conditionally to $M$, the unnormalized weights $S_m$ are i.i.d. distributed according to $H$, the atoms $\tau_m$ are i.i.d. distributed according to $P_0$ with $q_M$, $H$ and $P_0$ that are independent.
We write $\Phi \sim \operatorname{IFPP}(q_M,H,P_0)$.

In this paper, we consider \textit{Vectors of Independent Finite Point Processes} (Vec-IFPPs), thus we work on $\mathbb{G}=(\R^+)^d \times \mathbb{X}$. A Vec-IFPP is a  point process whose unnormalized weights are elements of $(\R^+)^d$, namely $\bmS_m = \left(S_{1,m},\dots,S_{d,m}\right)$. 
The atoms $\tau_m$ are still i.i.d. distributed according to $P_0$, leading to the following definition of the point process $\Phi$,
\begin{equation}
\label{eqn:def_pp_VecIFPP}
\Phi(A\times B) = \sum_{m=1}^M \delta_{(\bmS_m,\tau_m)}(A\times B), 
\end{equation}
where $A$ is a Borel set of $(\R^+)^d$ and $B\in\mathcal{X}$.
The unnormalized weights $\bmS_m$ are still i.i.d. distributed according to $H$, that is now a probability distribution on $(\R^{+})^d$. For the sake of simplicity, we only resort to the case when this is defined as the measure theoretic product of probability distributions $H_j$ defined on $\R^+$, 
i.e., $H({\rm d}\bms) = H_1({\rm d}s_1)\times\dots\times H_d({\rm d}s_d)$. In addition, we also assume all $H_j$ to belong to the same parametric family having density $h(\cdot\mid\gamma_j)$ with respect to the Lebesgue measure on $\R^+$, i.e $H_j({\rm d}s) = h(s\mid\gamma_j) {\rm d}s$, where $\gamma_j\in\R^+$ are component-specific parameters. We write $\Phi \sim \operatorname{Vec-IFPP}(q_M,H,P_0)$.


\subsection{Normalization}
\label{app:norm}
Given $\Phi \sim \operatorname{Vec-IFPP}(q_M,H,P_0)$ on $\G=(\R^+)^d\times\X$, we define a vector of length $d$ of unnormalized random measures $\left(\mu_1,\dots,\mu_d\right)$ on $\X$ as follows,
\begin{equation}
\mu_j(B) \ = \ \int_0^\infty \int_B s_j \Phi({\rm d}\bms,{\rm d}\tau) = \sum_{m=1}^M S_{j,m}\delta_{\tau_m}(B), \quad B \in \mathcal{X}.
\label{eqn:app_random_measure_def}
\end{equation}
Then, the vector of random probability measures $\left(P_1,\dots,P_d\right)$ is defined through normalization as
\begin{equation}
P_j(\cdot) = \frac{\mu_j(\, \cdot \, )}{\mu_j(\mathbb{X})}.
\label{eqn:app_normalized_def}
\end{equation}
If $q_M(0)=0$, then \eqref{eqn:app_normalized_def} is well defined since $0<\mu_j(\mathbb{X})<\infty$ almost surely. 
We write $(\mu_1,\dots,\mu_d) \sim \operatorname{Vec-IFPP}(q_M,H,P_0)$ and
$\left(P_1,\dots,P_d\right) \sim \operatorname{Vec-NIFPP}(q_M,H,P_0)$, which stands for vector of Normalized Independent Finite Point Process.

\subsection{Background on Palm calculus}
\label{app:palm_background}
To establish the validity of our results, we employ Palm calculus, a fundamental tool in the analysis of point processes. Palm calculus also allows to extend  Fubini's theorem, and to enable the interchange of the expectation and integral operators when both pertain to a point process. However, there is a subtle distinction: the expectation is now taken with respect to the law of another point process, specifically, the reduced Palm version derived from the original process.
In the case of Poisson processes, which encompasses completely random measures, the reduced Palm version coincides with the law of the original Poisson process. This property characterizes Poisson processes and is known as the Slivnyak-Mecke theorem \cite[Theorem 3.2.4]{baccelli2020}.
To effectively exploit this technique, it is essential to derive the reduced Palm distribution for any IFPP. This crucial result is presented in \Cref{thm:palm}. 

In this section, we provide a concise overview of Palm calculus, focusing on the fundamental concepts required to derive our findings. For a more comprehensive understanding of the topic, we refer the reader to \citep{baccelli2020,Dal(08),Kal(17)}. We start by recalling basic, yet fundamental, definitions in the theory of point processes that are the mean measure and its generalizations, the $n$-th power and the $n$-factorial power measure.

Consider a point process $\Phi$ on $\G$, its mean measure $M_\Phi$ is a measure on $(\G,\mathcal{G})$ defined by
\begin{equation*}
\label{eqn:app_meanmeasure_def}
M_\Phi(B) = \E[\Phi(B)], \quad B\in\mathcal{G}.
\end{equation*}
Let $\Phi^n$ be the $n$-power of $\Phi$ to be defined as the $n$-th measure theoretic product, i.e., $\Phi^n(B_1\times\dots\times B_n) = \Phi(B_1)\dots\Phi(B_n)$. For any point process $\Phi$, $\Phi^n$ can be expressed as
\begin{equation}
\label{eqn:app_powermeasure_def}
\Phi^n = \sum_{\bmi \in \Delta^n} \delta_{(\xi_{i_1},\dots,\xi_{i_n})},
\end{equation}
where $\Delta^n = \bigl\{\bmi = (i_1,\dots,i_n)\ : \ i_k\in\{1,\dots,n\}, \ \forall k = 1,\dots,n \bigr\}$, hence elements of multi-index $\bmi\in\Delta^n$ may repeat and therefore atoms $(\xi_{i_1},\dots,\xi_{i_n})$ do not need to be distinct.
It can be proved that $\Phi^n$ is still a well defined point process, hence it is straightforward to define the $n$-th moment measure $M_{\Phi^n}$ of $\Phi$ as the mean measure of $\Phi^n$.
Finally, let $\Phi^{(n)}$ be the object obtained taking the summation in Equation \eqref{eqn:app_powermeasure_def} only over distinct elements, 
\begin{equation*}
\label{eqn:app_factorialmeasure_def}
\Phi^{(n)} = \sum_{\bmi \in \Delta^{(n)}} \delta_{(\xi_{i_1},\dots,\xi_{i_n})},
\end{equation*}
where $\Delta^{(n)} = \bigl\{\bmi = (i_1,\dots,i_n)\ : \ i_k\in\{1,\dots,n\}, \text{ and } i_k\neq i_l \text{ for any } k \not = l\bigr\}$.
$\Phi^{(n)}$ is called $n$-th factorial power of $\Phi$, it is a point process itself and its mean measure, $M_{\Phi^{(n)}}$, is called $n$-th factorial mean measure of $\Phi$.
To conclude, it is crucial to observe that $\Phi^n$ and $\Phi^{(n)}$ coincide when evaluated over pairwise disjoint sets, which is a crucial assumption since it is often convenient to work in terms of factorial moment measure.

Using the result in \cite[Example 4.3.11]{baccelli2020}, it follows that the mean measure and the $n$-th factorial mean measure of a $\operatorname{Vec-IFPP}$ defined as in Equation \eqref{eqn:def_pp_VecIFPP} equal to
\begin{align}
\label{eqn:app_meanmeasure_vecIFPP}
M_\Phi({\rm d}\bms,{\rm d}x ) &\ = \ 
\E[M]P_0({\rm d}x)\prod_{j=1}^d H_j({\rm d}s_{j}), \\
\label{eqn:app_factmeasure_vecIFPP}
M_\Phi^{(n)}\left(
({\rm d}\bms_1,{\rm d}x_1)
\times \dots \times
({\rm d}\bms_n,{\rm d}x_n)
\right) & \ = \ 
\E\left[M^{(n)}\right]\prod_{i=1}^nP_0({\rm d}x_i)
\prod_{i=1}^n\prod_{j=1}^d H_j({\rm d}s_{j,i}),
\end{align}
where $\E[M^{(n)}] = \E[M(M-1)\dots (M-n+1)]$.
The technique employed in our computations involves the disintegration of the Campbell measure of a point process $\Phi$ with respect to its mean measure $M_{\Phi}$, usually called Palm kernel or family of Palm distributions of $\Phi$. To be precise, we define the Campbell measure $\mathcal{C}_\Phi$ on $\G \times \mathbb{M}(\G)$ as follows
\begin{equation*}
\label{eqn:app_campbell_def}
\mathcal{C}_{\Phi}(B \times L) = \mathbb{E}\left[\Phi(B) \mathbb{I}_L(\Phi)\right], \quad B \in \mathcal{G}, L \in \mathbb{M}(\mathcal{G}).
\end{equation*}
By virtue of the Radon-Nikodym theorem, it follows that for any fixed $L$, there exists a unique disintegration probability kernel $\{\mathbf{P}^x_{\Phi}(\cdot)\}_{x \in \mathcal{G}}$ of $\mathcal{C}_\Phi$ with respect to $M_{\Phi}$, which satisfies the following integral equation
\begin{equation*}
\label{eqn:app_palm_def}
\mathcal{C}_{\Phi}(B \times L) = \int_B \mathbf{P}_{\Phi}^x(L) M_{\Phi}(\mathrm{d}x), \quad B \in \mathcal{G}, L \in \mathbb{M}(\mathcal{G}).
\end{equation*}
Note that, for any $x \in \G$, $\mathbf{P}_\Phi^x$ is the distribution of a point process $\Phi_x$ on $\G$, such that $\mathbf{P}_{\Phi}^x(L)=\P\left(\Phi_x \in L\right)$. See \cite[Theorem 31.1]{kallenberg2021}.
Notably, \cite[Proposition 3.1.12]{baccelli2020} establishes that the point process $\Phi_x$ possesses an atom located at $x$, with probability one. This property enables us to interpret $\mathbf{P}_\Phi^x$ as the probability distribution of the point process $\Phi$ given that it contains an atom at $x$. $\Phi_x$ is commonly referred to as the \textit{Palm version} of $\Phi$ at $x$.
Finally, we define $\mathbf{P}^x_{\Phi^!}$ to be the distribution of the point process
$$
\Phi_x^{!}=\Phi_x-\delta_x,
$$
which is obtained by removing the point $x$ from the Palm version $\Phi_x$. The derived point process, $\Phi_x^{!}$, is known as \textit{reduced Palm version} of $\Phi$ in $x$ and its law, $\mathbf{P}^x_{\Phi^!}$, is called \textit{reduced Palm kernel}.
Hence, given a reduced Palm kernel, we can construct the non-reduced one by considering the distribution of $\Phi_x+\delta_x$.

The subsequent theorem is known as the Campbell-Little-Mecke (CLM) formula which extends Fubini's formula and allows the exchange of an integral and an expected value, taken with respect to the law of a point process.
\begin{theorem}[Campbell-Little-Mecke formula. Theorem 3.1.9 in \cite{baccelli2020}]
\label{thm:CLMformula}
Let $\Phi$ be a point process on $\G$ such that $M_{\Phi}$ is $\sigma$-finite, and denote by $\mathbf{P}_{\Phi}(\cdot)$ the distribution of 
$\Phi$. Let $\left\{\mathbf{P}_{\Phi}^x(\cdot)\right\}_{x \in \G}$ be a family of Palm distributions of $\Phi$. Then, for all measurable $g: \G \times \mathbb{M}(\G) \rightarrow \mathbb{R}^{+}$, one has
\begin{equation}
\label{eqn:app_CLMformula}
\E\left[\int_\G g(x, \Phi) \Phi(\mathrm{d} x)\right] 
\ = \ 
\int_\G  \E\left[g(x, \Phi_x)\right]  M_\Phi(\mathrm{d} x),
\end{equation}
where the expected value on the right-hand side is taken with respect to $\mathbf{P}_\Phi^x$, i.e., with respect to the law of the point process $\Phi_x$.
\end{theorem}
Sometimes, it is useful to state Equation \eqref{eqn:app_CLMformula} in terms of the reduced Palm kernel $\mathbf{P}_{\Phi^!}^{x}$, i.e.,
\begin{equation*}
\label{eqn:app_reducedCLMformula}
\E\left[\int_{\G} g\left(x, \Phi-\delta_x\right) \Phi(\mathrm{d} x)\right] \ = \ 
\int_{\G } \E\left[g(x,\Phi^!_x)\right] M_{\Phi}(\mathrm{d} x),
\end{equation*}
where the expected value on the right-hand side is taken with respect to $\mathbf{P}^x_{\Phi^!}$.

We conclude this section by introducing the multivariate extension of Equation \eqref{eqn:app_CLMformula}, known as the higher order CLM formula.
This requires the extension to the multivariate case of all previous definitions. Therefore, given a point process $\Phi$ define the $n$-th Campbell measure
$$
\mathcal{C}_{\Phi}^n(B \times L)
\ = \ 
\E\left[\int_B \indic_L(\Phi) \Phi^n(\mathrm{d} \bm{x})\right], \quad B \in \mathcal{G}^{n}, L \in \mathbb{M}(\G)
$$
where $\mathrm{d} \bm{x}=\left(\mathrm{d} x_1 \dots \mathrm{d} x_n\right)$ and $\Phi^n(\mathrm{d} \bm{x})=\prod_{i=1}^n \Phi\left(\mathrm{d} x_i\right)$. 
Let $M_{\Phi}^n$ be the mean measure of $\Phi^n$. 
Then, the $n$-th Palm distribution $\{\mathbf{P}_\Phi^{\bmx} \}_{\bm{x} \in \G^n}$ is defined as the disintegration kernel of $\mathcal{C}_{\Phi}^n$ with respect to $M_{\Phi^n}$, that is
$$
\mathcal{C}_{\Phi}^n(B \times L)=\int_B \mathbf{P}_{\Phi}^x(L) M_{\Phi^n}(\mathrm{d} \bm{x}), \quad B \in \G^{n}, L \in \mathbb{M}(\G).
$$
The \textit{higher order Palm version} $\Phi_{\bmx}$ is defined as the point process whose distribution is given by $\mathbf{P}_\Phi^{\bmx}$. Consequently, the \textit{higher order reduced Palm version} is the point process
$$
\Phi_{\bmx}^{!}=\Phi_{\bmx}-\sum_{i=1}^n\delta_{x_i},
$$
whose probability law, $\mathbf{P}_{\Phi^!}^{\bmx}$, is known as higher order reduced Palm kernel.
We are finally ready to state the multivariate extension of \Cref{thm:CLMformula} and its reduced form formulation. 
\begin{theorem}[Higher order CLM formula. Theorem 3.3.2 in \cite{baccelli2020}]
\label{thm:highorder_CLMformula}
Let $\Phi$ a point process on $\G$ such that $M_{\Phi^n}$ is $\sigma$-finite. Let $\left\{\mathbf{P}_{\Phi}^x(\cdot)\right\}_{\bm{x} \in \G^k}$ be a family of $n$-th Palm distributions of $\Phi$. Then, for all measurable $g: \R^n \times \mathbb{M}(\G) \rightarrow \R^{+}$
\begin{equation}
\label{eqn:highorder_CLMformula}
\E\left[\int_{\G^n} g(\bm{x}, \Phi) \Phi^n(\mathrm{d} \bm{x})\right]
\ = \ 
\int_{\G^n }
\E\left[g(\bm{x}, \Phi_{\bmx})\right]
M_{\Phi^n}(\mathrm{d} \bm{x}) . 
\end{equation}
\end{theorem}
When $\left(x_1, \ldots, x_n\right)$ are distinct, we have that the $n$-th moment measure $M_{\Phi^n}$ coincides with the $n$-th factorial moment measure $M_{\Phi^{(n)}}$. In this setting, we state the CLM formula we are going to use most frequently in our computations \citep[Theorem 3.3.6]{baccelli2020}. 
This is when we consider the reduced Palm version of $\Phi$ in $\bmx$ and the integral is taken with respect to the $n$-th factorial power of $\Phi$, namely,
\begin{equation}
\label{eqn:app_reduced_highorder_CLMformula}
\E\left[\int_{\G^n} 
g\left(\bmx, \Phi - \sum_{i=1}^n\delta_{x_i} \right) 
\Phi^{(n)}(\mathrm{d}\bmx)\right] \ = \ 
\int_{\G^n} 
\E\left[g(\bmx,\Phi^!_{\bmx})\right] M_{\Phi^{(n)}}(\mathrm{d}\bmx).
\end{equation}

\subsection{Properties of Vec-IFPP}
\label{app:properties}
In this section we derive the main properties of Vec-IFPP, in particular we characterize the higher order reduced Palm distribution associated to any IFPP defined on some Polish space $\mathbb{G}$ as well as its Laplace functional.
\begin{theorem}
\label{thm:palm}
Let $\Phi\sim\operatorname{IFPP}(q_M,\nu)$ defined on a Polish space $\mathbb{G}$, where $q_M$ is a discrete distribution over $\left\{0,1,2,\dots\right\}$ and $\nu$ is probability distribution on $\G$, independent of $q_M$. The higher order reduced Palm version $\Phi^!_{\bmxi}$ associated to $\Phi$ at $\bmxi = (\bmxi_1,\dots\bmxi_k)$ is a point process 
\[
\Phi^!_{\bmxi} = \sum_{m=1}^M \delta_{\bmxi_m}
\]
such that $\Phi^!_{\bmxi} \sim \operatorname{IFPP}(q^!_M,H)$. In particular, $q^!_M$ is defined as
\begin{equation}
q^!_M(m) = \frac{1}{\E[M^{(k)}]} q_M(m+k)\frac{(m+k)!}{m!},
\label{eqn:q!M}
\end{equation}
where 
\[
\E[M^{(k)}] := \E[M(M-1)\dots (M-k+1)] = \sum_{m=0}^\infty q_M(m+k)\frac{(m+k)!}{m!}
\]
is the normalizing constant of $q^!_M$ 
and the atoms $\xi_m$ are still i.i.d.  according to $H$.
\end{theorem}
\begin{proof}
We write the higher order reduced Campbell-Little-Mecke (CLM) formula \eqref{eqn:app_reduced_highorder_CLMformula} in our case:
\begin{equation}
\E\left[
\bigintssss_{\mathbb{G}^k} f\left(\xi_1,\dots,\xi_k, \Phi - \sum_{j=1}^k\delta_{\xi_k}\right)\Phi^{(k)}({\rm d}\bmxi) 
\right] 
=
\bigintssss_{\mathbb{G}^k} \E\left[
f\left(\xi_1,\dots,\xi_k, \Phi^!_{\bmxi}\right)
\right]
\text{M}_{\Phi^{(k)}}({\rm d}\bmxi),
\label{eqn:horCLM}
\end{equation}
where ${\rm d}\bmxi = ({\rm d}\xi_1,\dots,{\rm d}\xi_k)$ and $\text{M}_{\Phi^{(k)}}$ is the $k$-th factorial moment measure of $\Phi$. 
We now focus on the left-hand side of \eqref{eqn:horCLM}. Using the definition of $k$-th factorial process, the  integral inside the expected value is
\begin{equation*}
\begin{split}
\bigintsss_{\mathbb{G}^k} f\left(\xi_1,\dots,\xi_k, \Phi - \sum_{j=1}^k\delta_{\xi_k}\right)\Phi^{(k)}({\rm d}\bmxi)  & =
\sum_{\bmi\in\Delta^{(k)}}f\left(\xi_{i_1},\dots,\xi_{i_k}, 
\sum_{m=1}^M \delta_{\xi_m} - \sum_{j=1}^k\delta_{\xi_{i_j}}\right)  \\ 
&  = 
\sum_{\bmi\in\Delta^{(k)}}f\left(\xi_{i_1},\dots,\xi_{i_k}, 
\sum_{j=1}^{M-k}\delta_{\xi_{j}}\right),
\end{split}
\end{equation*}
where the set of $k$-tuple $\Delta^{(k)}$ is defined as 
$\Delta^{(k)} := \{ \bmi = (i_1,\dots,i_k) \ : \ i_j\in\{1,\dots,k\} \text{ and } i_h\neq i_j \text{ for all }h\neq j \}$. In the equation above, the equality from the first to the second line holds, given that $k\leq M$, since $k$ out of $M$ are subtracted from the process. The order of the atoms does not matter since they are exchangeable, hence it is fine to remove the first $k$ ones.
Then, using the result above for the inner integral and by definition of IFPP, the expected value on the left-hand side of \eqref{eqn:horCLM} 
equals
\begin{equation}
\label{eqn:palmproof_1}
\begin{split}
\E\Biggl[ & \int_{\mathbb{G}^k}f\left(\xi_1,\dots,\xi_k, \Phi - \sum_{j=1}^k\delta_{\xi_k}\right)  \Phi^{(k)}({\rm d}\bmxi) 
\Biggr]  \\
& \qquad = \sum_{m=k}^\infty q_M(m)
\sum_{\bmi \in \Delta^{(k)}}
\E\left[
f\left(\xi_{i_1},\dots,\xi_{i_k}, 
\sum_{j=1}^{m-k}\delta_{\xi_{j}}\right)
\right]
\end{split}
\end{equation}
where the outer sum is constrained to values $m\geq k$; otherwise, the inner sum is not even defined. Now, we focus on the expected value on the right-hand side. Note that this is an integral on $\mathbb{G}^m$ with respect to $(\xi_1,\dots,\xi_m)$. Using Fubini's theorem, such integral can be written in two iterative integrals, one on the first $k$ atoms and one on the last $m-k$ ones:
\begin{equation}
\label{eqn:palmproof_2}
\begin{split}
& \sum_{\bmi \in \Delta^{(k)}}
\E\left[
f\left(\xi_{i_1},\dots,\xi_{i_k}, 
\sum_{j=1}^{m-k}\delta_{\xi_{j}}\right)
\right] = 
\sum_{\bmi \in \Delta^{(k)}}
\E\left[
\bigintssss_{\mathbb{G}^k}
f\left(x_{1},\dots,x_{k}, 
\sum_{j=1}^{m-k}\delta_{\xi_{j}}\right)
P_0^k({\rm d}\bmx) 
\right]  \\
& \qquad =
\frac{m!}{(m-k)!}
\bigintssss_{\mathbb{G}^k}
\E\left[
f\left(x_{1},\dots,x_{k}, 
\sum_{j=1}^{m-k}\delta_{\xi_{j}}\right)
\right]
\prod_{j=1}^kP_0({\rm d}x_j) 
\end{split}
\end{equation}
where ${\rm d}\bmx = ({\rm d}x_1,\dots,{\rm d}x_k)$ and the expected value is now taken with respect to the law of the final $m-k$ atoms. Then, the final equality follows using Fubini's theorem once again and  by noticing  that the function in the sum does not depend  on the index $\bmi$; in addition we have observed that the cardinality of set $\Delta^{(k)}$ equals  $\binom{m}{k}/k! = m!/(m-k)!$. 
Then, we complete our manipulation of the left-hand side of 
\eqref{eqn:horCLM} by
plugging \eqref{eqn:palmproof_2} into \eqref{eqn:palmproof_1}, and we exploit Fubini's theorem to exchange the integral and the series. Therefore, renaming $\xi = x$, \eqref{eqn:horCLM} is equivalent to 
\begin{equation}
\label{eqn:palmproof_3}
\begin{split}
& \sum_{m=k}^\infty
q_M(m)
\frac{m!}{(m-k)!}
\bigintssss_{\mathbb{G}^k}
\E\left[
f\left(\xi_{1},\dots,\xi_{k}, 
\sum_{j=1}^{m-k}\delta_{\xi_{j}}\right)
\right]
\prod_{j=1}^kP_0({\rm d}\xi_j) 
\\
& \qquad \qquad = 
\int_{\mathbb{G}^k} \E\left[
f\left(\xi_1,\dots,\xi_k, \Phi^!_{\bmxi}\right)
\right]
\text{M}_{\Phi^{(k)}}({\rm d}\bmxi).
\end{split}
\end{equation}
We conclude using identity \eqref{eqn:palmproof_3} to identify $\Phi^!_{\bmxi}$. To this aim, we recall that $\text{M}_{\Phi^{(k)}}({\rm d}\bmxi) = \E[M^{(k)}]\prod_{j=1}^k P_0(\xi_j)$ \citep{baccelli2020} and, by  changing the index of the summation, we get
\begin{equation*}
\label{eqn:palmproof_4}
\begin{split}
& \sum_{m=0}^\infty
q_M(m+k)
\frac{(m+k)!}{m!}
\frac{1}{\E[M^{(k)}]}
\bigintssss_{\mathbb{G}^k}
\E\left[
f\left(\xi_{1},\dots,\xi_{k}, 
\sum_{h=1}^{m}\delta_{\xi_h}\right)
\right]
\text{M}_{\Phi^{(k)}}({\rm d}\bmxi)
\\
& \qquad \qquad = 
\int_{\mathbb{G}^k} \E\left[
f\left(\xi_1,\dots,\xi_k, \Phi^!_{\bmxi}\right)
\right]
\text{M}_{\Phi^{(k)}}({\rm d}\bmxi).
\end{split}
\end{equation*}
Then, the statement of the theorem  simply follows by identification.
\end{proof}
We now state the main results we use in subsequent proofs, these just trivially follow from \Cref{thm:palm}.
\begin{corollary}
\label{cor:palm4VIFPP}
Let $\Phi\sim\operatorname{Vec-IFPP}(q_M,H,P_0)$ defined on $(\R^+)^d\times\mathbb{X}$, where $q_M$ is a discrete distribution over $\left\{1,2,3,\dots\right\}$, $H$ and $P_0$ are probability distributions over $(\R^+)^d$ and $\X$, respectively, such that $q_M$, $H$ and $P_0$ are independent.
Then, the Palm distribution $\{P_\Phi^{\bms,x}\}_{(\bms,x)\in(\R^+)^d\times\mathbb{X}}$ is the distribution of the point process $\delta_{(\bms,x)} + \Phi^!_{\bms,x}$, where $\Phi^!_{\bms,x} \sim \operatorname{Vec-IFPP}(q^!_M,H,P_0)$ and $q^!_M$ is given in \eqref{eqn:q!M} by setting $k=1$.

\noindent Similarly, let $(\bms,\bmx) = (\bms_1,\dots,\bms_k,x_1,\dots,x_k)$, then the higher order Palm distribution 
$\{P_\Phi^{\bms,\bmx}\}_{(\bms,x)\in(\R^+)^{dk}\times\mathbb{X}^k}$ is the distribution of the point process 
$\sum_{i=1}^k\delta_{(\bms_i,x_i)} + \Phi^!_{\bms,\bmx}$, 
where $\Phi^!_{\bms,\bmx} \sim \operatorname{Vec-IFPP}(q^!_M,H,P_0)$ and $q^!_M$ is given in \eqref{eqn:q!M}.
\end{corollary}
In particular, note that $\Phi^!_{\bms,\bmx}$ in \Cref{cor:palm4VIFPP}
depends on $(\bms,\bmx)$ only through the length $k$. 
Moreover, see that $q^!_M(0) > 0$ even if $q_M(0) = 0$.\\

The following result provides the Laplace functional of a $\operatorname{Vec-IFPP}(q_M,H,P_0)$ distributed point process.

\begin{lemma}
\label{lemma:laplace}
Let $\Phi \sim \operatorname{Vec-IFPP}(q_M,H,P_0)$ on $\mathbb{G}=(\R^+)^d\times\mathbb{X}$, where $q_M$ is a discrete distribution over $\{0,1,2,\dots\}$, $H$ is probability distribution on $(\R^+)^d$, independent of $q_M$ and $P_0$ is a probability measure on $\mathbb{X}$, independent of $H$ and $q_M$.
Then, for any $f:(\R^+)^d \times \mathbb{X} \rightarrow \R^+$ the  Laplace functional of $\Phi$ is
\begin{equation*}
\label{eqn:laplace}
\begin{split}
\operatorname{L}_{\Phi}[f] :&= 
\E\left[\exp \left\{-\int_{(\R^+)^d \times \mathbb{X}} f(\bms,x)\Phi({\rm d}\bms,{\rm d}x)\right\}\right]
\\
&  = 
\sum_{m=0}^\infty q_M(m)
\left( 
\int_{\mathbb{G}} e^{-f(\bms,x)} H({\rm d}\bms)P_0({\rm d}x)
\right)^m .
\end{split}
\end{equation*}
\end{lemma}
\begin{proof}

First recall the identity $\text{L}_{\Phi}[f] = \text{G}_\Phi(e^{-f})$ between the Laplace functional and the generating functional. Then, we use \cite[Lemma 4.3.14]{baccelli2020} to write
\begin{equation*}
\begin{split}
\text{L}_{\Phi}[f] &= 
\E\left[\prod_{(\bms,x) \in \Phi} \exp\{-f(\bms,x)\}\right] = 
\sum_{m=0}^\infty
q_M(m)\E\left[\prod_{(\bms,x) \in \Phi} \exp\{-f(\bms,x)\}
\mid\Phi(\mathbb{G})=m\right]  \\
&  = q_M(0) + \sum_{m=1}^\infty
q_M(m) \int_{\mathbb{G}^m}
\left(\prod_{i=1}^m \exp\{-f(\bms_i,x_i)\} \right)H^m(d\bms_1,\dots,\bms_m)P_0^m({\rm d}x_1,\dots,{\rm d}x_m)  .
\end{split}    
\end{equation*}
We now leverage the mutual independence between unnormalized weights and the fact that atoms are i.i.d. to obtain
\begin{equation*}
\begin{split}
\text{L}_{\Phi}[f] \ = \  
\sum_{m=0}^\infty q_M(m)
\left( \int_{\mathbb{G}} \exp\{-f(\bms,x)\}
H({\rm d}\bms)P_0({\rm d}x)\right)^m,
\end{split}    
\end{equation*}
which concludes the proof.
\end{proof}

In particular, \Cref{lemma:laplace} can be used to compute the Laplace functional of the random measures $\left(\mu_1,\dots,\mu_d\right)$ defined in \eqref{eqn:app_random_measure_def}. 
\begin{lemma}
\label{lemma:laplace4mu}
Let $\Phi\sim\operatorname{Vec-IFPP}(q_M,H,P_0)$ be defined as in \Cref{lemma:laplace}. Let $(\mu_1,\dots,\mu_d)$ be random measures on $\X$ defined through $\Phi$ as in \eqref{eqn:app_random_measure_def}. 
Then, for any set of measurable functions $f_1,\dots,f_d$ such that each $f_j:\X\rightarrow\R^+$, we have
\begin{equation*}
\label{eqn:laplace4mu}
\begin{split}
\E\left[
\exp \left\{
-\sum_{j=1}^d \int_{\X} f_j(x)\mu_j({\rm d}x)
\right\}
\right]
= 
\sum_{m=0}^\infty
q_M(m)
\left(
\int_{\G}
\exp \left\{
-\sum_{j=1}^d f_j(x)s_j
\right\}
H({\rm d}\bms)P_0({\rm d}x)
\right)^m .
\end{split}
\end{equation*}
\end{lemma}

\begin{proof}
To prove the result, it is enough to show that this is a special case of \Cref{lemma:laplace}. Indeed, relying on Equation \eqref{eqn:app_random_measure_def}, we can express $\mu_j(\mydiff x)$ as $\mu_j(\mydiff x) = \int_\G \delta_{\tau_m}(\mydiff x) s_j \phi(\mydiff \bms, \mydiff \tau)$, then
\begin{equation*}
\begin{split}
\E\left[
\exp \left\{
-\sum_{j=1}^d \int_{\X} f_j(x)\mu_j({\rm d}x)
\right\}
\right] &= 
\E\left[
\exp \left\{
-\sum_{j=1}^d \int_{\X} f_j(x) 
\int_\G \delta_y({\rm d}x)s_j \Phi({\rm d}\bms,{\rm d}y)
\right\}
\right]  \\
&  =  
\E\left[
\exp \left\{
-\sum_{j=1}^d \int_{\G} f_j(x) s_j \Phi({\rm d}\bms,{\rm d}x)
\right\}
\right] = 
\text{L}_{\Phi}[g],
\end{split}
\end{equation*}
where $g(\bms,x)=\sum_{j=1}^d f_j(x)s_j$.
The statement now follows by an application of \Cref{lemma:laplace} with $f(\bms, x) = g(\bms,x)$.
\end{proof}

The following lemma is extensively used throughout our proofs. It consists of the Laplace transform of the reduced Palm version of $\Phi$ at $(\bms,\bmx)$, evaluated at a special choice of $f$ function.

\begin{lemma}
\label{lemma:Laplace2}
Let $\Phi$ be defined as in \Cref{lemma:laplace}. 
Let $(\bms,\bmx) = (\bms_1,\dots,\bms_k,x_1,\dots,x_k)$ and let $\Phi^!_{\bms,\bmx}$ be the higher order reduced Palm version of $\Phi$ at $(\bms,\bmx)$. 
Then, for any $\left(u_1,\dots,u_d\right)\in(\R^+)^d$, the following holds
\begin{equation*}
\E\left[
\exp\left\{
- \sum_{j=1}^d \int_{\mathbb{G}}
u_j w_j \Phi^!_{\bms,\bmx}({\rm d}\bmw,{\rm d}y)
\right\}
\right] \ = \
\frac{1}{\E[M^{(k)}]}\Psi(k,\bmu)
\label{eqn:reduced_laptrans}
\end{equation*}
where $\Psi(k,\bmu)$ is given by
\begin{equation}
\label{eqn:Psi}
\Psi(k,\bmu) = \sum_{m=0}^\infty \ q_M(m+k)\frac{(m+k)!}{m!}
\prod_{j=1}^d\left( \psi_j(u_j) \right)^m,
\end{equation}
and $\psi_j(u_j)$ is the Laplace transform of a random variable $S_j\sim H_j$, namely
\begin{equation*}
\label{lemma:psi}
\psi_j(u_j) = \E[e^{-u_j S_j}] = \int_{0}^\infty e^{-u_j s} H_j({\rm d}s).
\end{equation*}
\end{lemma}
\begin{proof}
The proof follows by combining the reduced Palm characterization given in \Cref{cor:palm4VIFPP} and \Cref{lemma:laplace} with $f(\bmw,y) = \sum_{j=1}^d u_jw_j$.
In particular, note that $f(\bmw,y)$ does not depend on $y$.
\end{proof}

\section{Proofs of the main results}
\label{section:appB}

\subsection{Proof of Theorem \ref{thm:peppf} }
\begin{proof}\label{proof:peppf}
We compute the pEPPF by resorting to the definition given by \cite{camerlenghi2019distribution}:
\begin{equation}
\label{eq:pEPPF_def}
\Pi_K^{(n)}(\bm{n}_1,\dots,\bm{n}_d) = \E\left[ \int_{\X^K} \prod_{j=1}^d\prod_{k=1}^K
P_j({\rm d}\theta^{**}_k)^{n_{j,k}}\right].
\end{equation}
By an application of  Fubini's theorem, we can exchange the integral and the expected value in \eqref{eq:pEPPF_def}. Thus, we focus on the evaluation of the following expected value:
\begin{equation}
\label{eqn:peppfproof_1}
\begin{split}
\E&\left[
\prod_{j=1}^d\prod_{k=1}^K
P_j({\rm d}\theta^{**}_k)^{n_{j,k}}
\right]  =
\E\left[
\prod_{j=1}^d\prod_{k=1}^K
\frac{1}{\mu_j(\mathbb{X})^{n_{j,k}}}\mu_j({\rm d}
\theta^{**}_k)^{n_{j,k}}
\right]  \\
&=
\E\left[
\prod_{j=1}^d\prod_{k=1}^K
\left\{
\int_0^\infty
\frac{u_j^{n_{j,k}-1}}{\Gamma(n_{j,k})}e^{-u_j\mu_j(\mathbb{X})}{\rm d}u_j
\right\}
\mu_j({\rm d}\theta^{**}_k)^{n_{j,k}}
\right]  \\
&=
\int_{[0,\infty]^d}
\E\left[
\prod_{j=1}^d
\frac{u_j^{n_{j}-1}}{\Gamma(n_{j})}e^{-u_j\mu_j(\mathbb{X})}
\prod_{k=1}^K \mu_j({\rm d}\theta^{**}_k)^{n_{j,k}}
\right]
{\rm d}\bmu  \\
&=
\int_{[0,\infty]^d}
\prod_{j=1}^d
\frac{u_j^{n_{j}-1}}{\Gamma(n_{j})}
\E\left[
\exp\left\{-\sum_{j=1}^d u_j\mu_j(\mathbb{X})\right\}
\prod_{k=1}^K \prod_{j=1}^d \mu_j({\rm d}\theta^{**}_k)^{n_{j,k}}
\right]
{\rm d}\bmu,
\end{split}
\end{equation}
where ${\rm d}\bmu=({\rm d}u_1,\dots,{\rm d}u_d)$. The first equality in \eqref{eqn:peppfproof_1} follows by definition of $P_j$, the second one by the identity
$\frac{1}{x^n} = \int_0^\infty \frac{u^{n-1}}{\Gamma(n)}e^{-ux}dx$, the third one follows by Fubini's theorem and recalling that $\sum_{k=1}^K n_{j,k} = n_j$ while the fourth one is just a convenient arrangement of terms.
Observe that the product of the random measures $\mu_j$'s in the last term of 
\eqref{eqn:peppfproof_1} equals

\begin{equation}
\label{eqn:peppfproof_2}
\begin{split}
\prod_{k=1}^K \prod_{j=1}^d \mu_j({\rm d}\theta^{**}_k)^{n_{j,k}} &= 
\prod_{k=1}^K  \int_{\mathbb{G}}
\delta_{x_k}({\rm d}\theta^{**}_k)
\prod_{j=1}^d s_{j,k}^{n_{j,k}}
\Phi({\rm d}\bms_k,{\rm d}x_k) \\
&=
\int_{\mathbb{G}^K} \prod_{k=1}^K
\left( 
\delta_{x_k}({\rm d}\theta^{**}_k)\prod_{j=1}^d s_{j,k}^{n_{j,k}}
\right)
\Phi^K({\rm d}\bms,{\rm d}\bmx),
\end{split}
\end{equation}
where $\bms_k = (s_{1,k},\dots,s_{d,k})$ are $(\R^+)^d$-dimensional integration variables representing the unnormalized weights of the Vec-IFPP with atoms $x_k$. 
Then, $\bmx=(x_1,\dots,x_k)$ is a $k$-dimensional vector while $\bms = (\bms_1,\dots,\bms_k)$ is a $k$-dimensional vector of $(\R^+)^d$-dimensional components. 
Finally, $\Phi^K({\rm d}\bms,{\rm d}\bmx)$ is the $K$-th power point process of $\Phi$, as defined in \Cref{app:palm_background}. See \citet[Ch. 1.3.1]{baccelli2020} for further details. \\
Plugging \eqref{eqn:peppfproof_2} into the last member of \eqref{eqn:peppfproof_1} and exploiting the definition of $\mu_j(\mathbb{X}) = \int_\mathbb{G} w_j\Phi({\rm d}\bmw,{\rm d}z)$, we get
\begin{equation}
\label{eqn:peppfproof_3}
\begin{split}
\E\left[
\prod_{j=1}^d\prod_{k=1}^K
P_j({\rm d}\theta^{**}_k)^{n_{j,k}}
\right] =
& \int_{[0,\infty]^d}
\prod_{j=1}^d
\frac{u_j^{n_{j}-1}}{\Gamma(n_{j})}
\E\Biggl[
\int_{\mathbb{G}^K}
\prod_{k=1}^K
\left( 
\delta_{x_k}({\rm d}\theta^{**}_k)\prod_{j=1}^d s_{j,k}^{n_{j,k}}
\right) \\
& \ \times
\exp\left\{
-\sum_{j=1}^d\int_\mathbb{G} u_j w_j\Phi({\rm d}\bmw,{\rm d}z)
\right\}
\Phi^K({\rm d}\bms,{\rm d}\bmx)
\Biggr]
{\rm d}\bmu.
\end{split}
\end{equation}
Now, we only focus on the expected value on the right-hand side of \eqref{eqn:peppfproof_3}. 
In the sequel,  we assume that $d\theta^{**}_k$ are infinitesimally disjoint sets which do not overlap, say balls centred at
$\theta^{**}_k$ with a sufficiently small radius $\varepsilon \to 0$. 
We now apply the higher order CLM formula, stated in Equation \eqref{eqn:highorder_CLMformula},
with 
\[
f(\bms,\bmx,\Phi) = 
\prod_{k=1}^K
\left( 
\delta_{x_k}({\rm d}\theta^{**}_k)\prod_{j=1}^d s_{j,k}^{n_{j,k}}
\right) 
\exp\left\{
-\sum_{j=1}^d\int_\mathbb{G} u_j w_j\Phi({\rm d}\bmw,{\rm d}z)
\right\}.
\]
As a consequence, the expected value under study boils down to
\begin{equation}
\label{eqn:peppfproof_4}
\begin{split}
&\E\left[
\int_{\mathbb{G}^K}
\prod_{k=1}^K
\left( 
\delta_{x_k}({\rm d}\theta^{**}_k)\prod_{j=1}^d s_{j,k}^{n_{j,k}}
\right) 
\exp\left\{
-\sum_{j=1}^d\int_\mathbb{G} u_j w_j\Phi({\rm d}\bmw,{\rm d}z)
\right\}
\Phi^K({\rm d}\bms,{\rm d}\bmx)
\right] \\
&= \int_{\mathbb{G}^K}
\prod_{k=1}^K
\left( 
\delta_{x_k}({\rm d}\theta^{**}_k)\prod_{j=1}^d s_{j,k}^{n_{j,k}}
\right) \\ 
&  \hspace{1cm} \times
\E\left[
\exp\left\{
-\sum_{j=1}^d\int_\mathbb{G} u_j w_j
\left(\Phi^!_{\bms,\bmx} + \sum_{k=1}^K \delta_{(\bms_k,x_k)}\right)({\rm d}\bmw,{\rm d}z)
\right\}
\right] \ 
\text{M}_{\Phi^K}({\rm d}\bms,{\rm d}\bmx)  \\
& =
\bigintssss_{\mathbb{G}^K}
\prod_{k=1}^K
\left( 
\delta_{x_k}(d\theta^{**}_k)\prod_{j=1}^d s_{j,k}^{n_{j,k}}
\right) \exp\left\{-\sum_{j=1}^d\sum_{k=1}^K u_j s_{j,k}\right\} \\
&  \hspace{1cm} \times
\E\left[
\exp\bigl\{
-\sum_{j=1}^d\int_\mathbb{G} u_j w_j
\Phi^!_{\bms,\bmx}({\rm d}\bmw,{\rm d}z)
\bigr\}
\right] \ 
\E\left[M^{(K)}\right]H^K({\rm d}\bms)P_0^K({\rm d}\bmx) .
\end{split}
\end{equation}
Note that in the second equality of Equation \eqref{eqn:peppfproof_4}  we replaced $\text{M}_{\Phi^K}({\rm d}\bms,{\rm d}\bmx)$ with
\[
\text{M}_{\Phi^{(K)}}({\rm d}\bms,{\rm d}\bmx) = \E\left[M^{(K)}\right]H^K({\rm d}\bms)P_0^K({\rm d}\bms),
\]
because the sets ${\rm d} \theta_k^{**}$ are infinitesimally pairwise disjoint.
The higher order reduced Palm version $\Phi^!_{\bms,\bmx}$ of $\Phi$ at $(\bms,\bmx)$ has been defined in \Cref{cor:palm4VIFPP}. The integral with respect to $\G^K$ is now simple to compute as it factorizes in the product of the base measures evaluated in small sets centred over the distinct values, namely,
\begin{equation}
\label{eqn:peppfproof_4.2}
\begin{split}
&\E\left[
\int_{\mathbb{G}^K}
\prod_{k=1}^K
\left( 
\delta_{x_k}({\rm d}\theta^{**}_k)\prod_{j=1}^d s_{j,k}^{n_{j,k}}
\right) 
\exp\left\{
-\sum_{j=1}^d\int_\mathbb{G} u_j w_j\Phi({\rm d}\bmw,{\rm d}z)
\right\}
\Phi^K({\rm d}\bms,{\rm d}\bmx)
\right] \\
& = \prod_{k=1}^K P_0({\rm d}\theta^{**}_k)
\E\left[M^{(K)}\right] \int_{(\R^+)^{dK}}
\prod_{j=1}^d\prod_{k=1}^K
e^{-u_js_{j,k}}s_{j,k}^{n_{j,k}}  \\
&  \hspace{3cm} \times
\E\left[
\exp\left\{
-\sum_{j=1}^d\int_\mathbb{G} u_j w_j
\Phi^!_{\bms,\bmx}({\rm d}\bmw,{\rm d}z)
\right\}
\right]
\
\prod_{k=1}^K\prod_{j=1}^d H_j({\rm d}s_{j,k}).
\end{split}
\end{equation}
Then, we use \Cref{lemma:Laplace2} to compute the expected value appearing in the final line of Equation \eqref{eqn:peppfproof_4.2}. This allows us to eliminate the normalizing constant $\E\left[M^{(K)}\right]$. Moreover, we note that such a result does not depend on $\bms$.  Hence, plugging it into \eqref{eqn:peppfproof_3}, we get
\begin{equation}
\label{eqn:peppfproof_5}
\begin{split}
&  \E\left[
\prod_{j=1}^d\prod_{k=1}^K
P_j({\rm d}\theta^{**}_k)^{n_{j,k}}
\right] \\
& \quad =
\bigintsss_{[0,\infty]^d}
\prod_{j=1}^d
\left\{
\frac{u_j^{n_{j}-1}}{\Gamma(n_{j})} \ 
\prod_{k=1}^K\kappa_j\left({u_j,n_{j,k}}\right)
\right\}
\Psi(K,\bmu) \ 
{\rm d}\bmu \
\prod_{k=1}^K P_0({\rm d}\theta^{**}_k)
\end{split}
\end{equation}
where $\kappa_j\left({u_j,n_{j,k}}\right) = \int_0^\infty e^{-u_js_j}s_j^{n_{j,k}}H_j({\rm d}s_j)$. We can now exploit \eqref{eqn:peppfproof_5} to evaluate the pEPPF in \eqref{eq:pEPPF_def}, and the result easily follows by integrating out the
$\theta_k^{**}$.
\end{proof}

\subsection{Proof of Theorem \ref{thm:moments}}
\label{proof:moments}
We prove the three parts of the theorem separately.

\begin{proof}
\noindent (i) For any Borel set $A\in\mathcal{X}$, we have
\begin{equation*}
\label{eqn:momproof_i1}
\begin{split}
\E[P_j(A)] &= 
\E\left[\frac{\mu_j(A)}{\mu_j(\mathbb{X})}\right] =
\E\left[
\int_{0}^\infty e^{-u\mu_j(\mathbb{X})} {\rm d}u
\int_{\mathbb{G}}s_j\delta_x(A)\Phi({\rm d}\bms,{\rm d}x)
\right]  \\
&  = 
\int_{0}^\infty 
\E\left[
\int_{\mathbb{G}}
\exp\left\{
-\int_{\mathbb{G}}
uw_j\Phi({\rm d}\bmw,{\rm d}y)
\right\}
s_j\delta_x(A)
\Phi({\rm d}\bms,{\rm d}x)
\right] {\rm d}u.
\end{split}
\end{equation*}
The first equality follows by definition of $P_j(A)$, the second one exploits the definition of $\mu_j$ as well as the identity $1/x = \int_{0}^\infty e^{-ux}{\rm d}x$. For the third equality we use Fubini's theorem. 
Finally use the CLM formula, i.e., Equation \eqref{eqn:app_CLMformula}, with 
\[
f(\bms,x,\Phi) =         
\exp\left\{
-\int_{\mathbb{G}}
uw_j\Phi({\rm d}\bmw,{\rm d}y)
\right\}
s_j\delta_x(A).
\]
As a consequence, the previous expression becomes
\begin{equation}
\label{eqn:momproof_i2}
\begin{split}
\E[P_j(A)] &=\int_{0}^\infty 
\int_{\mathbb{G}}
s_j\delta_x(A)
\E\left[
\exp\left\{
-\int_{\mathbb{G}}
uw_j
\left(\Phi^!_{\bms,x} + \delta_{(\bms,x)}\right)({\rm d}\bmw,{\rm d}y)
\right\}
\right] \ 
\text{M}_{\Phi}({\rm d}\bms,{\rm d}x)
{\rm d}u  \\
&  =
\int_{0}^\infty 
\int_{\mathbb{G}}
e^{-us_j}s_j\delta_x(A)
\E\left[
\exp\left\{
-\int_{\mathbb{G}}
uw_j\Phi^!_{\bms,x}({\rm d}\bmw,{\rm d}y)
\right\}
\right] \ 
\E[M]H({\rm d}\bms)P_0({\rm d}x)
{\rm d}u .
\end{split}
\end{equation}
To conclude, note that $H({\rm d}\bms)=\prod_{l=1}^d H_l({\rm d}s_l)$ with only the $j$-th component that does not integrate to 1. Moreover the expected value in the final line of Equation \eqref{eqn:momproof_i2} can be computed using \Cref{lemma:laplace} with $f(\bmw,y)=uw_j$. Hence,
\begin{equation*}
\label{eqn:momproof_i3}
\begin{split}
\E[P_j(A)]  &= 
P_0(A)
\int_{0}^\infty 
\left(
\int_0^\infty
e^{-us_j}s_jH_j({\rm d}s_j)
\right)
\prod_{l\neq j}\left(
\int_0^\infty
H_l({\rm d}s_l)
\right) \\
& \quad \times 
\left(
\sum_{m=0}^\infty q_M(m+1)(m+1) (\psi_j(u))^m
\right)
{\rm d}u  \\
& = 
P_0(A)
\int_{0}^\infty 
\kappa_j(u,1)
\Psi(1,u)
{\rm d}u = P_0(A).
\end{split}
\end{equation*}
The last equality holds since the final integral is equal to the pEPPF given in \Cref{thm:peppf} when $d=1$, $n=1$ and $k=1$. Hence it is trivially equal to $1$.

\noindent (iii) 
Without loss of generality, we set $j=1$ and $l=2$. 
Leveraging on the integral identity 
$\frac{1}{x^n} = \int_0^\infty \frac{u^{n-1}}{\Gamma(n)}e^{-ux}{\rm d}x$, for any Borel sets $A_1,A_2\in\mathcal{X}$,  we get
\begin{equation}
\label{eqn:momproof_ii1_0}
\begin{split}
&   \E\left[P_1(A_1)^{n_1}P_2(A_2)^{n_2}\right]\\ 
& \qquad\qquad =
\E\left[
\frac{1}{\mu_1(\X)^{n_1}\mu_2(\X)^{n_2}}
\mu_1(A_1)^{n_1}\mu_2(A_2)^{n_2}
\right]  \\
&  \qquad\qquad = 
\E\left[
\int_{[0,\infty]^2}
\frac{u_1^{n_1-1}u_2^{n_2-1}}{\Gamma(n_1)\Gamma(n_2)}
e^{-(u_1\mu_1(\X) + u_2\mu_2(\X))}
\mu_1(A_1)^{n_1}\mu_2(A_2)^{n_2}
\right] .
\end{split}
\end{equation}
As done in previous calculations, we now want to expand the definitions of $\mu_j(A_j)^{n_j}$, for $j=1,2$, using Equation \eqref{eqn:app_random_measure_def}. 
More specifically, we need to express the $\mu_j(A_j)$s as  $n=n_1+n_2$ different integrals, which are explicitly introduced through these expressions
\begin{equation*}
\label{eqn:mom_def_integrationvariables}
\begin{split}
&\mu_1(A_1)^{n_1} = 
\left(\int_\G \delta_{v_1}(A_1)r_{1,1}\Phi(\mydiff \bmr_1, \mydiff v_1)\right)
\cdots
\left(\int_\G \delta_{v_{n_1}}(A_1)r_{n_{1},1}\Phi(\mydiff \bmr_{n_1}, \mydiff v_{n_1})\right),\\
&\mu_2(A_2)^{n_2} = 
\left(\int_\G \delta_{z_1}(A_2)t_{1,2}\Phi(\mydiff \bmt_1, \mydiff z_1)\right)
\cdots
\left(\int_\G \delta_{z_{n_2}}(A_2)t_{n_2,2}\Phi(\mydiff \bmt_{n_2}, \mydiff z_{n_2})\right),
\end{split}
\end{equation*}
where $\bmr_i = \left(r_{i,1},\dots,r_{i,d}\right)$ for $i=1,\dots,n_1$ and 
$\bmt_i = \left(t_{i,1},\dots,t_{i,d}\right)$ for $i=1,\dots,n_2$.
As a consequence, Equation \eqref{eqn:momproof_ii1_0} becomes
\begin{equation*}
\label{eqn:momproof_ii1}
\begin{split}
\E\left[P_1(A_1)^{n_1}P_2(A_2)^{n_2}\right]=&  
\E\Biggl[ \ 
\int_{[0,\infty]^2}
\frac{u_1^{n_1-1}u_2^{n_2-1}}{\Gamma(n_1)\Gamma(n_2)}
e^{-(u_1\mu_1(\X) + u_2\mu_2(\X))} \\
&\times \prod_{i=1}^{n_1}
\int_{\G}r_{i,1}\delta_{v_i}(A_1)
\Phi({\rm d}\bmr_{i},{\rm d}v_{i})  \prod_{l=1}^{n_2}
\int_{\G}t_{l,2}\delta_{z_{l}}(A_2)
\Phi({\rm d}\bmt_{l},{\rm d}z_{l})
\ {\rm d}\bmu,
\Biggr] 
\end{split}
\end{equation*}
where ${\rm d}\bmu = ({\rm d}u_1,{\rm d}u_2)$. 
We then collect all $n$ integration variables by defining
a vector $\bmx = (v_{1},\dots,v_{n_1},z_{1},\dots,z_{n_2})$ and a vector of vectors 
$\bms = (\bmr_{1},\dots,\bmr_{n_1},\bmt_{1},\dots,\bmt_{n_2})$.
Then, the  Fubini's theorem implies
\begin{equation*}
\label{eqn:momproof_ii2}
\begin{split}
\E[P_1&(A_1)^{n_1}P_2(A_2)^{n_2}]  =
\bigintssss_{[0,\infty]^2}
\frac{u_1^{n_1-1}u_2^{n_2-1}}{\Gamma(n_1)\Gamma(n_2)}  \\
& \times \ \E\left[
\bigintssss_{\G^n} 
e^{-(u_1\mu_1(\X) + u_2\mu_2(\X))} \delta_{\bmx}(A_1^{n_1}\times A_2^{n_2})
\left(\prod_{i=1}^{n_1}s_{i,1}\right)
\left(\prod_{i=n_1+1}^{n_1+n_2}s_{i,2}\right)
\Phi^n({\rm d}\bms,{\rm d}\bmx)
\right] 
{\rm d}\bmu.
\end{split}
\end{equation*}
where $A^{n} := A\times\dots\times A$ for $n$ times and $s_{i,j}$ represents the $j$-th component of the vector in the $i$-th position of $\bms$.

Differently from previous computations, it is not convenient to immediately apply the higher order CLM formula. Indeed, integrals are not defined over infinitely small sets but over possibly overlapping sets $A_1^{n_1}\times A_2^{n_2}$. As a consequence, we can not replace the $n$-power $\Phi^n$ with the $n$-factorial power $\Phi^{(n)}$ point process. 
Instead, we must compute the expected value using \cite[Lemma 14.E.4]{baccelli2020}
\begin{equation*}
\label{eqn:momproof_ii3}
\begin{split}
\E[P_1&(A_1)^{n_1}P_2(A_2)^{n_2}] = \bigintssss_{[0,\infty]^2}
\frac{u_1^{n_1-1}u_2^{n_2-1}}{\Gamma(n_1)\Gamma(n_2)}
\sum_{k=1}^n\sum_{(*)_k}\frac{1}{k!}
\prod_{j=1}^2 \binom{n_j}{n_{j,1},\dots,n_{j,k}}  \\
&  \hspace{1.5cm} \times
\E\left[
\bigintssss_{\G^k}
e^{-(u_1\mu_1(\X) + u_2\mu_2(\X))}
\prod_{m=1}^k
\left(
\delta_{y_m}(B_m)
\prod_{j=1}^2
w_{j,m}^{n_{j,m}}
\right)
\Phi^{(k)}({\rm d}\bmw,{\rm d}\bmy)
\right] 
{\rm d}\bmu ,
\end{split}
\end{equation*}
where we defined the set $B_m := \bigcap_{i=1}^{n_{1,m}}A_1 \cap \bigcap_{i=1}^{n_{2,m}}A_2$ and we are summing over all possible partitions of $k$ elements but, since the order of the elements does not matter, we only consider their counts. Namely, the set $(*)_k$ is defined as
\begin{equation*}
\label{eqn:momproof_ii4}
\begin{split}
(*)_k := \Bigl\{ & (n_{1,1}\dots,n_{1,k},n_{2,1}\dots,n_{2,k}) : 
n_{j,m} \geq 0 \text{ for }j=1,2\ m=1,\dots,k, \text{ and } \\
&\sum_{m=1}^k n_{j,m} = n_j \text{ for }j=1,2, \text{ and }
n_{1,m}+n_{2,m}\geq 1 \text{ for } m=1,\dots,k
\Bigr\}.
\end{split}
\end{equation*}
Consequently, note that the integration set $B_m$ is never empty but it is equal to $A_1\cap A_2$ if $n_{j,m}>0$ for both $j=1,2$ or it is equal to $A_1$ when $n_{2,m}=0$ or $A_2$ if $n_{1,m}=0$. 
We now apply the higher order CLM formula \eqref{eqn:app_reduced_highorder_CLMformula}, and we use the expression of the $k$-th factorial moment measure \eqref{eqn:app_factmeasure_vecIFPP} to get
\begin{equation}
\label{eqn:momproof_ii5}
\begin{split}
& \E[P_1(A_1)^{n_1}P_2(A_2)^{n_2}] \\
&\qquad\qquad= \bigintssss_{[0,\infty]^2}
\frac{u_1^{n_1-1}u_2^{n_2-1}}{\Gamma(n_1)\Gamma(n_2)}
\sum_{k=1}^n\sum_{(*)_k}
\frac{1}{k!}\prod_{j=1}^2 \binom{n_j}{n_{j,1},\dots,n_{j,k}} \\
& \qquad\qquad \hspace{1cm} \times \E\left[M^{(k)}\right]
\bigintssss_{\G^k}
\prod_{m=1}^k \delta_{y_m}(B_m)
\prod_{j=1}^2\prod_{m=1}^k
e^{u_jw_{j,m}}w_{j,m}^{n_{j,m}} \\
& \qquad\qquad\hspace{1cm} \times 
\E\left[
\exp\left\{
-\sum_{j=1}^2 \bigintssss_{\mathbb{G}} u_j s_j
\Phi^!_{\bmw,\bmy}({\rm d}\bms,{\rm d}z)
\right\}
\right]
P_0^{k}({\rm d}\bmy)H^k({\rm d}\bmw)
{\rm d}\bmu  \\
&\qquad\qquad = \sum_{k=1}^n
\left\{
\prod_{m=1}^k P_0(B_m)
\sum_{(*)_k} 
\frac{1}{k!}\prod_{j=1}^2 \binom{n_j}{n_{j,1},\dots,n_{j,k}}\right.
\\
& \qquad\qquad\hspace{1cm} \times
\left.\bigintssss_{[0,\infty]^2}
\frac{u_1^{n_1-1}u_2^{n_2-1}}{\Gamma(n_1)\Gamma(n_2)}
\prod_{j=1}^2\prod_{m=1}^k
\kappa_j(u_j,n_{j,m})
\Psi(2,\bmu)
{\rm d}\bmu
\right\} .
\end{split}
\end{equation}
To conclude, we recognize that the final integral is the pEPPF when $d=2$, $K_{(n_1,n_2)}=k$ and counts $\bmn_j=(n_{j,1},\dots,n_{j,k})$, for each $j=1,2$. 
Moreover, note that in the special case $A_1 = A_2 = A \in \mathcal{X}$, we have $B_m=A$, in particular  $B_m$ does not depend on $m$.
In such a case, we have
\begin{equation}
\label{eqn:momproof_ii6}
\begin{split}
\sum_{k=1}^n
P_0(A)^k
\sum_{(*)_k} 
\frac{1}{k!}\prod_{j=1}^2 \binom{n_j}{n_{j,1},\dots,n_{j,k}}
\Pi_k^{(n)}(\bmn_1,\bmn_2) = 
\sum_{k=1}^{n_1+n_2}
P_0(A)^k
\P(K_{(n_1,n_2)}=k),
\end{split}
\end{equation}
thus,  (ii) easily follows from \eqref{eqn:momproof_ii6}.
\end{proof}

\subsection{Correlation}\label{app:correlation}
Let $A,B \in \mathcal{X}$. The correlation between $P_j(A)$ and $P_l(B)$ can be computed as
\begin{equation}
\text{corr}(P_j(A),P_l(B)) \ = \ 
\frac{\E\Bigl[P_j(A)P_l(B)\Bigr] - \E\Bigl[P_j(A)\Bigr]\E\Bigl[P_l(B)\Bigr]}{\sqrt{\text{Var}(P_j(A))\text{Var}(P_l(B))}}.
\label{eqn:corr_def}
\end{equation}
The numerator in Equation \eqref{eqn:corr_def} easily follows from \Cref{thm:moments}. As for the denominator, we use the following result from \cite{cremaschi2020},
\begin{equation*}
\text{Var}(P_j(A)) = \P\left(K_{j,(2)}=1\right)P_0(A)(1-P_0(A)).
\label{eqn:varianza_cremaschi}
\end{equation*}
If we plug the previous expressions in \eqref{eqn:corr_def}, we obtain
\begin{equation}
\label{eqn:corr_general}
\begin{split}
&   \text{corr} \left( P_j(A),P_l(B)\right)  \\
& \qquad =
\frac{\mathbb{P}\left(K_{(1,1)} = 1\right)
\left( \ P_0(A\cap B) - P_0(A)P_0(B) \ \right)}
{\sqrt{\mathbb{P}\left(K_{j,(2)} = 1\right)\mathbb{P}\left(K_{l,(2)} = 1\right)P_0(A)(1-P_0(A))P_0(B)(1 - P_0(B))}}.
\end{split}
\end{equation}
In particular, Equation \eqref{eqn:corr} in the main paper follows from Equation \eqref{eqn:corr_general} when $A=B$.

\subsection{Coskewness} \label{app:coskewness}
Consider two random variables, $X$ and $Y$, both having finite third moments. Let $\mu_X$ denote the mean of $X$, $\sigma_X^2$ its variance, and define the skewness of $X$ as $\operatorname{Sk}(X) = \frac{\mathbb{E}\left[(X-\mu_X)^3\right]}{\sigma_X^3}$. The same definitions and notation can be applied to $Y$.
The coskewness of $X$ over $Y$ is defined as
\begin{equation*}
\begin{split}
\operatorname{CoSk}(X, Y)=
\frac{\mathbb{E}\left(\left(X-\mu_X\right)^2\left(Y-\mu_Y\right)\right)}{\sigma_X^2 \sigma_Y} & =
\frac{\mathbb{E}\left(X^2 Y\right)-2 \mu_X \mathbb{E}(X Y)-\mathbb{E}\left(X^2\right) \mu_Y+2 \mu_X^2 \mu_Y}{\sigma_X^2 \sigma_Y} \\
& =\frac{\operatorname{Cov}\left(X^2, Y\right)-2 \mu_X \operatorname{Cov}(X, Y)}{\sigma_X^2 \sigma_Y}
\end{split}
\end{equation*}
We note that the coskewness operator $\operatorname{CoSk}$ is not symmetric, and the coskewness of $Y$ over $X$ is defined analogously. Additionally, the coskewness can deviate from zero even when $X$ and $Y$ are symmetric, as it depends on the random variable $X+Y$. This mixed third moment serves as an indicator of the joint asymmetry of $X$ and $Y$ (see, for instance, \citet{friend1980, fang1997}). In the field of Econometrics, the coskewness is frequently employed to study the risk associated with financial portfolios.
In our setting, if we let $X=P_j(A)$ and $Y=P_l(A)$.
From \Cref{thm:moments} we have that both $\mu_X$ and $\mu_Y$ equal $P_0(A)$. Then, using Equations \eqref{eqn:mixedmom1} and \eqref{eqn:mixedmom2} in the paper we get
\begin{equation}
\label{eqn:cosk}
\begin{split}
&\text{CoSk}\left(P_j(A),P_l(A)\right) 
= \frac{1}{\left(\mathbb{P}(K_{(2)} = 1)P_0(A)(1-P_0(A))\right)^{3/2}}  \\
& \quad \times \left(
\E[P_0(A)^{K_{(2,1)}}]   - P_0(A)
\left(
\E[P_0(A)^{K_{(2)}}] - 2P_0(A)\left(\mathbb{P}(K_{(2)} = 1) + \mathbb{P}(K_{(1,1)} = 2)
\right)\right)
\right).
\end{split}
\end{equation}
A more general formulation of \eqref{eqn:cosk} for $\text{CoSk}\left(P_j(A),P_l(B)\right)$ could also be computed using Equation \eqref{eqn:momproof_ii5}.

\subsection{Proofs of Equations \eqref{eqn:corr_Dir} and \eqref{eqn:corr_limits}}\label{app:corr_limiting}

First, we  prove Equation \eqref{eqn:corr_Dir}.
Without loss of generality, we set $j=1$ and $l=2$. Then, we recall that when $A=B$, Equation \eqref{eqn:corr_general} boils down to
\begin{equation}
\label{eqn:app_corr_iniziale}
\text{corr} \left( P_1(A),P_2(A)\right)  \ = \  \frac{\P\left(K_{(1,1)} = 1\right)}
{\sqrt{\P\left(K_{1,(2)} = 1\right)\mathbb{P}\left(K_{2,(2)} = 1\right)}},
\end{equation}
where both the numerator and the denominator admit an explicit representation in terms of pEPPF and EPPF, respectively.
We first focus on the numerator in \eqref{eqn:app_corr_iniziale}. Since we are dealing with Vec-FDP, we can exploit the results of \Cref{section:HMFM_properties} to obtain 
\begin{equation}
\label{eqn:peppf_K11}
\begin{split}
\P\left(K_{(1,1)} = 1\right) &= \Pi_1^{(2)}\left(1,1\right)  \\
&= \int_{[0,\infty]^2}\left(1+\frac{\Lambda}{(1+u_1)^{\gamma_1}(1+u_2)^{\gamma_2}}\right)  \\
& \times
\exp\left\{-\Lambda\left(1 - \frac{1}{(1+u_1)^{\gamma_1}(1+u_2)^{\gamma_2}}\right)\right\}
\frac{\gamma_1\gamma_2}{(1+u_1)^{1+\gamma_1}(1+u_2)^{1+\gamma_2}}{\rm d}u_1{\rm d}u_2 .
\end{split}
\end{equation}
Applying the following change of variables $x_j = \frac{1}{(1+u_j)^{\gamma_j}}$, for $j=1,2$, Equation \eqref{eqn:peppf_K11} equals 
\begin{equation*}
\label{eqn:corr_numerator1}
\begin{split}
\P\left(K_{(1,1)} = 1\right) &= 
\int_{[0,1]^2} (1+\Lambda x_1 x_2)
e^{-\Lambda(1-x_1x_2)}{\rm d}x_1{\rm d}x_2 \\
&  = 
e^{-\Lambda} \left(\int_{[0,1]^2} e^{\Lambda x_1x_2}{\rm d}x_1{\rm d}x_2 +
\Lambda \int_{[0,1]^2}x_1x_2 e^{\Lambda x_1x_2}{\rm d}x_1{\rm d}x_2 \right).
\end{split}
\end{equation*}
Integrand functions are positive, hence we use Fubini's theorem to compute the bi-dimensional integral as two iterative integrals. In particular, we first integrate by parts with respect to $x_2$ to get,
\begin{equation}
\label{eqn:corr_numerator2}
\begin{split}
\P\left(K_{(1,1)} = 1\right) &=
e^{-\Lambda} 
\left(\int_0^1 \frac{1}{\Lambda x_1} (e^{\Lambda x_1}-1){\rm d}x_1 +
\int_0^1 e^{\Lambda x_1} {\rm d}x_1 -
\int_0^1 \frac{1}{\Lambda x_1}(e^{\Lambda x_1} - 1)
{\rm d}x_1 \right)  \\
& = 
e^{-\Lambda} \int_0^1 e^{\Lambda x_1}{\rm d}x_1 = 
\frac{1 - e^{-\Lambda}}{\Lambda}.
\end{split}
\end{equation}
One can exploit similar arguments to compute  the denominator $\P\left(K_{1,(2)} = 1\right)$ in \eqref{eqn:app_corr_iniziale}.
Indeed, we note that, marginally, $P_1$ and $P_2$ are Finite Dirichlet Processes of \cite{argiento2022annals}. Then,
\begin{equation}
\label{eqn:peppf_K12}
\begin{split}
\P\left(K_{j,(2)} = 1\right) = \Pi_1^{(2)}\left(2\right) = 
\int_{0}^\infty
\left(1+\frac{\Lambda}{(1+u)^{\gamma_j}}\right)
e^{-\Lambda\left(1 - \frac{1}{(1+u)^{\gamma_j}}\right)}
\frac{u\gamma_j(1+\gamma_j)}{(1+u)^{2+\gamma_j}}{\rm d}u,
\end{split}
\end{equation}
for $j=1,2$. The change of variable, $x_j = \frac{1}{(1+u_j)^{\gamma_j}}$, for $j=1,2$, is applied to Equation \eqref{eqn:peppf_K12} which boils down to
\begin{equation}
\label{eqn:corr_denominator1}
\begin{split}
\P\left(K_{j,(2)} = 1\right) &= 
(\gamma_j + 1)e^{-\Lambda}
\int_0^1
(1 + \Lambda x)e^{\Lambda x}(1 - x^{1/\gamma_j}) {\rm d}x.
\end{split}
\end{equation}
A similar formula holds true for the probability $\P\left(K_{j,(1)} = 1\right)$.
Substituting  the expressions \eqref{eqn:corr_numerator2} and \eqref{eqn:corr_denominator1} in \eqref{eqn:app_corr_iniziale}, we get Equation \eqref{eqn:corr_Dir}.

We now focus on the proof of Equation \eqref{eqn:corr_limits} to obtain the limit of the correlation in \eqref{eqn:corr_Dir} as $\gamma_j$ and $\gamma_l$ go to either $0$ or $\infty$, simultaneously. 
We note  that the integral $I(\gamma_j , \Lambda)$ in \eqref{eqn:corr_Dir}
does not admit an analytical solution, hence we need to study the limiting cases for small and large values of $\gamma_j$.\\
Firstly, when $\gamma_j \rightarrow 0$, we use the dominated convergence theorem to show that 
\begin{equation*}
\label{eqn:corr_limgamma0}
\begin{split}
\lim_{\gamma_j \rightarrow 0 }    I(\gamma_j , \Lambda) & =  \lim_{\gamma_j \rightarrow 0 }  \int_0^1
(1 + \Lambda x)e^{-\Lambda (1-x)}(1 - x^{1/\gamma_j}) {\rm d}x \\
& =
e^{-\Lambda} \int_0^1
(1 + \Lambda x)e^{\Lambda x} {\rm d}x = 1.
\end{split}
\end{equation*}
The final equality has been obtained integrating by parts after noting that 
$\int (1+\Lambda x)e^{\Lambda x}dx = x e^{\Lambda x} + c$. Hence, we can evaluate the limit in \eqref{eqn:corr_Dir} to obtain:
\[
\lim_{\gamma_j,\gamma_l \rightarrow 0 }
\text{corr}(P_j(A), P_l(A)) = \frac{1 - e^{-\Lambda}}{\Lambda}.
\]
The limiting case for $\gamma_j , \gamma_l\rightarrow \infty$ is obtained similarly. Once again, we first exchange the limit and the integral leveraging on the dominated convergence theorem and then we integrate by parts as before to obtain,
\begin{equation}
\label{eqn:corr_limgammainf}
\begin{split}
\lim_{\gamma_j \rightarrow \infty } (\gamma_j + 1)I(\gamma_j , \Lambda)  &= 
\lim_{\gamma_j \rightarrow \infty } (\gamma_j + 1)e^{-\Lambda}
\int_0^1
(1 + \Lambda x)e^{\Lambda x}(1 - x^{1/\gamma_j}) {\rm d}x \\
&=
-e^{-\Lambda} \int_0^1
(1 + \Lambda x)e^{\Lambda x}\log(x) {\rm d}x = 
\frac{1-e^{-\Lambda}}{\Lambda}.
\end{split}
\end{equation}
Hence, the limit of \eqref{eqn:corr_Dir} can be evaluate on the basis of  \eqref{eqn:corr_limgammainf}, to get
\[
\lim_{\gamma_j,\gamma_l \rightarrow \infty }
\text{corr}(P_j(A), P_l(A)) = 1.
\]

\subsection{Graphical representation of the correlation function}

We now provide a graphical visualization of the correlation function. Equation \eqref{eqn:corr_Dir} depends on three parameters, $\gamma_1$, $\gamma_2$ and $\Lambda$. For graphical convenience, we always refer to the case when $\gamma_1$ and $\gamma_2$ are set to a common value, namely $\gamma$. 
In the left panel of \Cref{fig:cor_main} in the main manuscript, we show the correlation function in Equation \eqref{eqn:corr_Dir} evaluated on a grid of values of $\gamma$  and fixed values of $\Lambda$.
Then, in the right panel of the same figure, we do the opposite, evaluating the correlation over a grid of values of $\Lambda$ and some fixed values of $\gamma$. In the first case, curves are monotonically increasing functions of $\gamma$, whose lowest value is given in the first row of Equation \eqref{eqn:corr_limits}. 
They rapidly increase for small values of $\gamma$ and then slowly tend to their limiting value, that is $1$, as shown in the second row of Equation \eqref{eqn:corr_limits}.
The right panel of \Cref{fig:cor_main} in the main manuscript depicts a function that monotonically decreases as $\Lambda$ increases.
In particular, we note that the curve representing $\gamma=1$ is much higher than the other ones, showing how crucial is this parameter in determining the value of the correlation.
To conclude, \Cref{fig:corr_3d} shows the case when both $\Lambda$ and $\gamma$ are free to change.


\begin{figure}[!ht]
\centering
\includegraphics[width=0.5\textwidth]{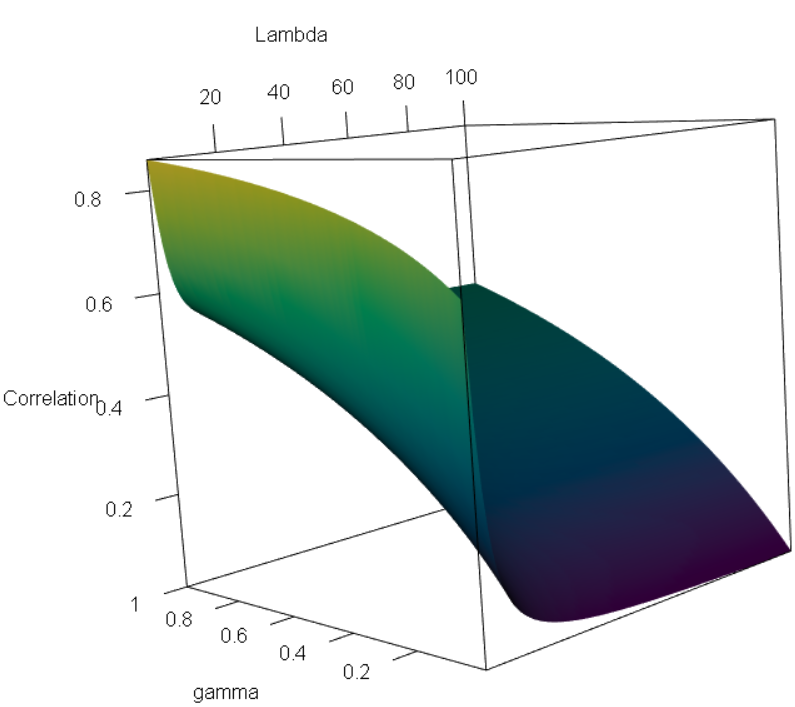}
\caption{ Correlation function for $\gamma_1 = \gamma_2$. 
}
\label{fig:corr_3d}
\end{figure}

\subsection{Proof of Theorem \ref{thm:posterior}}
\label{proof:posterior}

We now introduce useful notation and a vector of additional variables 
$\bmU_n = \left(U_{1},\dots,U_{d}\right)$ which helps us to describe Bayesian inference for our model. 
Let $\bmtheta = \left(\bmtheta_1,\dots,\bmtheta_d\right)$ be the collection of all random variables $\theta_{j,i}$ across all possible groups, for $j=1,\dots,d$ and $i=1,\dots,n_j$.\\
The conditional distribution of $(\bmtheta \mid \mu_1,\dots,\mu_d)$ is given by
\begin{equation}
\mathbb{P}(\bmtheta \in d\bmtheta \mid \mu_1,\dots,\mu_d) = 
\prod_{j=1}^d\frac{1}{T_j^{n_j}} \prod_{i=1}^{n_j} \mu_j\left({\rm d} \theta_{j,i}\right),    
\label{eqn:introU_1}
\end{equation}
where $T_j=\mu_j(\mathbb{X})$. 
Similarly to \citep{jlp2009,argiento2022annals}, 
we introduce a vector of auxiliary variables $\bmU_n=\left(U_{1},\dots,U_{d}\right)$ such that $U_{j} \mid T \ind \operatorname{Gamma}(n_j, T_j)$. 
Leveraging on the integral identity 
$\frac{1}{x^n} = \int_0^\infty \frac{u^{n-1}}{\Gamma(n)}e^{-ux}{\rm d}x$
we introduce a suitable augmentation of the underlying probability space and consider the joint conditional distribution of 
$\left(\bmtheta, \bmU_n \mid \mu_1,\dots,\mu_d\right)$
\[
\mathbb{P}\left(\bmtheta \in {\rm d}\bmtheta, \bmU_n \in {\rm d}\bmu \mid \mu_1,\dots,\mu_d \right) =
\prod_{j=1}^d\frac{u_j^{n_j-1}}{\Gamma(n_j)} e^{-u_jT_j }{\rm d}\bmu
\prod_{i=1}^{n_j} \mu_j\left({\rm d}\theta_{j,i}\right).
\]
As explained in \Cref{subsection:clustering}, since $(\mu_1,\dots,\mu_d)$ are almost surely discrete, with positive probability there will be ties
within the sample $\bmtheta$. For this reason, $\bmtheta$ is equivalently characterized by the couple $(\bmtheta^{**},\rho)$, introduced in section \Cref{subsection:clustering}. Moreover, since the order is not relevant, we recall that $\rho$ is further characterized given its number of clusters $K_{(n_1,\dots,n_d)}=k$ and their counts $(\bmn_1,\dots,\bmn_k)$. As a consequence, we may write
\begin{equation}  
\mathbb{P}\left(\bmtheta \in {\rm d}\bmtheta, \bmU_n \in {\rm d}\bmu \mid \mu_1,\dots,\mu_d \right) =
\prod_{j=1}^d\frac{u_j^{n_j-1}}{\Gamma(n_j)} e^{-u_jT_j }{\rm d}\bmu
\prod_{k=1}^{K} \mu_j\left({\rm d}\theta^{**}_k\right)^{n_{j,k}}.
\label{eqn:introU_2}
\end{equation}
The augmentation of \eqref{eqn:introU_1} through $\bmU_n$ is made possible since the marginal distribution of $\bmU_n$ exists. Indeed, for each component, say $j$, it holds that
\begin{equation*}
\begin{split}
\P\left(U_j\in {\rm d}u_j\right) &= 
\E\left[\P\left(U_j\in {\rm d}u_j\mid T_j\right)\right] = 
\E\left[\frac{u_j^{n_j-1}}{\Gamma(n_j)}T_j^{n_j}e^{-T_ju_j}{\rm d}u_j\right]  \\
& = \frac{u_j^{n_j-1}}{\Gamma(n_j)}{\rm d}u_j\int_0^\infty t^{n_j}e^{-tu_j}f_{T_j}(t){\rm d}t,
\end{split}
\end{equation*}
which is well defined since $T_j$ is, almost surely, a sum of a finite number of positive random variables hence it is finite and admits density $f_{T_j}$. Finally, leveraging on the independence of the components of $\bmU_n$, it follows that it 
has a density, with respect to the Lebesgue measure, that is
\begin{equation}
\label{eqn:marginal_U}
f_{\bmU_n}(\bmu)=\prod_{j=1}^d\frac{u_j^{n_j-1}}{\Gamma(n_j)} \int_0^{+\infty} 
t^{n_j}e^{-t_j u_j} f_{T_j}(t_j) {\rm d}t_j.
\end{equation}

To simplify the notation, we drop subscript $n$ for random variables $\bmU_n$ and simply write $\bmU = (U_1,\dots,U_d)$. Nevertheless, we recall that such a random vector depends on the group sizes $(n_1,\dots,n_d)$. 

\begin{proof}
The strategy to prove \Cref{thm:posterior} consists in describing the Laplace functional of $(\mu_1,\dots,\mu_d)$ given $\bmtheta$ and $\bmU$.
Let $f_j:\X\rightarrow\R^+$ be $\mathcal{X}$ measurable for each $j\in\{1,\dots,d\}$. Then, following \citet{jlp2009,nRM2023}, we have
\begin{equation}
\label{eqn:postproof_1}
\begin{split}
\E\Biggl[ &
\exp\left\{
-\sum_{j=1}^d \int_\X f_j(x)\mu_j({\rm d}x)
\right\}
\mid \bmtheta\in {\rm d}\theta, \bmU\in {\rm d}\bmu 
\Biggr]  \\
& = 
\frac{
\E\left[
\exp\left\{
-\sum_{j=1}^d \int_\X f_j(x)\mu_j({\rm d}x)
\right\}
\P\left(\bmtheta\in {\rm d}\theta, \bmU\in {\rm d}\bmu\mid \mu_1,\dots,\mu_d\right)
\right]
}{
\E\left[
\P\left(\bmtheta\in {\rm d}\theta, \bmU\in {\rm d}\bmu\mid \mu_1,\dots,\mu_d\right)
\right]
}.
\end{split}
\end{equation}
We now focus on the numerator of \eqref{eqn:postproof_1} since the denominator will follow by setting $f_j = 0$ for each $j$. 
We note that Equation \eqref{eqn:introU_2} equals the probability appearing in the expected value. Then, the numerator under consideration equals to
\begin{equation*}
\label{eqn:postproof_2}
\begin{split}
\E\Biggl[ & 
\exp\left\{
-\sum_{j=1}^d \int_\X f_j(x)\mu_j({\rm d}x)
\right\}
\P\left(\bmtheta\in {\rm d}\theta, \bmU\in {\rm d}\bmu\mid \mu_1,\dots,\mu_d\right)
\Biggr] \\
& = \E\left[
\exp\left\{
-\sum_{j=1}^d \int_\X f_j(x)\mu_j({\rm d}x)
\right\}
\prod_{j=1}^d
\left\{
\frac{u_j^{n_j-1}}{\Gamma(n_j)} e^{-u_jT_j }{\rm d}\bmu
\prod_{k=1}^{K} \mu_j\left({\rm d}\theta^{**}_k\right)^{n_{j,k}}
\right\}
\right]  \\
& =       \prod_{j=1}^d\left(\frac{u_j^{n_j-1}}{\Gamma(n_j)}\right){\rm d}\bmu \ 
\E\left[
\exp\left\{
-\sum_{j=1}^d \int_\X \left(u_j+f_j(x)\right)\mu_j({\rm d}x)
\right\}
\prod_{j=1}^d \prod_{k=1}^{K} 
\mu_j\left({\rm d}\theta^{**}_k\right)^{n_{j,k}}
\right].
\end{split}
\end{equation*}
Reasoning as for \eqref{eqn:peppfproof_2}, the previous expression equals
\begin{equation}
\label{eqn:postproof_3}
\begin{split}
\E\Biggl[ & 
\exp\left\{
-\sum_{j=1}^d  \int_\X f_j(x)\mu_j({\rm d}x)
\right\}
\P\left(\bmtheta\in {\rm d}\theta, \bmU\in {\rm d}\bmu\mid \mu_1,\dots,\mu_d\right)
\Biggr] =\prod_{j=1}^d\left(\frac{u_j^{n_j-1}}{\Gamma(n_j)}\right){\rm d}\bmu \ \times\\
& \E\left[
\int_{\G^k}
\prod_{k=1}^K
\left( \delta_{y_k}({\rm d}\theta^{**}_k)\prod_{j=1}^d s_{j,k}^{n_{j,k}} \right)  
\exp\left\{
-\sum_{j=1}^d \int_\G w_j\left(u_j+f_j(x)\right)\Phi({\rm d}\bmw,{\rm d}x)
\right\}
\Phi^K({\rm d}\bms,{\rm d}\bmy)
\right].
\end{split}
\end{equation}
Note that \eqref{eqn:postproof_3} has the same form of \eqref{eqn:peppfproof_3}. Therefore, applying the higher order CLM formula we get
\begin{equation}
\label{eqn:postproof_4}
\begin{split}
\E\Biggl[ & 
\exp\left\{
-\sum_{j=1}^d  \int_\X f_j(x)\mu_j({\rm d}x)
\right\}
\P\left(\bmtheta\in {\rm d}\theta, \bmU\in {\rm d}\bmu\mid \mu_1,\dots,\mu_d\right)
\Biggr] \\
& = \prod_{j=1}^d
\left\{
\left(\frac{u_j^{n_j-1}}{\Gamma(n_j)}\right){\rm d}\bmu 
\prod_{k=1}^K
\int_{0}^\infty
e^{-s_j\left(f_j(\theta^{**}_k)+u_j\right)}s^{n_{j,k}}H_j({\rm d}s_j)
\right\}  \\
& \quad \times
\E\left[
\exp\left\{
-\sum_{j=1}^d \int_\G w_j\left(u_j+f_j(\theta^{**}_k)\right)\Phi^!_{\bms,\bmy}({\rm d}\bmw,{\rm d}x)
\right\}
\right] \
\E\left[M^{(K)}\right]
\prod_{k=1}^K P_0({\rm d}\theta^{**}_k).
\end{split}
\end{equation}
As a consequence, the denominator of \eqref{eqn:postproof_1} equals 
\begin{equation}
\label{eqn:postproof_5}
\begin{split}
\E\Bigl[&
\P(\bmtheta  \in {\rm d}\theta, \bmU\in {\rm d}\bmu\mid \mu_1,\dots,\mu_d )
\Bigr] = \prod_{j=1}^d\left\{
\frac{u_j^{n_j-1}}{\Gamma(n_j)}
\prod_{k=1}^K \kappa_j(u_j,n_{j,k})
\right\}  \\
& \times \quad
\E\left[
\exp\left\{
-\sum_{j=1}^d \int_\G w_j u_j\Phi^!_{\bms,\bmy}({\rm d}\bmw,{\rm d}x)
\right\}
\right]\E\left[M^{(K)}\right]\prod_{k=1}^K P_0({\rm d}\theta^{**}_k).
\end{split}
\end{equation}
Taking the ratio between \eqref{eqn:postproof_4} and \eqref{eqn:postproof_5}, we obtain
\begin{equation}
\label{eqn:postproof_6}
\prod_{j=1}^d
\prod_{k=1}^K \int_0^\infty
\frac{
e^{-s_j f_j(\theta^{**}_k)}e^{-s_j u_j}s_j^{n_{j,k}}
}{
\kappa_j(u_j,n_{j,k})
}H_j({\rm d}s_j)  
\frac{
\E\left[
\exp\left\{
-\sum_{j=1}^d \int_\G w_j\left(u_j+f_j(x)\right)\Phi^!_{\bms,\bmy}({\rm d}\bmw,{\rm d}x)
\right\}
\right]
}{
\E\left[
\exp\left\{
-\sum_{j=1}^d \int_\G w_ju_j\Phi^!_{\bms,\bmy}({\rm d}\bmw,{\rm d}x)
\right\}
\right]
},
\end{equation}
which is the Laplace functional of two independent processes, conditionally to $\bmU$, $K_{(n_1,\dots,n_d)}$ and $\bmtheta^{**}$. 
To conclude the proof of the theorem, it is enough to show that the Laplace functionals of the random measures described in the statement coincide with those in \eqref{eqn:postproof_6}.

\noindent (i) We first show that the first term in Equation \eqref{eqn:postproof_6} is the Laplace functional of the vector of random measures $\left(\mu_1^{(a)},\dots,\mu_d^{(a)}\right)$, where 
$\mu_j^{(a)}=\sum_{k=1}^K S_{j,k}^{(a)}\delta_{\theta^{**}_k}$. We refer to the main manuscript for the definition of random variables $S_{j,k}^{(a)}$, for $j=1,\dots,d$ and $k=1,\dots,K$.
The Laplace transform of $(\mu^{(a)}_1,\dots,\mu^{(a)}_d)$ equals 
\begin{equation}
\label{eqn:postproof_10}
\begin{split}
\E\Biggl[&
\exp \left\{
-\sum_{j=1}^d \int_{\X} f_j(x)\mu^{(a)}_j(dx)
\right\} \mid \bmtheta\in {\rm d}\theta, \bmU\in {\rm d}\bmu
\Biggr]  \\
&  = \E\left[ \prod_{j=1}^d
\exp \left\{
-\sum_{k=1}^K S^{(a)}_{j,k}f_j(\theta^{**}_k)
\right\} \mid \bmtheta\in {\rm d}\theta, \bmU\in {\rm d}\bmu
\right]  \\
&  = \E\left[ 
\prod_{j=1}^d\prod_{k=1}^K
\exp \left\{
- S^{(a)}_{j,k}f_j(\theta^{**}_k)
\right\} \mid \bmtheta\in {\rm d}\theta, \bmU\in {\rm d}\bmu
\right]  \\
&  = 
\prod_{j=1}^d\prod_{k=1}^K \int_0^\infty
\frac{e^{s_j\left(u_j + f_j(\theta^{**}_k)\right)}s_j^{n_{j,k}}}{
\kappa_j(u_j,n_{j,k})
}H_j({\rm d}s_j).
\end{split}
\end{equation}
The first equality exploits the fact that the atoms of each of the $\mu_j^{(a)}$ are non-random. The second equality holds thanks to the independence between the unnormalized jumps and finally for the third equality we computed the Laplace transform of of random variables $S^{(a)}_{j,m}$. The final term of Equation \eqref{eqn:postproof_10} is exactly the same term appearing in Equation \eqref{eqn:postproof_6}, which concludes the proof.

\noindent (ii) We now show that the second term in Equation \eqref{eqn:postproof_6} is the Laplace functional of the vector of random measures $\left(\mu_1^{(na)},\dots,\mu_d^{(na)}\right) \sim \operatorname{Vec-IFPP}(q^*_M,H^*,P_0)$. An explicit formulation of $q^*_M$ and $H^*$ is here provided. They equal 
\begin{equation*}
\label{eqn:qstar_Hstar}
\begin{split}
H^*_j(\mydiff s_j) &\ = \
\frac{e^{-u_js_j}H_j(\mydiff s_j)}{\psi_j(u_j)}, \text{ for }j=1,\dots,d \\
q_M^*(m) &\ = \  \frac{
q_M(m+k)\frac{(m+k)!}{m!}\prod_{j=1}^d\left(\psi_j(u_j)\right)^m
}{
\sum_{h=0}^\infty q_M(h+k)\frac{(h+k)!}{h!}\prod_{j=1}^d\left(\psi_j(u_j)\right)^h},
\end{split}
\end{equation*}
where $\psi_j(u_j)=\int_0^\infty e^{-u_js_j}H_j(\mydiff s_j)$.
To compute the Laplace transform of $\left(\mu_1^{(na)},\dots,\mu_d^{(na)}\right)$, we use \Cref{lemma:laplace4mu} and we get
\begin{equation}
\label{eqn:postproof_7}
\begin{split}
\E\Biggl[ &
\exp \left\{
-\sum_{j=1}^d \int_{\X} f_j(x)\mu^{(na)}_j({\rm d}x)
\right\} \mid \bmtheta\in {\rm d}\theta, \bmU\in {\rm d}\bmu
\Biggr] = \\
&  = 
\sum_{m=0}^\infty
q^*_M(m)
\left(
\int_\G
\exp\left\{ \sum_{j=1}^d -s_jf_j(x)\right\}
P_0({\rm d}x)
\prod_{j=1}^d
H^*_j({\rm d}s_j)
\right)^m\\
& = 
\sum_{m=0}^\infty
q^*_M(m)
\left(
\prod_{j=1}^d
\int_{0}^\infty\int_\X
\frac{
\exp \left\{-s_j\left(f_j(x)+u_j\right)\right\}
}{
\psi_j(u_j)
}
H_j({\rm d}s_j)P_0({\rm d}x)
\right)^m,   
\end{split}
\end{equation}
where the second equality follows by the definition of $H^*_j$.

The Laplace transform derived in Equation \eqref{eqn:postproof_7} must coincide with the second term of  Equation \eqref{eqn:postproof_6}. To show this, we focus on the expected value in the numerator of the second term in Equation \eqref{eqn:postproof_6}, since the denominator will follow  by setting $f_j = 0$.
We use \Cref{lemma:laplace} for $\Phi^!_{\bms,\bmy}$ with $f(\bmw,x)=\sum_{j=1}^d w_j(u_j+f_j(x))$ to get 
\begin{equation}
\label{eqn:postproof_8}
\begin{split}
&\E\Biggl[
\exp\left\{
-\sum_{j=1}^d \int_\G w_j\left(u_j+f_j(x)\right)\Phi^!_{\bms,\bmy}({\rm d}\bmw,{\rm d}x)
\right\}
\Biggr] \\
& \quad =
\sum_{m=0}^\infty
q_M(m+k)\frac{(m+k)!}{m!}\frac{1}{\E\left[M^{(K)}\right]}
\left(
\prod_{j=1}^d \int_0^\infty\int_\X
e^{-w_j(u_j+f_j(x))}H_j({\rm d}w_j)P_0({\rm d}x)
\right)^m.
\end{split}
\end{equation}
Therefore, the denominator equals to
\begin{equation}
\label{eqn:postproof_9}
\begin{split}
\E\left[
\exp\left\{
-\sum_{j=1}^d \int_\G w_ju_j\Phi^!_{\bms,\bmy}({\rm d}\bmw,{\rm d}x)
\right\}
\right] 
= \sum_{m=0}^\infty q_M(m+k)\frac{(m+k)!}{m!}\frac{1}{\E[M^{(K)}]}
\left(
\prod_{j=1}^d \psi_j(u_j)
\right)^m.
\end{split}
\end{equation}
To conclude, we multiply and divide the numerator by $        \prod_{j=1}^d\left(\psi_j(u_j)\right)^m$ and then take the ratio between equations \eqref{eqn:postproof_8} and \eqref{eqn:postproof_9}, simplifying common terms. The identity with Equation \eqref{eqn:postproof_7} follows after recognizing the definition of $q^*_M(m)$, which has been stated above.
\end{proof}

\subsection{Full conditional distribution of $ \mathbf{U}_n$}\label{app:full_U}
As a byproduct of Theorem \ref{thm:posterior}, we also obtain the full conditional distribution of vector $\bmU_n$. 
\begin{corollary}
\label{cor:fullU}
The conditional distribution of $\bmU_n$ given a realization of $(\theta_1,\dots,\theta_d)$ in $K$ distinct values $(\theta^{**}_1,\dots,\theta^{**}_K)$ with counts $(\bmn_1,\dots,\bmn_d)$ depends on the distinct values only through their counts and it has density on $(\R^+)^d$ given by
\begin{equation*}
f_{\bmU_n}(\bmu\mid \bmn_1,\dots,\bmn_d) \propto 
\Psi(K,\bmu)
\prod_{j=1}^d
\left\{
\frac{u_j^{n_j - 1}}{\Gamma(n_j)}
\prod_{k=1}^K \kappa_j(u_j,n_{j,k})
\right\}
\label{eqn:fullU}
\end{equation*}
\end{corollary}
\begin{proof}
The statement follows from Equation \eqref{eqn:postproof_5}, conditioning on $\bmtheta$ and computing the expected value as done in Equation \eqref{eqn:postproof_9}. 
\end{proof}

\subsection{Proof of Equation \eqref{eqn:Psi_def_VecFDP}\label{proof:Psi}} 
We refer to the case when $q_M$ is chosen to be a $1-$shifted Poisson distribution with parameter $\Lambda$, that is
\begin{equation*}
\P\left(M=m\mid \Lambda\right) \ = \ q_M(m) \ = \ 
\frac{e^{-\Lambda}\Lambda^{m-1}}{(m-1)!},
\quad
\text{ for }m=1,2,\dots
\end{equation*}
In such a case, the $\Psi(K,\bmu)$ function, defined in Equation \eqref{eqn:Psi} for $K\geq0$ and $\bmu\in(\R)^d$ equals 
\begin{equation*}
\label{eqn:Psi_VecFDP_proof1}
\begin{split}
\Psi(K,\bmu) &= 
\sum_{m=0}^\infty
\frac{(m+K)!}{m!}
\frac{e^{-\Lambda}\Lambda^{m+K-1}}{(m+K-1)!}
\prod_{j=1}^d\left(\psi_j(u_j)\right)^m \\
& =
e^{-\Lambda}\Lambda^{K-1}
\sum_{m=0}^\infty
\frac{(m+K)\Lambda^{m}}{m!}
\bmpsi(\bmu)^m,
\end{split}
\end{equation*}
where $\bmpsi(\bmu) = \prod_{j=1}^d(\psi_j(u_j) $.
We note that, in the case under consideration, $\psi_j(u_j) = 1/(u_j+1)^{\gamma_j} < 1$, hence, the series does converge. As a consequence, we can split the series in two parts
\begin{equation*}
\label{eqn:Psi_VecFDP_proof2}
\begin{split}
\Psi(K,\bmu) &= 
e^{-\Lambda}\Lambda^{K-1}
\left(
\sum_{m=1}^\infty
\frac{\Lambda^{m}}{(m-1)!}
\bmpsi(\bmu)^m
+
K\sum_{m=0}^\infty
\frac{\Lambda^{m}}{(m)!}
\bmpsi(\bmu)^m
\right),
\end{split}
\end{equation*}
where we point out that the first sum starts from $m=1$ because we simplified $m/m!$, which is exactly equal to $0$ when $m=0$. Then, we change of the index in the first sum ($h=m-1$)
\begin{equation*}
\label{eqn:Psi_VecFDP_proof3}
\begin{split}
\Psi(K,\bmu) &= 
e^{-\Lambda}\Lambda^{K-1}
\left(
\Lambda\bmpsi(\bmu)
\sum_{h=0}^\infty
\frac{1}{h!}
\left(\Lambda\bmpsi(\bmu)\right)^h
+
K\sum_{m=0}^\infty
\frac{\Lambda^{m}}{(m)!}
\prod_{j=1}^d\left(\psi_j(u_j)\right)^m
\right) .
\end{split}
\end{equation*}
The result follows after recognizing the exponential series
\begin{equation*}
\Psi(K,\bmu) = e^{-\Lambda}\Lambda^{K-1}
\left(
\Lambda\bmpsi(\bmu)
e^{\Lambda\bmpsi(\bmu)}
+
K e^{\Lambda\bmpsi(\bmu)}
\right)
= 
\Lambda^{K-1} 
\left(K+\Lambda \bmpsi(\bmu)\right) 
e^{-\Lambda\left(1-\bmpsi(\bmu)\right)}.
\end{equation*}

\subsection{Proof of Theorem \ref{thm:priorK}}\label{proof:priorK}
\begin{proof}
Consider a sequence of $n$ observations divided into $d=2$ groups of length $n_1$ and $n_2$, respectively. We want to compute the probability of having $K_{(n_1,n_2)}=K$ global clusters under a $\operatorname{Vec-FDP}(q_M,\bmgamma,\Lambda)$ prior.\\
Firstly, we remind the following formula \cite[Theorem 8.16]{chara2002}
\begin{equation}
\label{eqn:gen_formula}
\lvert C(n, K; -\gamma, -\rho )\rvert\ = \ \frac{1}{K!} \sum  \binom{n}{r_1, \ldots , r_K} \prod_{j=1}^K  (\gamma)_{r_j} 
\end{equation}
where $(x)_{n}$ denotes the Pochhammer symbol, i.e., $(x)_n = \frac{\Gamma(x+n)}{\Gamma(x)}$. In general, this is defined for any $n\in\R^+$ but, for positive and integer values of $n$, it reduces to the rising factorial of $x$ of order $n$. Moreover, in Equation \eqref{eqn:gen_formula}
the sum is extended over all the vectors $(r_1, \ldots , r_K) $ of positive integers such that $\sum_{j=1}^K r_j = n$.\\
The probability under study $\P(K_{(n_1, n_2)} = K )$ can be written as follows in terms of the pEPPF function,
\begin{equation*}
\label{eqn:Kprior_general_formula}
\P(K_{n_1, n_2} = K ) =  \sum_{(\star)}  \frac{1}{K!}  \prod_{j=1}^2 
\binom{n_j}{n_{j,1}, \ldots , n_{j,K}} \cdot \Pi_K^{(n)} (\bm{n}_1,  \bm{n}_2) 
\end{equation*}
where the sum $(\star)$ is extended over all the vectors $(n_{1,1}, \ldots , n_{1,K})$ and $(n_{2,1}, \ldots , n_{2,K})$ of non-negative integers satisfying the following constraints
\[
n_{j,k} \geq 0, \; k  \in {1, \ldots , K} ,\; j \in\{ 1,2 \}, \; \sum_{k=1}^K n_{j,k} =n_j \; j \in \{ 1,2\}
\text{ and } n_{1,k}+n_{2,k} \geq 1 \text{ as } k=1, \ldots , K .
\]
By exploiting the expression of the pEPPF we get
\begin{equation} 
\label{eqn:K_somme_partizioni}
\begin{split}
\P(K_{(n_1, n_2)} =K ) &= 
\sum_{\bar{M}=0}^\infty  \frac{(\bar{M}+K)!}{\bar{M}!}
q_M (\bar{M}+K)   \prod_{j=1}^2  \frac{1}{(\gamma_j (\bar{M}+K))_{n_j}}\\
& \qquad \times
\sum_{(\star)} \frac{1}{K!}  \prod_{j=1}^2 
\binom{n_j}{n_{j,1}, \ldots , n_{j,K}} \cdot \prod_{j=1}^2 \prod_{k=1}^K  (\gamma_j)_{n_{j,k}}.
\end{split}
\end{equation}
We now focus on the evaluation of the sums over the partitions in \eqref{eqn:K_somme_partizioni}. This may be calculated by counting how many zeros we have in the vectors $(n_{j,1}, \ldots , n_{j,K})$. We denote by $r_1$ (resp. $r_2$) the number of possible zeros in the first vector (resp. second vector). Thanks to the validity of the condition
$ n_{1,k}+n_{2,k} \geq 1$, we know that if $n_{1,k}$ then $n_{2,k} \geq 1$ and vice versa.
Moreover, we observe that one has $\binom{K}{r_1}$ ways to choose the zeros in the first vector and $\binom{K-r_1}{r_2}$ possibilities to choose the zeros in the second one. Thus, we get
\begin{align*}
&\sum_{(\star)} \frac{1}{K!}  \prod_{j=1}^2 
\binom{n_j}{n_{j,1}, \ldots , n_{j,K}} \cdot \prod_{j=1}^2 \prod_{k=1}^K  (\gamma_j)_{n_{j,k}}\\
& \qquad = \sum_{r_1=0}^K  \sum_{r_2=0}^{K-r_1} \binom{K}{r_1} \binom{K-r_1}{r_2}
\sum_{(\star_1)}\sum_{(\star_2)} \frac{1}{K!}  \prod_{j=1}^2 
\binom{n_j}{n_{j,1}, \ldots , n_{j,K-r_j}} \cdot \prod_{j=1}^2 \prod_{k=1}^{K-r_j}  (\gamma_j)_{n_{j,k}}
\end{align*}
where
\[
(\star_j) =  \left\{ (n_{j,1}, \ldots , n_{j,K-r_j}) :\; n_{j,k} \geq 1 , k=1,\ldots , K-r_j, \,  \sum_{k=1}^{K-r_j} n_{j,k} =n_j  \right\}
\]
as $j=1,2$. 
We now exploit \eqref{eqn:gen_formula} to obtain
\begin{equation}\label{eqn:sommmona_K}
\begin{split}
&\sum_{(\star)} \frac{1}{K!}  \prod_{j=1}^2 
\binom{n_j}{n_{j,1}, \ldots , n_{j,K}} \cdot \prod_{j=1}^2 \prod_{k=1}^K  (\gamma_j)_{n_{j,k}}  \\
& \qquad \times \sum_{r_1=0}^K  \sum_{r_2=0}^{K-r_1} 
\prod_{j=1}^2 \lvert C (n_j  , K-r_j ; -\gamma_j ) \rvert
\end{split}
\end{equation}
If we now replace \eqref{eqn:sommmona_K} in \eqref{eqn:K_somme_partizioni}, the statement of the theorem follows.
\end{proof}

\section{Predictive distributions}
\label{app:predictive}

In this section, we provide and prove the formula for the predictive distribution for a new client entering a restaurant $j$. This has already been stated in Equation \eqref{eqn:predittiva_FDP} for the special case of a $\operatorname{Vec-FDP}$ prior. Instead, in this section we refer to the more general $\operatorname{Vec-NIFPP}$.
More precisely, consider a realization $\left(\bmtheta_1,\dots,\bmtheta_d\right)$ from $\left(P_1,\dots,P_d\right)\sim\operatorname{Vec-NIFPP}(q_M,H,P_0)$ with $K$ distinct values $\left(\theta^{**}_1,\dots,\theta^{**}_K\right)$ and a partition $\rho = \left\{C_1,\dots,C_K\right\}$ with counts $\left(\bmn_1,\dots,\bmn_d\right)$ satisfying the constraints in Equation \eqref{eqn:cluster_vincoli}.
Following the approach of \citet{jlp2009}, \citet{favaroteh2013} and \citet{argiento2022annals} we work conditionally to $\bmU_n = \bmu$. Then, for each group, say $j$, we have  
\begin{equation}
\label{eqn:app_pred_general}
\begin{split}
\mathbb{P}(\theta&_{j,n_j+1} \in \cdot \mid \bmtheta_1,\dots,\bmtheta_d, \bmu)  \\
& \propto 
\sum_{k=1}^K \frac{\kappa_j(u_j,n_{j,k}+1)}{\kappa_j(u_j,n_{j,k})}\delta_{\theta^{**}_k}(\cdot) 
\ + \
\kappa_j(u_j,1)
\left(
\prod_{h\neq j}
\kappa_h(u_h,0)
\right)
\frac{\Psi(K+1,\bmu)}{\Psi(K,\bmu)}
P_0(\cdot)
\end{split}
\end{equation}
%
Without loss of generality, consider a new client entering the first restaurant, where $n_1$ customers are currently seated. This $(n_1+1)$-th client can $(i)$ choose an already existing table, or $(ii)$ sit to a non-existing table. 
Note that, as detailed in \Cref{subsection:CRFP}, we distinguish between empty and non-existing tables. The latter refers to tables that do not appear in any restaurant of the franchise, while we admit the possibility of having tables with no clients in some restaurants as long as at least one customer in the whole franchise is seated at the corresponding table, see \Cref{fig:tavoli}. 

As far as the proof of \Cref{eqn:app_pred_general} is concerned, considering  case $(i)$: we compute the probability of sitting to table $k$, $k=1,\dots,K$, where dish $\theta^{**}_k$ is served. 
This probability equals 
\begin{equation*}
\label{eqn:pred_gen_1}
\begin{split}
\P\bigl(\theta_{1,n_1+1} = \theta^{**}_k  \mid \bmtheta_1,\dots,\bmtheta_d\bigr) \ = \ 
\Pi^{(n+1)}_{K}(\widetilde{\bmn}_1,\dots,\bmn_d),
\end{split}
\end{equation*}
where $\widetilde{\bmn}_1 = (n_{1,1},\dots,n_{1,k}+1,\dots,n_{1,K})$. If we further condition to $\bmU_n=u$, we obtain
\begin{equation*}
\label{eqn:pred_gen_2}
\begin{split}
\P\bigl(&\theta_{1,n_1+1}  = \theta^{**}_k \mid
\bmtheta_1,\dots,\bmtheta_d, \bmu \bigr) \\
&\propto
\Psi(\bmu,K)
\left(\prod_{j=2}^K
\frac{(u_j)^{n_j-1}}{\Gamma(n_j)}
\prod_{l=1}^K
\kappa_j(u_j,n_{j,l})
\right)
\frac{(u_1)^{n_1}}{\Gamma(n_1+1)}
\kappa_1(u_1,n_{1,k}+1)
\prod_{l\neq k}
\kappa_1(u_1,n_{1,l}) \\
& =
\Psi(\bmu,K)
\left(\prod_{j=2}^K
\frac{(u_j)^{n_j-1}}{\Gamma(n_j)}
\prod_{l=1}^K
\kappa_j(u_j,n_{j,l})
\right)
\frac{(u_1)^{n_1}}{\Gamma(n_1+1)}
\frac{\kappa_1(u_1,n_{1,k}+1)}{\kappa_1(u_1,n_{1,k})}
\prod_{l=1}^K
\kappa_1(u_1,n_{1,l}),
\end{split}
\end{equation*}
where the second equality follows since we just multiplied and divided by $\kappa_1(u_1,n_{1,k})$.

Now, considering case $(ii)$: we are interested in the probability
of sitting at a new, non-existing, table where a new dish $\theta^{**}_{K+1}$ drawn from $P_0$ is served. This equals to
\begin{equation*}
\label{eqn:pred_gen_3}
\begin{split}
\P\bigl(\theta_{1,n_1+1} \in \cdot \mid \bmtheta_1,\dots,\bmtheta_d\bigr) \ = \ 
\Pi^{(n+1)}_{K+1}(\bmn^\dagger_1,\dots,\bmn^\dagger_d),
\end{split}
\end{equation*}
where 
$\bmn^\dagger_1 = (n_{1,1},\dots,n_{1,K},1)$ and 
$\bmn^\dagger_j = (n_{j,1},\dots,n_{j,K},0)$ for each $j=2,\dots,d$.
As before, we further condition to $\bmU_n=u$ and we obtain
\begin{equation*}
\label{eqn:pred_gen_4}
\begin{split}
\P\left(\theta_{1,n_1+1} \in \cdot \mid
\bmtheta_1,\dots,\bmtheta_d, \bmu \right) & \propto 
\ \Psi(\bmu,K+1)
\left(\prod_{j=2}^K
\frac{(u_j)^{n_j-1}}{\Gamma(n_j)}
\kappa_j(u_j,0)
\prod_{l=1}^K
\kappa_j(u_j,n_{j,l})
\right) \\
&
\qquad \times \frac{(u_1)^{n_1}}{\Gamma(n_1+1)}
\kappa_1(u_1,1)
\prod_{l=1}^K
\kappa_1(u_1,n_{1,l})
P_0(\cdot).
\end{split}
\end{equation*}

We can collect the previous expression in the following compact form
\begin{equation*}
\label{eqn:pred_gen_5}
\begin{split}
&\P\left(\theta_{1,n_1+1} \in \cdot \mid
\bmtheta_1,\dots,\bmtheta_d, \bmu \right) \propto\\
&
\Psi(\bmu,K)
\sum_{k=1}^K
\left\{
\left(\prod_{j=2}^K
\frac{(u_j)^{n_j-1}}{\Gamma(n_j)}
\prod_{l=1}^K
\kappa_j(u_j,n_{j,l})
\right)
\frac{(u_1)^{n_1}}{\Gamma(n_1+1)}
\frac{\kappa_1(u_1,n_{1,k}+1)}{\kappa_1(u_1,n_{1,k})}
\prod_{l=1}^K
\kappa_1(u_1,n_{1,l}) 
\delta_{\theta^{**}}(\cdot)
\right\}\\
& + \Psi(\bmu,K+1)
\left(\prod_{j=2}^K
\frac{(u_j)^{n_j-1}}{\Gamma(n_j)}
\kappa_j(u_j,0)
\prod_{l=1}^K
\kappa_j(u_j,n_{j,l})
\right)
\frac{(u_1)^{n_1}}{\Gamma(n_1+1)}
\kappa_1(u_1,1)
\prod_{l=1}^K
\kappa_1(u_1,n_{1,l})
P_0(\cdot).
\end{split}
\end{equation*}
The latter is a mixture (with unnormalized weights) of Dirac point masses centred in the distinct values $\theta^{**}_k$ and the prior $P_0$. 
Finally, we cancel out the common term in the unnormalized weights and we get
\begin{equation*}
\label{eqn:pred_gen_6}
\begin{split}
\P\bigl(\theta&_{1,n_1+1} \in \cdot \mid
\bmtheta_1,\dots,\bmtheta_d, \bmu \bigr) \\
&\propto
\Psi(\bmu,K)
\sum_{k=1}^K
\left\{
\frac{\kappa_1(u_1,n_{1,k}+1)}{\kappa_1(u_1,n_{1,k})}
\delta_{\theta^{**}}(\cdot)
\right\} + 
\Psi(\bmu,K+1)
\left(\prod_{j=2}^K
\kappa_j(u_j,0)
\right)
\kappa_1(u_1,1)
P_0(\cdot).
\end{split}
\end{equation*}
Dividing by $\Psi(\bmu,K)$, Equation \Cref{eqn:app_pred_general} is obtained. 

Figure \ref{fig:tavoli} provides a graphical representation of the Chinese restaurant franchise process for Vec-NIFPP based on a sample of six, five and five observations in three groups, respectively.   The first (left panel) and the second client (middle panel) enter the first restaurant and sit at different tables. Tables having the same color must be prepared in all other restaurants. Finally, (right panel) one of the possible configurations when all clients are seated.  
Global clustering is obtained by merging tables having the same color.   

\begin{figure}[h]
\centering
\includegraphics[width=0.31\linewidth]{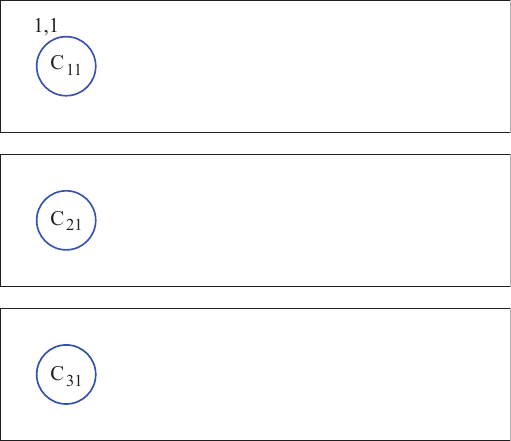}
\hfill
\includegraphics[width=0.31\linewidth]{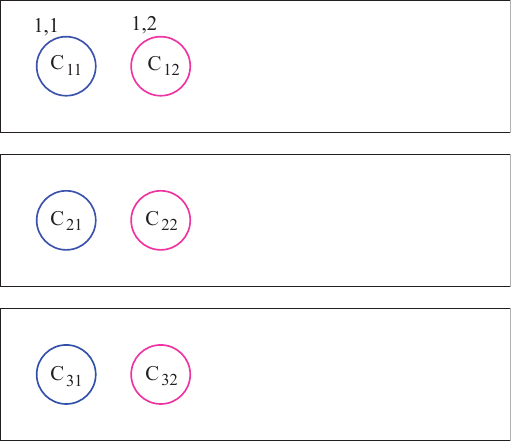}
\hfill
\includegraphics[width=0.31\linewidth]{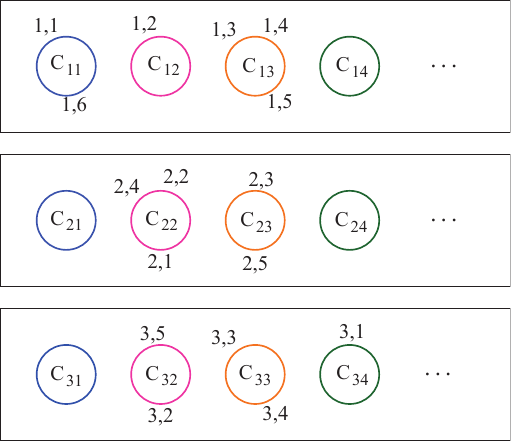}
\caption{Chinese restaurant franchise process representation for $\operatorname{Vec-NIFPP}$ based on a sample of $6$, $5$ and $5$ observations in $d=3$ groups, respectively. }
\label{fig:tavoli}
\end{figure}

\section{MCMC Algorithm for HMFM}
\label{app:MCMC}
From this point on, we place our attention on the HMFM model introduced in \Cref{section:HMFM}.
\subsection{Hyperparameters setup}
\label{app:hyperparam}
As mentioned in \Cref{section:HMFM}, we rely on the sample equivalence principle \citep{diaconis1979} to design a suitable reparametrization of the prior $$
\Lambda\mid\bmgamma \sim \operatorname{Gamma}\left(a_\Lambda + d a_\gamma, b_\Lambda + b_\gamma\sum_{j=1}^d \gamma_j\right),
$$ 
with an easier interpretation. 
Indeed, it is straightforward to show that
\[
\E[\Lambda\mid\bmgamma] = \left( \frac{1}{\E[\Lambda]}\frac{a_\Lambda/a_\gamma}{a_\Lambda/a_\gamma + d} + \frac{1}{\widehat{\Lambda}}\frac{d}{a_\Lambda/a_\gamma+d}\right)^{-1},
\]
where $\E[\Lambda]$ is the prior mean of $\Lambda$, i.e., $a_\Lambda/b_\Lambda$ and $\widehat{\Lambda} = \frac{a_\gamma}{b_\gamma}\frac{d}{\sum_{j=1}^d \gamma_j}$
is the maximum likelihood estimator for $\Lambda$ based on a sample $\gamma_1,\dots,\gamma_d$. Hence, $a_\Lambda/a_\gamma$ represents the size of a prior sample and it can be compared to the size $d$. Let $\Lambda_0$ and $\text{V}_\Lambda$ be the prior expected value and variance of $\Lambda$, respectively, and $\gamma_0$ be  
a common prior guess   for the $\gamma_j$'s, so that $\frac{1}{d}\sum_{j=1}^d \gamma_j = \gamma_0$. Then, assuming $\E[\Lambda] = \widehat{\Lambda}$, 
we obtain
\[
a_\gamma = \frac{1}{d}\frac{\Lambda^2_0}{\text{V}_\Lambda},
\quad
b_\gamma = \frac{a_\gamma}{\gamma_0\Lambda_0},
\quad
a_\Lambda = \frac{\Lambda_0^2}{\text{V}_\Lambda},
\quad
b_\Lambda = \frac{\Lambda_0}{\text{V}_\Lambda}.
\]
We now provide an intuition for choosing $\Lambda_0$, $\text{V}_\Lambda$ and $\gamma_0$. We would like our mixture to be sparse, which, in practice, means that the expected number of clusters is (much) smaller than the expected number of mixture components. 
Under the HMFM model, this translates to small values of $\gamma_j$'s. 
Looking at each season separately, 
we note that
each element of the Vec-FDP 
is a finite DP. Hence, it is easy to compute the prior distribution for each local number of cluster $K_j$, using  \Cref{eqn:Kprior_d2} with $n_2 = 0$, and compare it with the law of $M$. Setting a small number of clusters for each season implies being a priori noninformative about the heterogeneity across different seasons.



\subsection{Conditional Algorithm}
\label{app:MCMC_conditional}
The conditional algorithm for the HMFM model, introduced in \Cref{subsection:computational}, provides full Bayesian inference on both the mixing parameters $(P_1,\dots, P_d)$ and the clustering structure $\rho$. We use the following notation, $\mathcal{L}\left(X\mid \text{rest}\right)$, to indicate the distribution of a random variable $X$ conditionally to everything but itself. 
The algorithm is a blocked Gibbs sampler whose state space is 
$(\bmc_1,\dots, \bmc_d,\bmtau,\bmS, M,\bmU,\Lambda,\bmgamma)$. The algorithm is composed of two main steps.

\paragraph{Step 1:}  Metropolis-within-Gibbs algorithm to update one parameter at a time. 
\begin{itemize}
\item[1.1 ] Update cluster allocations $\bmc$. This step is done independently for each group $j$ and for each individual $i$ by sampling $c_{j,i}$ from a discrete distribution such that
\[
\mathbb{P}(c_{j,i} = m \mid \text{rest}) \propto S_{j,m}f(y_{j,i}\mid\tau_m), \quad \text{for }m=1,\dots,M.
\]
After sampling the whole vector, compute the number of allocated components $K$.
This operation takes $O(Mn)$ time.
\item[1.2 ] Update mixture parameters $\bmtau = \left(\tau_1,\dots,\tau_M\right)$. From \Cref{thm:posterior}, the full conditional distribution further factorizes in two independent parts, one for the first $K$ allocated components and one for the remaining $M^{*}$ non-allocated components. In particular, for $k=1,\dots,K$, sample $\tau_k$ independently from 
\[
\mathbb{P}\left(\tau_k \in d\tau_k \mid \text{rest} \right) \propto
\left\{
\prod_{(j,i):c_{j,i}=k}
f(y_{j,i}\mid\tau_k)
\right\}
P_0(d\tau_k)
\]
while for $m=K+1,\dots,M$, sample $\tau_m$ independently from the prior $P_0$.
In general, to find elements belonging to some cluster $k$ requires performing a search a vector of length $n$, which takes $O(n)$ time. Hence, overall, this operation would take $O(Kn)$ time. Nevertheless, a smart implementation of the code can avoid such a search. Therefore, in practice, the computational time of this update scales sublinearly with respect to $n$.
\item[1.3 ] Update the unnormalized mixture weights $\bmS$ differently for the allocated and non-allocated parts of the process. In particular, for $k=1,\dots,K$, sample $S_{j,k}$ from 
$\operatorname{Gamma}(\gamma_j + n_{j,k},U_j + 1)$, independently for each $j=1,\dots,d$, while for $m=K+1,\dots,M$, sample $S_{j,m}$ from 
$\operatorname{Gamma}(\gamma_j, U_j + 1)$.
This operation takes $O(Md)$ time.
\item[1.4 ] Update $\Lambda$ from the mixture of gamma densities
\begin{equation*}
\label{eqn:conditional_fullLambda_conditional}
\begin{split}
&\frac{K\left(b^* + 1 - \bmpsi(\bmu)  \right)}{(a^*-1) (\bmpsi(\bmu)   + K(b^*+1)}
\operatorname{Gamma}\left(\Lambda;a^*+K-1, 
b^* + 1 - \bmpsi(\bmu)  \right)  \\ 
& \qquad + 
\frac{(a^*+K-1)(\bmpsi(\bmu)  }
{(a^*-1) \bmpsi(\bmu)   + K(b^*+1)}
\operatorname{Gamma}\left(\Lambda;a^*+K, 
b^* + 1 - \bmpsi(\bmu)  \right)
\end{split}
\end{equation*}
where 
$
a^* = a_\Lambda + da_\gamma; \ b^* = b_\Lambda + b_\gamma \sum_{j=1}^d \gamma_j.
$ 
This update is straightforward as it only takes $O(1)$ time.
\end{itemize}

\paragraph{Step 2:} To avoid conditioning to $\bmU$ when updating $\bmgamma$ and $M$, we write the full conditional distribution of $\Delta_2$ as $\mathcal{L}\left(\bmU\mid \text{rest}\right)\mathcal{L}\left(M,\bmgamma\mid\Delta_1,\bmy\right)$.
\begin{itemize}
\item[2.1 ] Update $U_j$ from $\operatorname{Gamma}(n_j,T_j)$, independently for each $j=1,\dots,d$.
This operation takes $O(d)$ time.
\item[2.2 ] Update $M=K+M^*$, where $K$ has been updated in Step 1.1 and $M^*$ is sampled from
\[
\mathbb{P}\left(M^* = m^*\mid\bmgamma,\Delta_2,\bmy\right) \propto
\frac{m^*+K}{m^*!}\Lambda^{m^*}
\prod_{j=1}^d
\frac{\Gamma\left(\gamma_j(m^*+K)\right)}{\Gamma\left(\gamma_j(m^*+K) + n_j\right)}.
\]
This step is accomplished using an adaptive Metropolis-Hastings step.
This update takes $O(1)$ time.
\item[2.3 ] Update $\gamma_j$ from a density proportional to
\begin{equation*}
\label{eqn:conditional_fullgamma_marginal}
\begin{split}
\frac{\Gamma\left(\gamma_j(M^*+K)\right)}{\Gamma\left(\gamma_j(M^*+K)+n_j\right)}
\prod_{k=1}^K
\left(\frac{\Gamma\left(n_{j,k}+\gamma_j\right)}{\Gamma\left(\gamma_j\right)}\right)
\operatorname{Gamma}(\gamma_j;a_\gamma,\Lambda b_\gamma),
\end{split}
\end{equation*}
independently for each $j=1,\dots,d$. This update is performed via the adaptive Metropolis-Hastings step.
This operation takes $O(d)$ time.
\end{itemize}
Overall, the computational time is driven by the update of cluster allocation variables and the unique values. As discussed above, in practice the latter can be easily reduced. In general, each iteration takes $O((M+K)n)$ even though in practice the computational time is much closer to $O(Mn)$.

\subsection{Marginal Algorithm}
\label{app:MCMC_marginal}
We now turn our attention to the marginal algorithm for the HMFM model. The algorithm is based on the Chinese restaurant franchise characterization of the Vector of Finite Dirichlet process (see \Cref{subsection:CRFP} in the main manuscript). The state space of the Markov Chain is identified by the parameters $\left(\bmc_1,\dots,\bmc_d, \bmU, \Lambda, \bmgamma\right)$. The algorithm is composed of four steps:
\begin{paragraph}{Step 1:} Update cluster allocations $\bmc$. 
This is done in $n=\sum_{j=1}^d n_j$ sub-steps. 
For each restaurant $j=1,\dots,d$ and for each customers $i=1, \ldots, n_j$,
we remove observation $i$ from the cluster and denote by 
$\rho_j^{-i}=\left\{C_{j,1}^{-i}, \ldots, C_{j,K^{-i}}^{-i}\right\}$ the resulting partition in $K^{-i}$ clusters. 
Then, the $i$-th observation is assigned to a new cluster $C_{j,K^{-i}+1}^{-i}$ with probability proportional to
\begin{equation*}
\label{eqn:newtable_fullcond}
\begin{split}
\mathbb{P}
\left(i \in C_{j,K^{-i}+1}^{-i} \mid \bmu, \rho_j^{-i}, \bmy\right) 
\propto 
\bmpsi(\bmu)
\gamma_j\Lambda
\frac{K^{-i}+1+\Lambda\bmpsi(\bmu)}{K^{-i}+\Lambda\bmpsi(\bmu)}
\mathcal{M}\left(y_{j,i}\right)
\end{split}
\end{equation*}
while it is assigned to each of the existing clusters $C_{j,l}^{-i}$, for  $l=1, \ldots, K^{-i}$, with probability proportional to
\begin{equation*}
\label{eqn:oldtable_fullcond}
\begin{split}
\mathbb{P}
\left(i \in C_{j,l}^{-i} \mid \bmu, \rho_j^{-i}, \bmy\right) 
\propto 
\left(n_{j,k}^{-i} + \gamma_j\right)
\frac{\mathcal{M}\left(\bmy_{C_{l}^{-i}\cup i}\right)}
{\mathcal{M}\left(\bmy_{C_{l}^{-i}}\right)}
\end{split}
\end{equation*}
where $\bmpsi(\bmu) = \prod_{j=1}^d \psi_j\left(u_j\right) $ and, for each global cluster $C_{l}\in \rho$, $\bmy_{C_{l}}$ is the vector of the $y_{h,g}$, $h\in\{1,\dots,d\}$ and $g\in\{1,\dots,n_h\}$, such that ${h,g} \in C_l$ and 
\[
\mathcal{M}\left(\bmy_{C_l}\right) = 
\int_\X 
\prod_{(h,g) \in C_l} f\left(y_{h,g} \mid \theta\right) P_0(d\theta)
\]
is the marginal distribution of the data within the cluster $C_l$ with sampling model $f\left(y_{h,g} \mid \theta\right)$ and prior $P_0(d\theta)$.
\end{paragraph}
The computational time is proportional to $(K+1)n$ times the evaluation time for the marginal distribution $\mathcal{M}(\cdot)$). This latter cost may be problem specific. In our experiments and application, $\mathcal{M}(\cdot)$ is available in closed analytical form and its parameters are easily computed by updating the sufficient statistics of the likelihood $f$. Hence, for our purposes, this update takes $O((K+1)n)$.
\begin{paragraph}{Step 2:}
Update the whole vector $\bmU$ by sampling from a density proportional to 
$$
\prod_{j=1}^d\left\{
U_j^{n_j-1} \frac{1}{(1+U_j)^{n_j + K\gamma_j}}
\right\}
\left( K + \Lambda\prod_{j=1}^d\frac{1}{(1+U_j)^{\gamma_j}}\right)
\exp \left\{
\Lambda\prod_{j=1}^d\frac{1}{(1+U_j)^{\gamma_j}}
\right\}.
$$
This operation takes $O(d)$ time.
\end{paragraph}
\begin{paragraph}{Step 3:}
Update the whole vector $\bmgamma$ by sampling from a density proportional to
\begin{equation*}
\begin{split}
\prod_{j=1}^d&\left\{
\operatorname{Gamma}\left(\gamma_j;a_\gamma,\Lambda b_\gamma\right)
\frac{1}{(1+u_j)^{K\gamma_j}}
\prod_{m=1}^K\frac{\Gamma(\gamma_j + n_{j,m})}{\Gamma(\gamma_j)}
\right\}  \\
& \times \left(K+\Lambda\prod_{j=1}^d\frac{1}{(1+U_j)^{\gamma_j}}\right)
\exp \left\{
\Lambda\prod_{j=1}^d\frac{1}{(1+U_j)^{\gamma_j}}
\right\}
\end{split}
\end{equation*}
This operation takes $O(d)$ time.
\end{paragraph}    
\begin{paragraph}{Step 4:}
Update $\Lambda$ from the mixture of gamma densities
\begin{equation*}
\label{eqn:conditional_fullLambda_main}
\begin{split}
& \frac{K\left(b^* + 1 - \bmpsi(\bmu)  \right)}{(a^*-1)( \bmpsi(\bmu)   + K(b^*+1)}
\operatorname{Gamma}\left(\Lambda;a^*+K-1, 
b^* + 1 - \bmpsi(\bmu)  \right)  \\ 
& \qquad + 
\frac{(a^*+K-1)\bmpsi(\bmu)  }
{(a^*-1)( \bmpsi(\bmu)   + K(b^*+1)}
\operatorname{Gamma}\left(\Lambda;a^*+K, 
b^* + 1 - \bmpsi(\bmu)  \right)
\end{split}
\end{equation*}
where 
$
a^* = a_\Lambda + da_\gamma; \ b^* = b_\Lambda + b_\gamma \sum_{j=1}^d \gamma_j
$.
This update is straightforward as it only takes $O(1)$ time.
\end{paragraph}    

Note that both Steps 2 and 3 require samples from $d$-dimensional distributions which are not available in closed form and which can not be factorized in the product of $d$ independent components. Joint updates must be performed. We suggest to use adaptive Metropolis-Hastings techniques or gradient based methods.
In particular, we use here the Metropolis Adjusted Langevin Algorithm (MALA) introduced by \cite{MALA1998}.
Overall, each iteration takes $O((K+1)n)$.

\section{Simulation study}
\label{app:simstudy}
In this section we present a simulation study to show our model’s ability to perform clustering as well as density estimation. 
In particular, we investigate the performance of the HMFM considering both the conditional (HMFM-cond) and the marginal (HMFM-marg) samplers. Thus, we compare the HMFM model  with the HDP \citep{teh2006hierarchical} mixture model. In the first scenario, we also  compare the HMFM against the MFM model 
that is assumed independently for each group and fitted using the algorithm proposed by 
\citet{argiento2022annals}. 
The latter is implemented in the 
\texttt{R} package \texttt{AntMAN} \citep{antman2021}. 

The simulation study consists of  three scenarios. 
The first experiment considers $d=2$ groups which share a component. Thus, the example shows the benefits of the hierarchical analysis against the independent group-specific analysis. 
In the second scenario, again $d=2$ groups are considered, but without any components shared across the groups. This example shows how the HDP borrows too much information  from the other group with respect to the HMFM, leading to misleading results. 
Finally, the third experiment emphasizes such differences in a more realistic situation, where $d=15$ groups are considered; here, the HMFM outperforms the HDP in terms of both clustering and density estimation.
For all experiments, $50$ independent datasets have been generated.

For a fair comparison, we fit all competing models by setting the same base measure $P_0$. In particular, we follow the  strategy explained in \Cref{section:HMFM} and set a $\operatorname{Normal-InvGamma}$ prior with $\mu_0 = \text{mean}(\mathbf{y})$, 
$k_0 = 1 /(\text{range}(\mathbf{y}))^2$, 
$\nu_0 = 4$ and $\sigma^2_0 = 1/2$. 
Regarding the hyperparameters, for the HDP we set default hyperprior $\alpha \sim \operatorname{Gamma}(1,1)$ and $\gamma \sim \operatorname{Gamma}(1,0.1)$ \citep{teh2006hierarchical}.
Also,  we use default  hyperparameters   for the independent analysis as suggested in the $\texttt{AntMAN}$ package, i.e., $\gamma \sim \operatorname{Gamma}(1,1)$ and $\Lambda \sim \operatorname{Gamma}(1,1)$.
Finally,  for the  HMFM, we follow the strategy 
described in \Cref{app:hyperparam} to set the values of 
$\Lambda_0 $, $\text{V}_\Lambda$ and $\gamma_0 $; the values are detailed for each experiment below.

To assess the performance of recovering the true clustering, we compute the Co-clustering error (CCE; \citealt{casarin2020})  and the Adjusted Rand Index (ARI; \citealt{arandi}). 
The CCE is based on 
the posterior probabilities $\hat{\pi}_{lk}\in [0,1]$ of two data points $l$ and $k$ to belong to the same cluster, i.e., the proportion of times that two observations have been assigned to the same mixing component over the MCMC iterations. Specifically, the error is defined as the average $L_1$ distance between the true pairwise co-clustering matrix with elements $\pi_{lk}\in \{0,1\}$ and the estimated co-clustering probability (similarity), namely
\begin{equation*}
\label{eqn:CCE_def}
\text{CCE} = \frac{1}{n} 
\sum_{l=1}^{n}\sum_{k=1}^{n} \ \lvert \ \pi_{lk} - \hat{\pi}_{lk} \ \rvert,
\end{equation*}
where $n = \sum_{j=1}^d n_j$. Being a distance, lower values of CCE indicate better fit, attaining zero in the absence of co-clustering error.

On the other hand, the ARI is computed using the  estimated partition which is obtained by minimizing the variation of information function \citep{wade2018,dahl2022}. The index ranges between zero (the true and the estimated clustering do not agree on any pair of points) and one (the two clustering are the same). 

Regarding density estimation,  we compute 
a Predictive Score (PS; \citealt{casarin2020}) for each group $j=1,\dots,d$, that is the $L_1$ distance between the group-specific density $f(y_{j,n_j+1})$ and the corresponding predictive density $\hat{f}(y_{j,n_j+1}\mid\mathbf{Y})$, namely 
\begin{equation*}
\label{eqn:PS_def}
\text{PS}_j =  \int_{-\infty}^{\infty} \ \lvert f(y_{j,n_j+1}) - \hat{f}(y_{j,n_j+1}\mid \mathbf{Y})  \rvert \ dy_{j,n_j+1}.
\end{equation*}
The PS is a distance so lower values indicate better fit, with zero being the lowest possible value. 

\subsection{Experiment 1}
\label{subsection:exp1_app}
This experiment considers $d=2$ groups, both having two local clusters ($K_1 = 2$, $K_2 = 2$), one of which is shared; hence,  the global number of clusters is $K = 3$. The mixing probabilities are set so that the shared component has a lower value in the second group. Namely
\begin{equation}
\begin{split}
y_{1,i} \ &\iid \ 0.5\text{N}(-3,0.1) + 0.5\text{N}(0,0.5), 
\ \text{for }i=1,\dots,300 \\
y_{2,i} \ &\iid \ 0.2\text{N}(0,0.5) + 0.8\text{N}(1.75,1.5),
\ \text{for }i=1,\dots,300 \\
\end{split}
\label{eq:sim1}
\end{equation}
Hence, the component specific parameters are $\mu = (-3,0,1.75)$ and $\sigma = (\sqrt{0.1},\sqrt{0.5},\sqrt{1.5})$.
We generate $n_1 = 300$ and $n_2 = 300$ observations for each group. 
Note that the second group is defined so that the two components strongly overlap. As a result, the shared component is completely masked in this group. Nevertheless, the masked component can be spotted by exploiting the sharing of information with the other group. \Cref{fig:exp1_esempio} shows   the empirical distribution of a dataset generated from model Equation \eqref{eq:sim1} and the underlying densities.
\begin{figure}[h]
\centering
\includegraphics[width=0.75\linewidth]{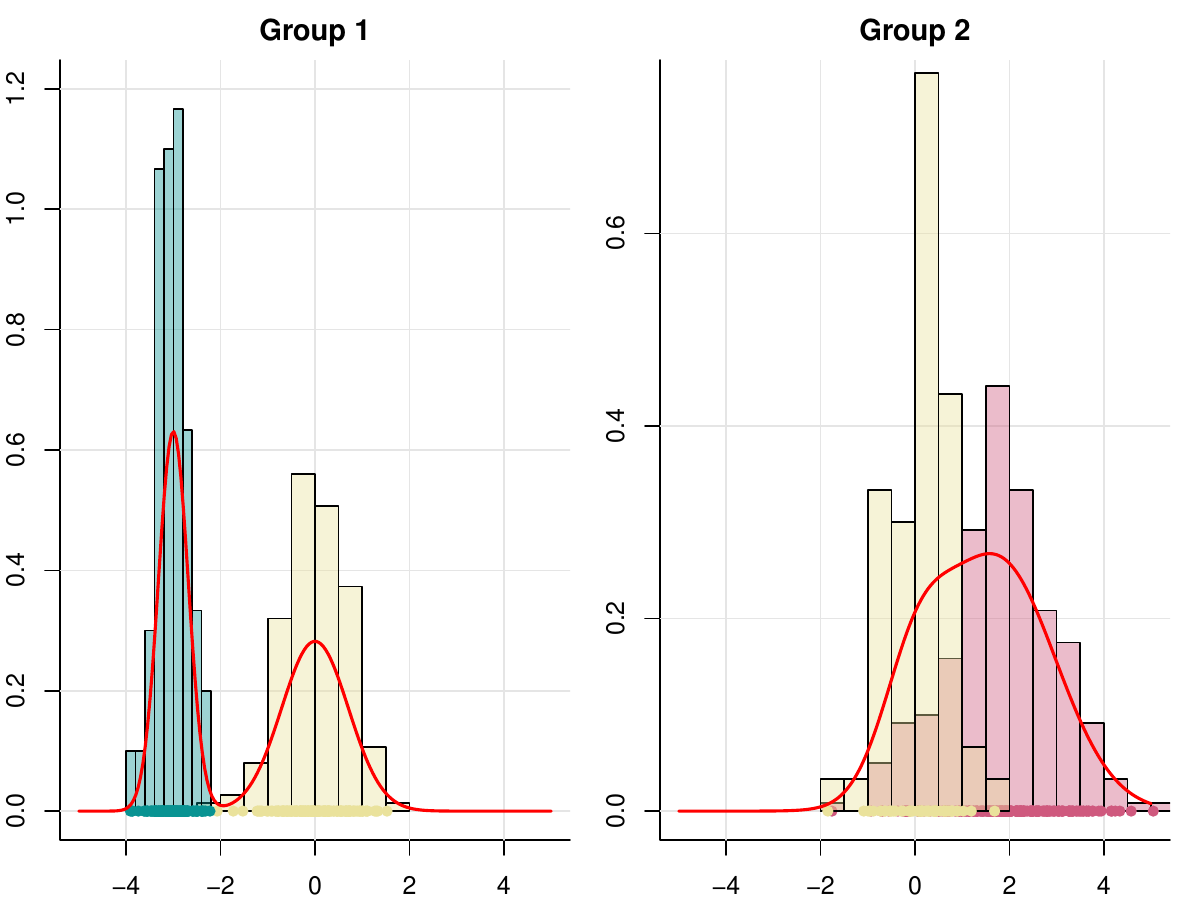}
\caption{Empirical distributions of a dataset simulated under Experiment 1. Dots represent the observations while lines represent the  underlying densities. Colours relate to the mixing components.}
\label{fig:exp1_esempio}
\end{figure}

For each simulated dataset, we fit the HDP, the independent group-specific MFM and the proposed HMFM by setting, following the guidelines in \Cref{app:hyperparam}, $\Lambda_0 = 5$, $\text{V}_\Lambda = 5$, $\gamma_0 = 0.5$.
A clear advantage of joint modelling is the possibility of deriving model-based clustering also across different groups, which is not possible when we run independent analyses. However, for the sake of a fair comparison with the MFM, clustering performances are evaluated within each group, not globally.  \Cref{tab:exp1_LocalARI} presents the results of model comparison based on the mean and standard deviation (in brackets) of the ARI over the simulated datasets. The 
results 
show that all methods, with negligible differences, are able to perfectly recover the true clustering in the first group, where the two components are well separated. Instead, the advantage of the sharing of information allowed by hierarchical modelling is evident in the second group, where components overlap. Indeed, the MFM fails to identify the presence of two different clusters in the second group, gathering all observations together and obtaining an ARI value close to zero. On the other hand, HDP and HMFM are able to borrow information from the first group to recognize the presence of two clusters in the second group, leading to higher values of ARI, with HMFM (both marginal and conditional samplers) outperforming the HDP. 
To clarify, we compute the percentage of simulated datasets for which each method is able 
to obtain a partition of the second group consisting only of a single cluster. 
The results are reported in the third column of \Cref{tab:exp1_LocalARI} and show how the MFM consistently fails to identify at least two clusters for the second group, 
showing the limitations of the independent analyses against the hierarchical approach. 
Finally, we point out that the HMFM is able to outperform the HDP, with slightly better performances for the marginal sampler over the conditional one.

\begin{table}
\caption{The first two columns report the mean and the standard deviation (in brackets)  of the ARI for the two groups over the $50$ simulated datasets under Experiment 1. The last column shows the percentage of times that each method gathers all observations of the second group in a single cluster, hence, lower values are better.}
\centering
\label{tab:exp1_LocalARI}
\begin{tabular}{cccc}
\hline
& Group 1      & Group 2      & \% 1 cluster  \\ \hline
MFM             & 0.992 (0.011) & 0.002 (0.037) & 98\% \\
HDP                & 0.992 (0.011) & 0.122 (0.183) & 68\% \\
HMFM - Marginal    & 0.990 (0.012) & 0.216 (0.191) & 42\% \\
HMFM - Conditional & 0.991 (0.012)  & 0.177 (0.190) & 52\% \\ \hline
\end{tabular}
\end{table}

To conclude the clustering comparison, \Cref{fig:exp1_CCE} provides a picture of the distribution of the CCE  within each group over the simulated datasets. Not surprisingly, differences are negligible in the first group, while the MFM commits a higher error in the second one, with also higher variability among different datasets.
\begin{figure}
\centering
\includegraphics[width=1\linewidth]{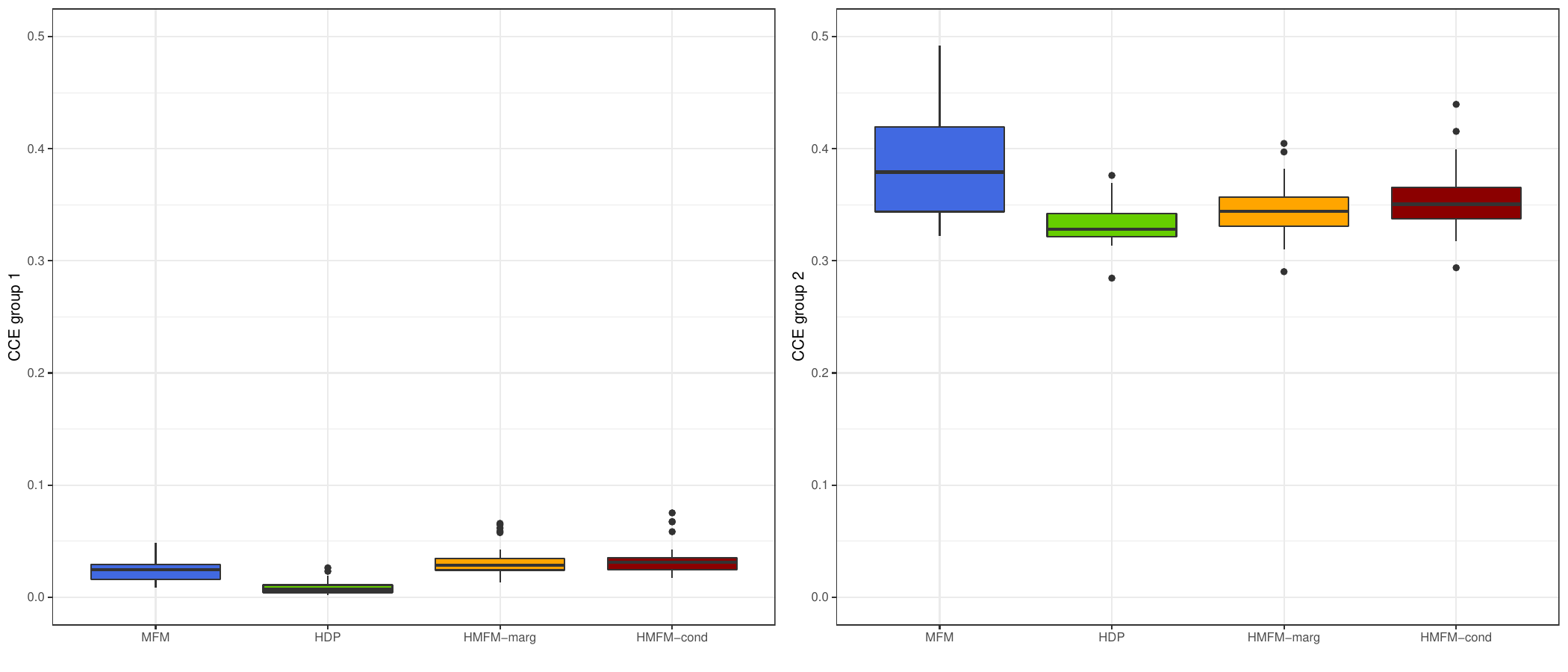}
\caption{ Co-Clustering Error in the first (left) and second (right) group. Boxplot are obtained over $50$ datasets simulated under Experiment 1. 
}
\label{fig:exp1_CCE}
\end{figure}

Similar conclusions can be drawn in terms of density estimation. 
\Cref{fig:exp1_PS} displays the boxplots for the group-specific PS. The figure shows that differences are negligible in the first group, for which all methods provide a good estimate of the density. However, the  density for the second group estimated using the MFM is worse than the one obtained by the HMFM and the HDP. 

\begin{figure}
\centering
\includegraphics[width=\linewidth]{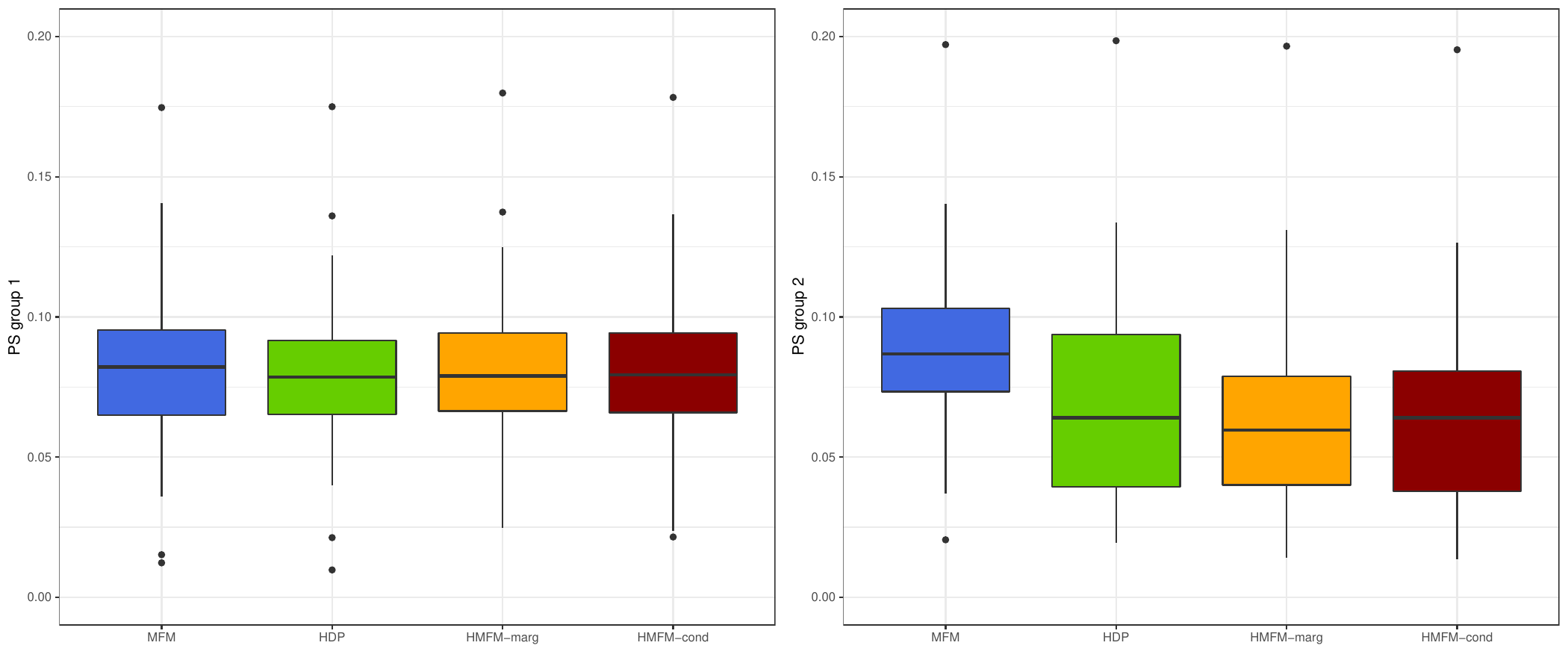}
\caption{Predictive Score in the first (left) and second (right) group. Boxplot are obtained over $50$ datasets simulated under Experiment 1. }
\label{fig:exp1_PS}
\end{figure}

\subsection{Experiment 2}
\label{subsection:exp2_app}
This experiment considers $d=2$ groups coming from a single component, not shared across the groups, namely $K_1 = 1$, $K_2 = 1$ so that $K = 2$. The component specific parameters are $\mu = (0,1)$ and $\sigma = (1,1)$, hence the sampling model is 
$y_{1,i} \ \iid \ \text{N}(0,1)$ and $y_{2,i} \ \iid \ \text{N}(1,1)$, each for $i=1,\dots,n_j$.

\begin{figure}
\centering
\includegraphics[width=0.75\linewidth]{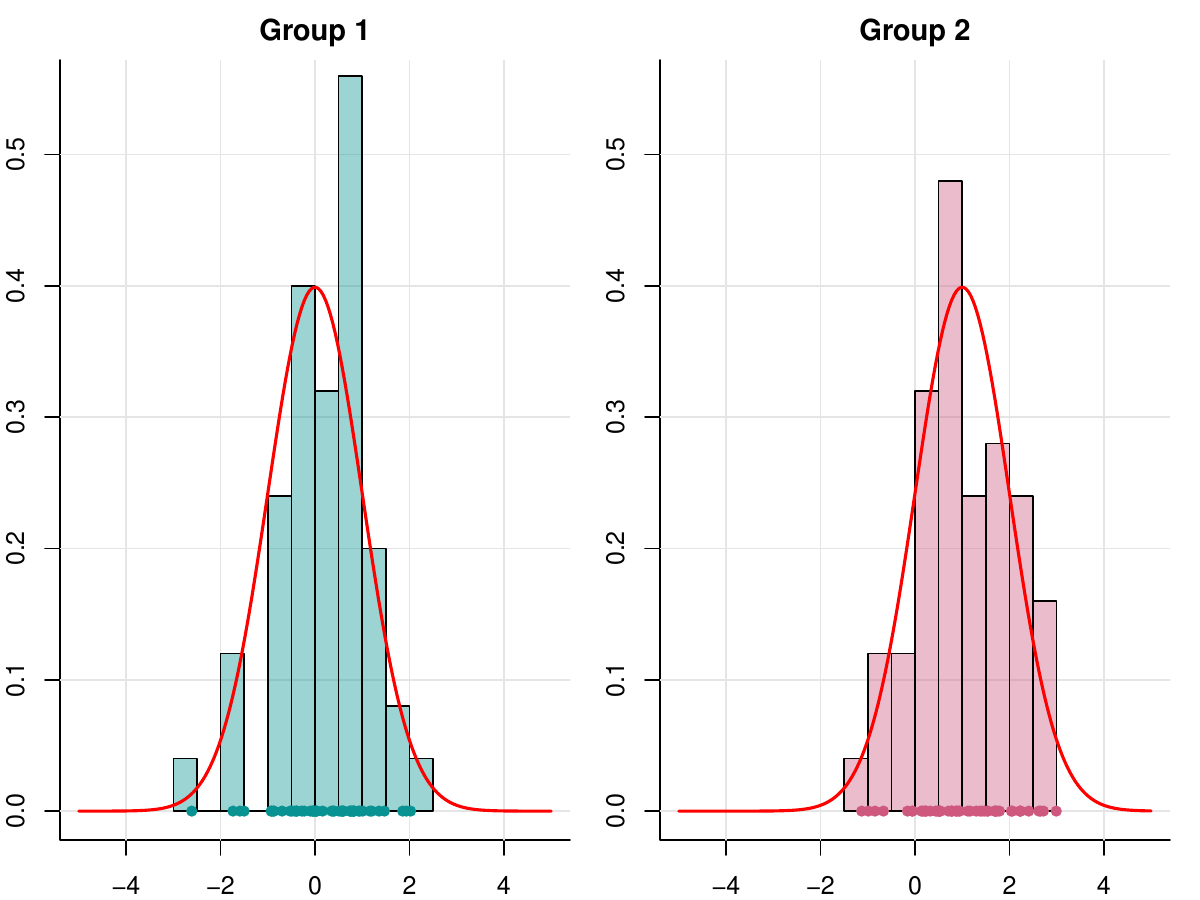}
\caption{Empirical distributions of a dataset simulated under Experiment 2. Dots represent the observations while lines represent the  underlying densities. Colours relate to the mixing components.}
\label{fig:exp2_esempio}
\end{figure}

We repeat the experiment for an increasing number of  observations, $n = 50,100,200$, which we equally divide into the two groups. 
See \Cref{fig:exp2_esempio} for an example of the empirical distribution of a dataset with n=100 generated from the model and the underlying densities. 
This experiment is designed to assess whether the borrowing of information may lead to misleading results in situations where groups do not share any features. Hence, for this scenario, we confine the  comparison between the HDP and the HMFM. 
The latter, is fitted by setting 
$\Lambda_0 = 10$, $\text{V}_\Lambda = 2$, $\gamma_0 = 0.01$.

\Cref{tab:exp2_ARI} presents the results of model comparison based on the mean and standard deviation (in brackets) of the ARI over the simulated datasets. 
Although the scenario is simple and we would expect that both methods will recover the true partition, the table shows that  the HDP struggles to find the underlying global clustering, especially for a small number of observations. Rather, the HMFM is able to recover the true partition for all the sample sizes.


\begin{table}[h]
\caption{
Mean and standard deviation (in brackets) of the ARI for the two groups over the 50 simulated datasets under Experiment 2. }
\centering
\label{tab:exp2_ARI}
\begin{tabular}{cccc}
\hline
& n = 50 & n = 100 & n = 200 \\ \hline
HDP                & 0.515 (0.484)      & 0.877 (0.310)      & 0.957 (0.172)       \\
HMFM - Marginal    & 0.980 (0.141)      & 0.999 (0.005)      & 1 (0)       \\
HMFM - Conditional & 0.978 (0.141)      & 0.994 (0.033)      & 1 (0)      \\ \hline
\end{tabular}
\end{table}

The same conclusions can be drawn from the CCE, see \Cref{fig:exp2_CCE}, which takes into account all the posterior pairwise probabilities of observations to be clustered together. 
Clearly, the HMFM  outperforms the HDP in terms of clustering estimation. Also, the proposed model provides  better estimates of the densities over the HDP, see \Cref{fig:exp2_PS}, although differences are here less evident.


This simple example showcases that the HDP's propensity to excessively share information across groups can compromise the accurate recovery of the true underlying partition. 
To understand why, we resort to the restaurant franchise metaphor discussed in \Cref{subsection:CRFP}.
According to the HDP, when a new customer arrives, the probability of consuming a dish not yet served in that specific restaurant, but available in other franchise restaurants, hinges on the cumulative number of clients eating that dish across all restaurants.
In contrast, the HMFM uses information only about whether a dish is being consumed or not in other restaurants. 
This experiment unveils a potential issue in the HDP's clustering mechanism; the concentration of all clients within the first restaurant (group) eating the same dish increases the probability of that dish being offered in the second restaurant (group), even though no components are shared between them. Rather, this confusion is circumvented by the HMFM's clustering mechanism. We notice that this phenomenon is directly tied to the choice of prior, and its impact diminishes as the sample size increases, as shown in \Cref{tab:exp2_ARI}. 
However, even though  the HDP reaches  the correct clustering for $n=200$, \Cref{fig:exp2_CCE} shows that the HMFM model still outperforms the competing approach.

\begin{figure}
\centering \includegraphics[width=1\linewidth]{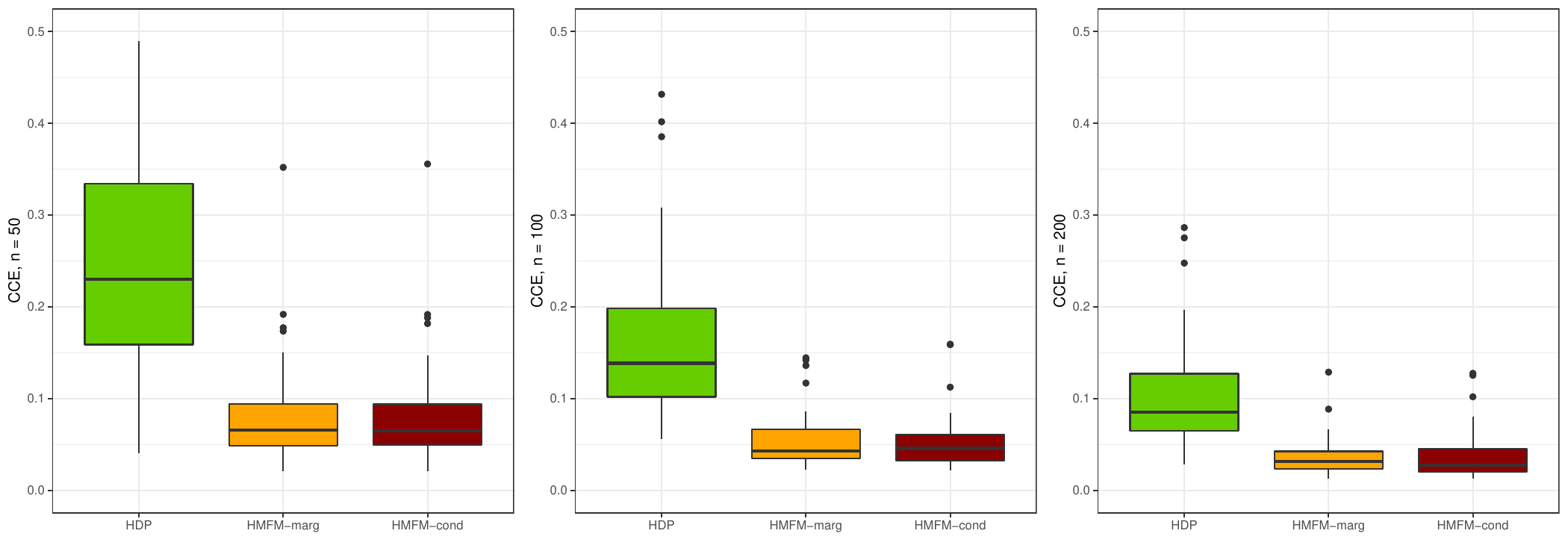}
\caption{ Co-Clustering Error for different sample sizes. Boxplots are obtained over $50$ datasets simulated under Experiment 2. }
\label{fig:exp2_CCE}
\end{figure}

\begin{figure}
\centering
\includegraphics[width=1\linewidth]{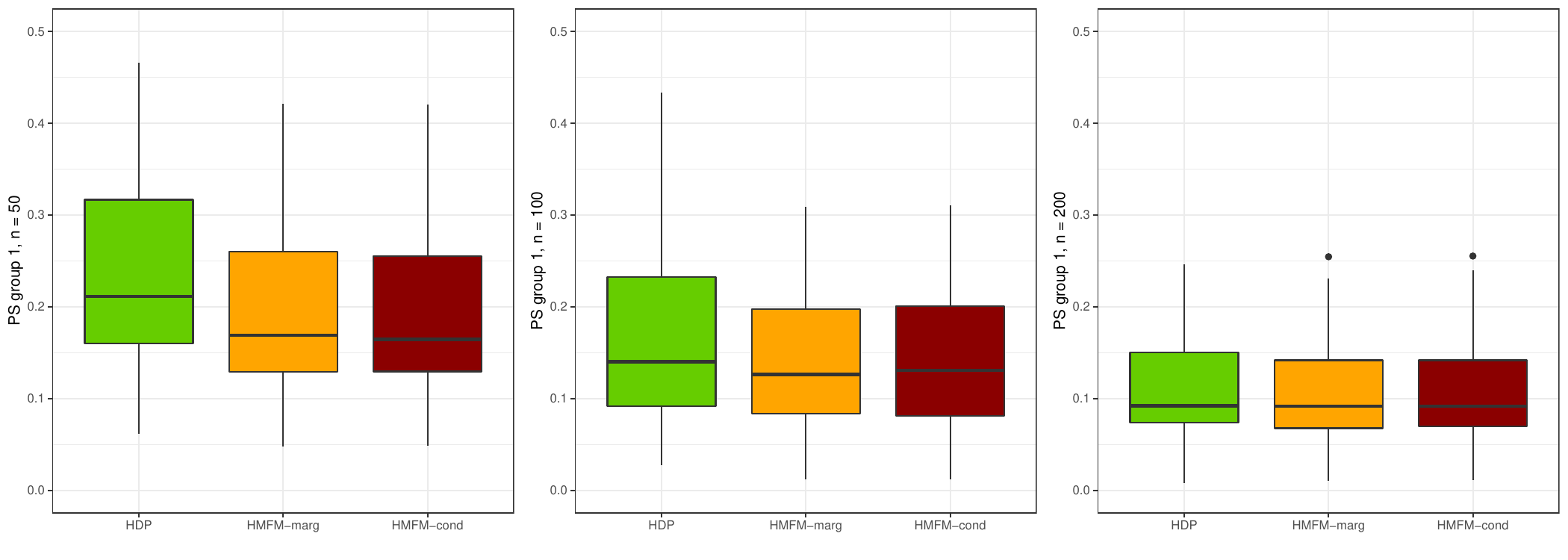} \\
\includegraphics[width=1\linewidth]{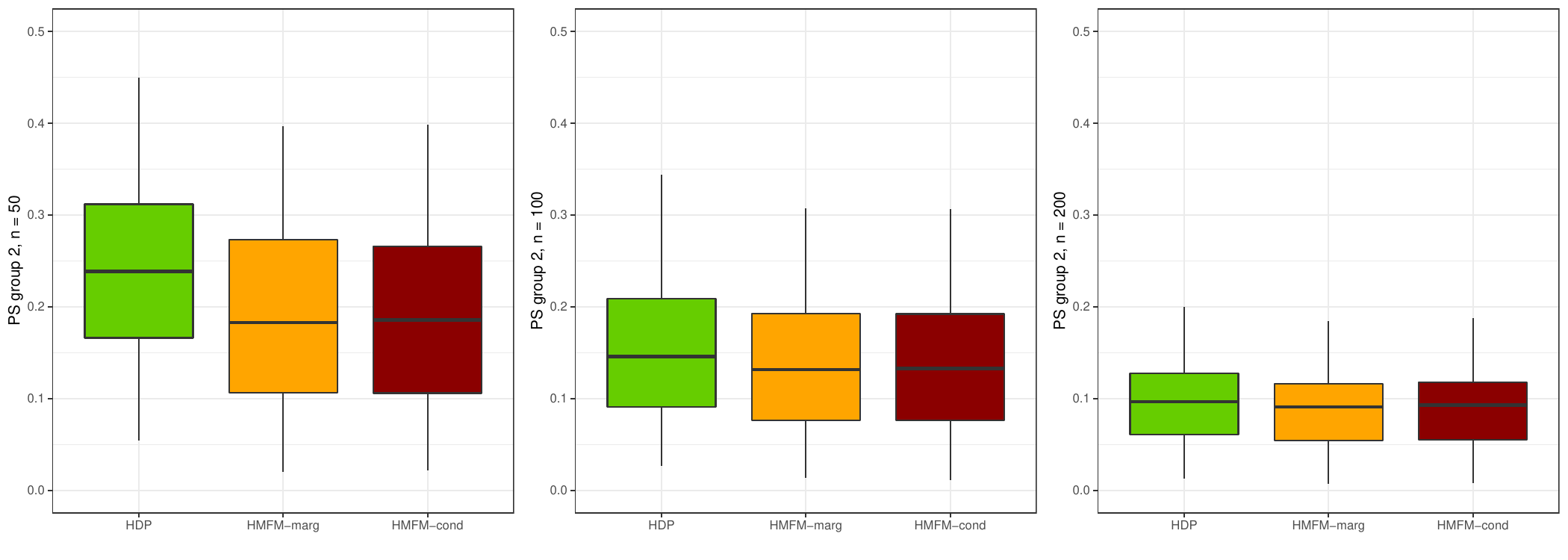} 
\caption{ Predictive score  for the first group (top panels) and the second group (bottom panels) for different sample sizes. 
Boxplot are obtained over $50$ datasets simulated under Experiment 2.}
\label{fig:exp2_PS}
\end{figure}
\subsection{Experiment 3}
\label{subsection:exp3_app}
The final experiment considers $d=15$ groups, each consisting of $n_j = 30$ observations. 
Data have been generated so that for the first $12$ groups, we consider the following component specific parameters $\mu_1 = (-3,0,1)$, $\sigma_1 = (\sqrt{0.5}, \sqrt{0.5}, \sqrt{0.5})$.
To increase the variability and differences between the groups, for each $j=1,\dots,12$ we first sample the number of local components $K_j \in \{2,3\}$ and then, conditionally to $K_j$, we choose uniformly the components to be selected. Regarding the weights, for each of the $12$ groups, we assign most of the mass to  the second component, if included, namely, $\pi_{j,2} = 0.5$ if $K_j = 3$ or $\pi_{j,2} = 2/3$ if $K_j = 2$. The remaining components receive the remaining mass split in equal parts.  
Finally, the last three groups all  share two components whose parameters are $\mu_2 = (-1.5,1.5)$, $\sigma_2 = (\sqrt{0.5},\sqrt{0.5})$ with equal mixing weights. Then, the global number of components is $K = 5$. 
\Cref{fig:exp3_esempio} shows an example of data simulated under this scenario. 
\begin{figure}
\centering
\includegraphics[width=1\linewidth]{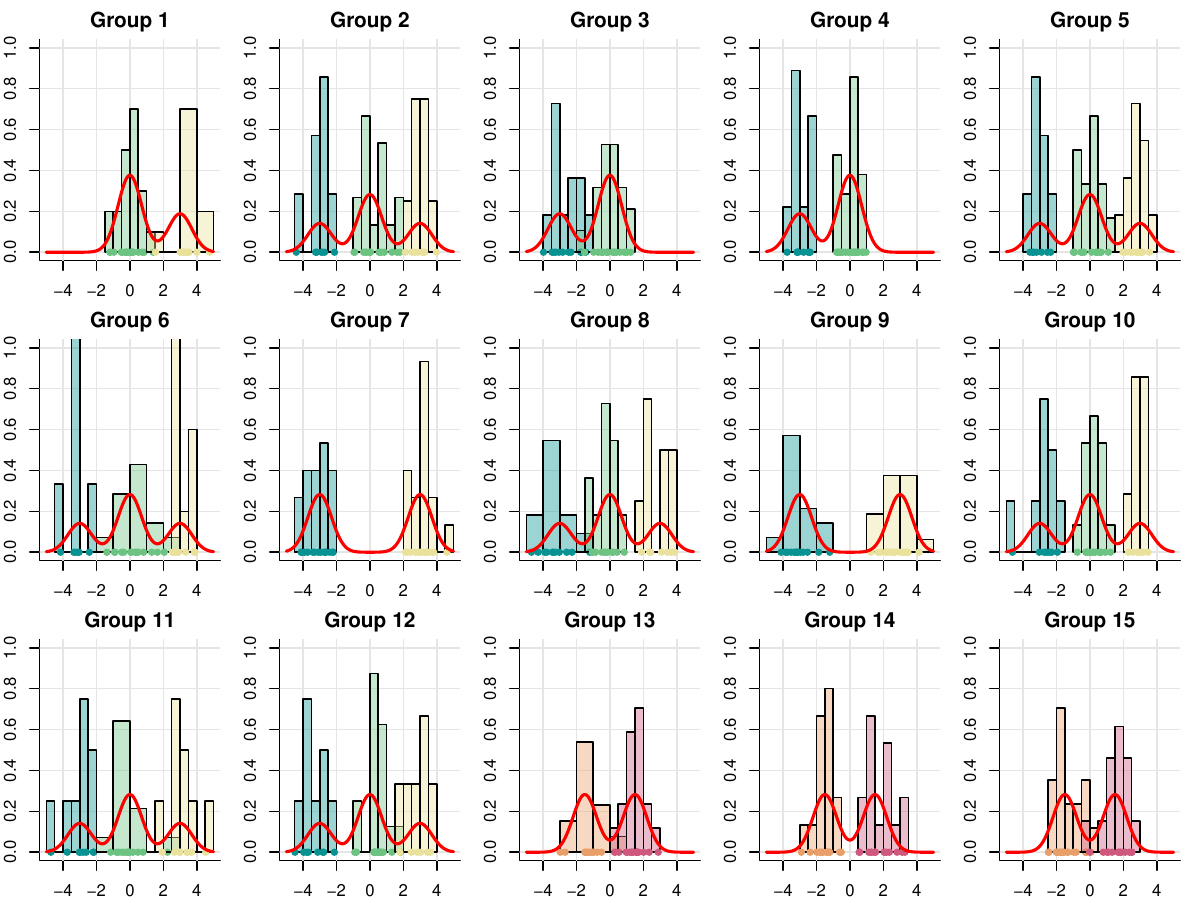}
\caption{Empirical distributions of a dataset simulated under Experiment 3. Dots represent the observations while lines represent the  underlying densities. Colours relate to the mixing components. }
\label{fig:exp3_esempio}
\end{figure}

For each simulated dataset, we fit the HDP and the HMFM; in particular, the latter, is fitted by setting the hyperparameters 
$\Lambda_0 = 15$, $\text{V}_\Lambda = 3$, $\gamma_0 = 0.05$.  \Cref{fig:exp3_clustering} shows the boxplot of the ARI (left) and  CCE (right) of the global clustering evaluated over the simulated datasets. 
Both metrics show that the HFMF significantly outperforms the HDP. 
The same conclusions can be drawn by inspecting the ARI and the CCE within each group (not shown here for brevity).

\begin{figure}
\centering
\includegraphics[width=0.45\linewidth]{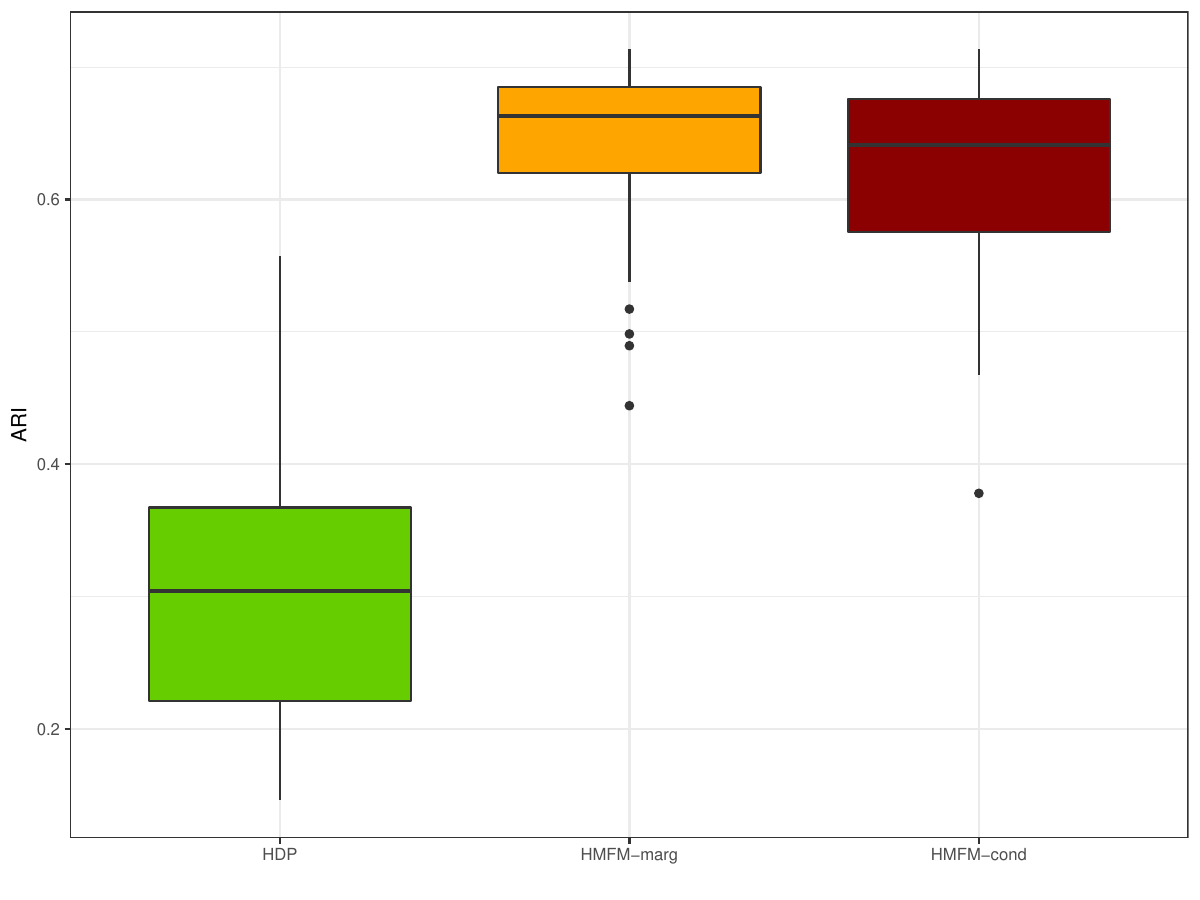} 
\hfill
\includegraphics[width=0.45\linewidth]{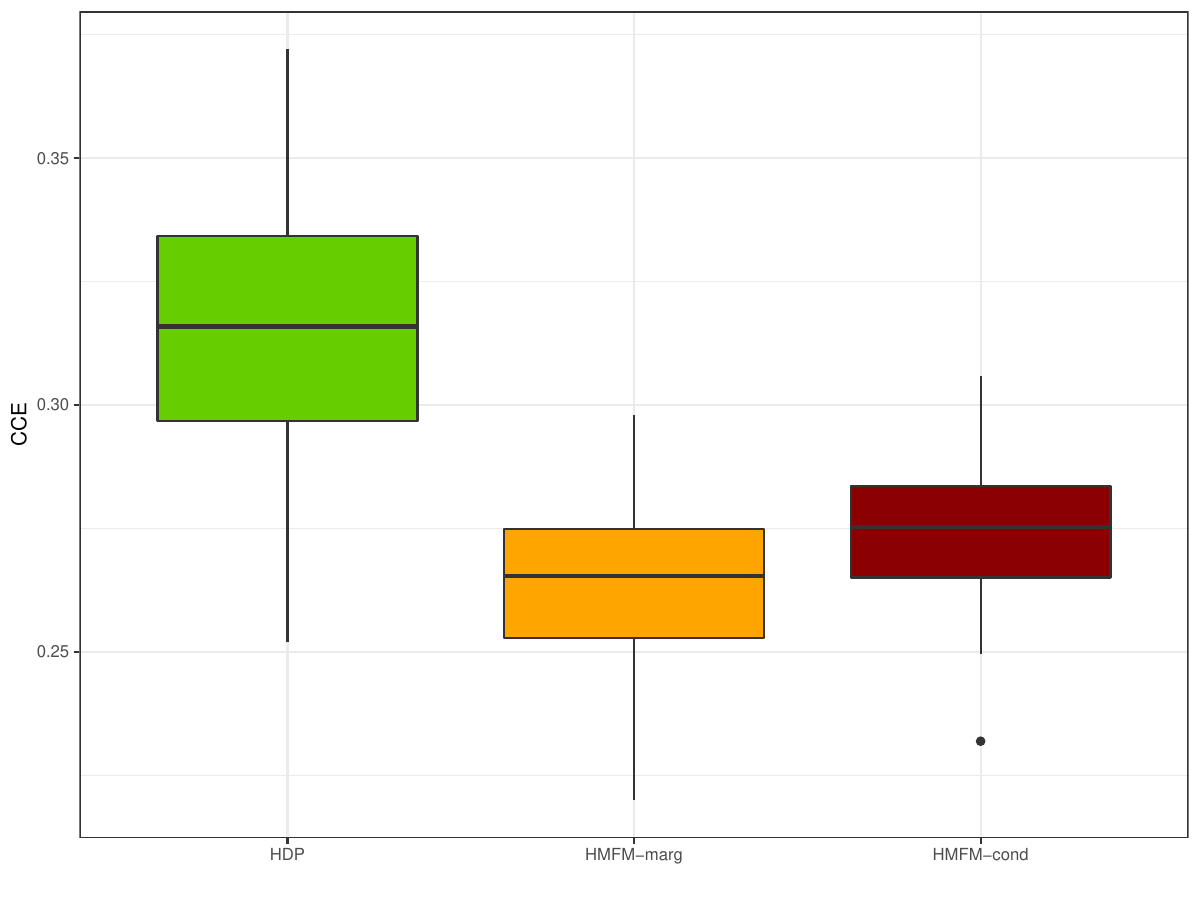} 
\caption{ ARI (left) and  CCE (right) of the global clustering evaluated over the datasets simulated under Experiment 3.  }
\label{fig:exp3_clustering}
\end{figure}

Similar to the previous experiments, the proposed mixture model outperforms the HDP also in terms of group-specific density estimation, even though the differences are less evident than for the clustering.

\begin{figure}[h]
\centering
\includegraphics[width=1\linewidth]{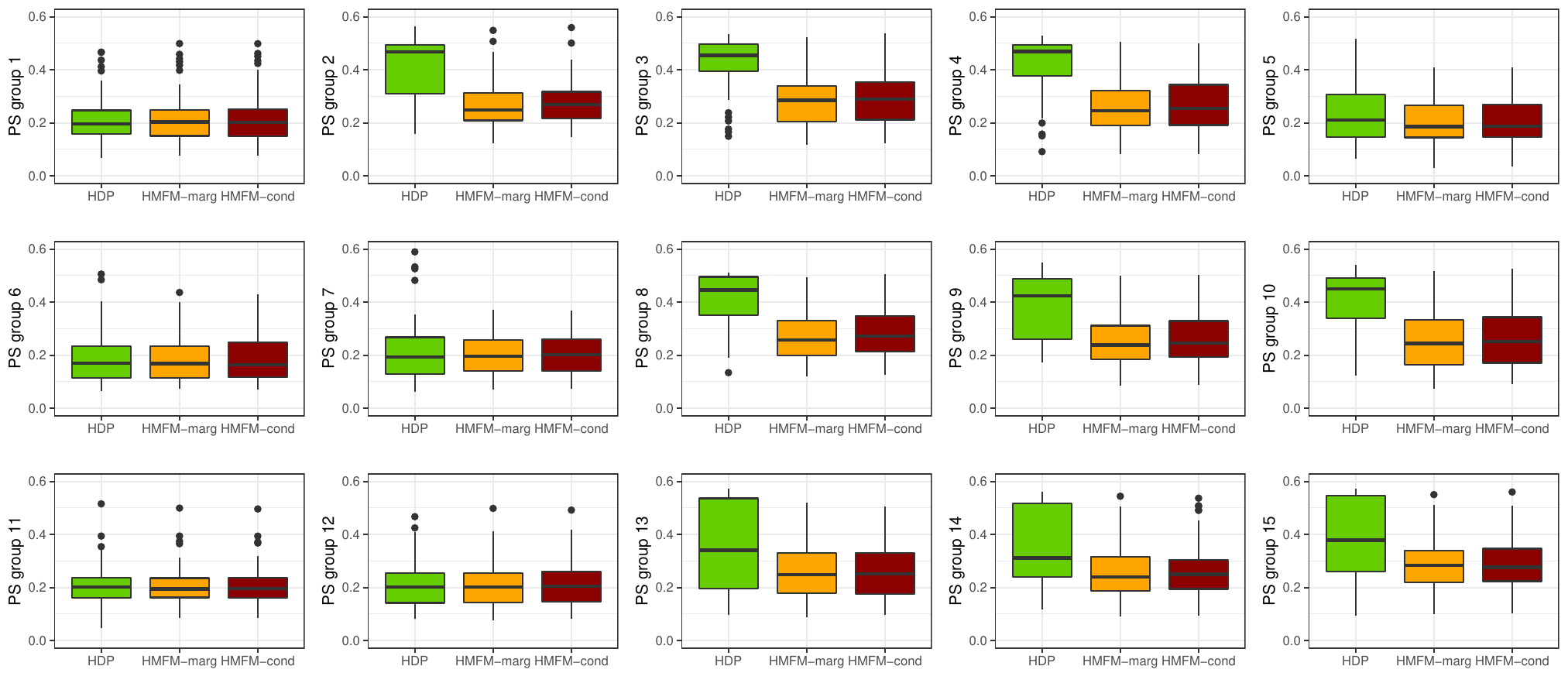}
\caption{ Predictive Score in all the groups. Boxplot are obtained over $50$ datasets simulated under Experiment 3. }
\label{fig:exp3_PS}
\end{figure}

\subsection{Comparison of computational times}
\label{app:computational_time}
The goal of this section is to asses the scalability of the HDP and the HMFM  
when the number of observations increases. In particular, we consider $n = (24,50,100,200,400,800,1600)$ data, equally divided in  $d=2$ groups. 
Data are generated from a 
simple mixture model with well separated components, enabling all methods to quickly converge to the true clustering. 
We run each method for $10,000$ iterations, burning out half of them. The total execution time is then divided by the number of iterations, so that the mean time for a single iteration is computed and displayed in \Cref{fig:comp_times}.


\begin{figure}
\centering
\includegraphics[width=0.75\linewidth]{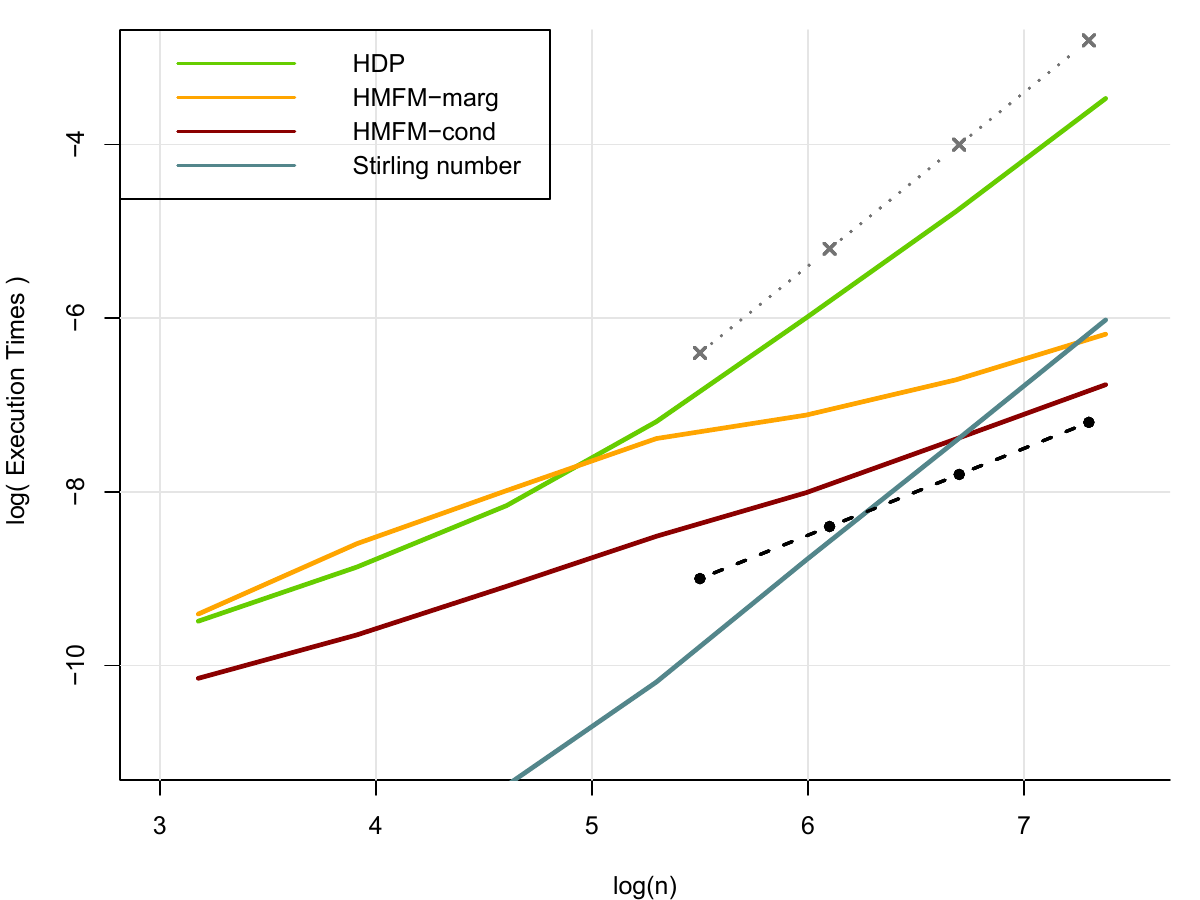}
\caption{ 
Execution time in log-scale for the HMFM,  HDP and calculation of Stirling numbers.  The dashed line and the dotted line represent, respectively, a linear and a quadratic increasing trend with the number of observations $n$.
}
\label{fig:comp_times}
\end{figure}

The HDP sampling strategy 
is based on the direct assignment scheme  suggested by \cite{teh2006hierarchical}. 
Differently from the traditional Chinese restaurant franchise process where observations are firstly assigned to the tables and then the menu (mixture components) are assigned to the tables, the direct assignment sampler  assigns data directly to the mixture components.
The computational bottleneck of this algorithm is the computation of the unsigned Stirling numbers $s(n_{j,k},m)$ for each $j=1,\dots,d$, $k=1,\dots,K$ and for all positive integers $m\leq n_{j,k}$, where $n_{j,k}$ is the number of observations 
in group $j$ assigned to cluster $k$. 
The computational time to compute $s(n_{j,k},m)$ is $O((n_{j,k})^2)$. Hence, iterating over all clusters in all groups, the overall time for updating all cluster allocation labels is $O(n^2)$. This is empirically shown in \Cref{fig:comp_times}. 
As expected, for large $n$,  the slope of the computational time needed to calculate the Striling numbers equals the one of the HDP. 
In particular, such slope is steeper than the one of HMFM  (both marginal and conditional sampler) 
which grows with a linear trend. This empirically verifies that our proposed model is faster and scales much better when increasing the number of observations.

\section{Additional results on shot put data analysis}
\label{app:shotput}
Here we detail the complete specification of the HMFM model for the shot put data:
\begin{equation*}
\label{eqn:shotput_mixtureFDP}
\begin{split}
\bmy_{j,1}, \dots, \bmy_{j,n_j} \mid \bmS_j, \theta, M 
& \ \iid \
\sum_{m=1}^M \frac{S_{j,m}}{T_j} 
d\text{N}_{N_{j,i}}(\bmy_{j,i} \mid \mu_{m}\mathbf{1}_{N_{j,i}} + X_{j,i} \bmbeta_j, \sigma^2_{m}\mathbf{I}_{N_{j,i}}),
\quad \text{for } j=1,\dots, d\\
\bmbeta_1,\dots,\bmbeta_d 
&\ \iid \
\text{N}_d\left( \bmbeta_0, \Sigma_0 \right)\\
\tau_{1},\dots,\tau_{M} \mid M
&\ \iid \ \operatorname{Normal-InvGamma}\left(\mu_0,k_0,\nu_0,\sigma_0\right)\\
S_{j,1},\dots,S_{j,M} \mid M,\gamma_j
& \ \iid \ \operatorname{Gamma}(\gamma_j,1),
\quad \text{for } j = 1,\dots,d\\
M \mid \Lambda & \ \sim \ \operatorname{Pois}_1(\Lambda)\\
\gamma_1,\dots,\gamma_d \mid \Lambda 
& \ \sim \
\operatorname{Gamma}(a_\gamma,\Lambda b_\gamma) \\
\Lambda 
& \ \sim \
\operatorname{Gamma}(a_\Lambda, b_\Lambda),
\end{split}
\end{equation*}
where $d\text{N}_{N_{j,i}}(\bmy \mid \bmmu,\Sigma)$ is the density of a multivariate normal distribution evaluated in $\bmy$ with mean $\bmmu$ and covariance matrix $\Sigma$.


\begin{figure}[h]
\centering
\includegraphics[width=0.48\linewidth]{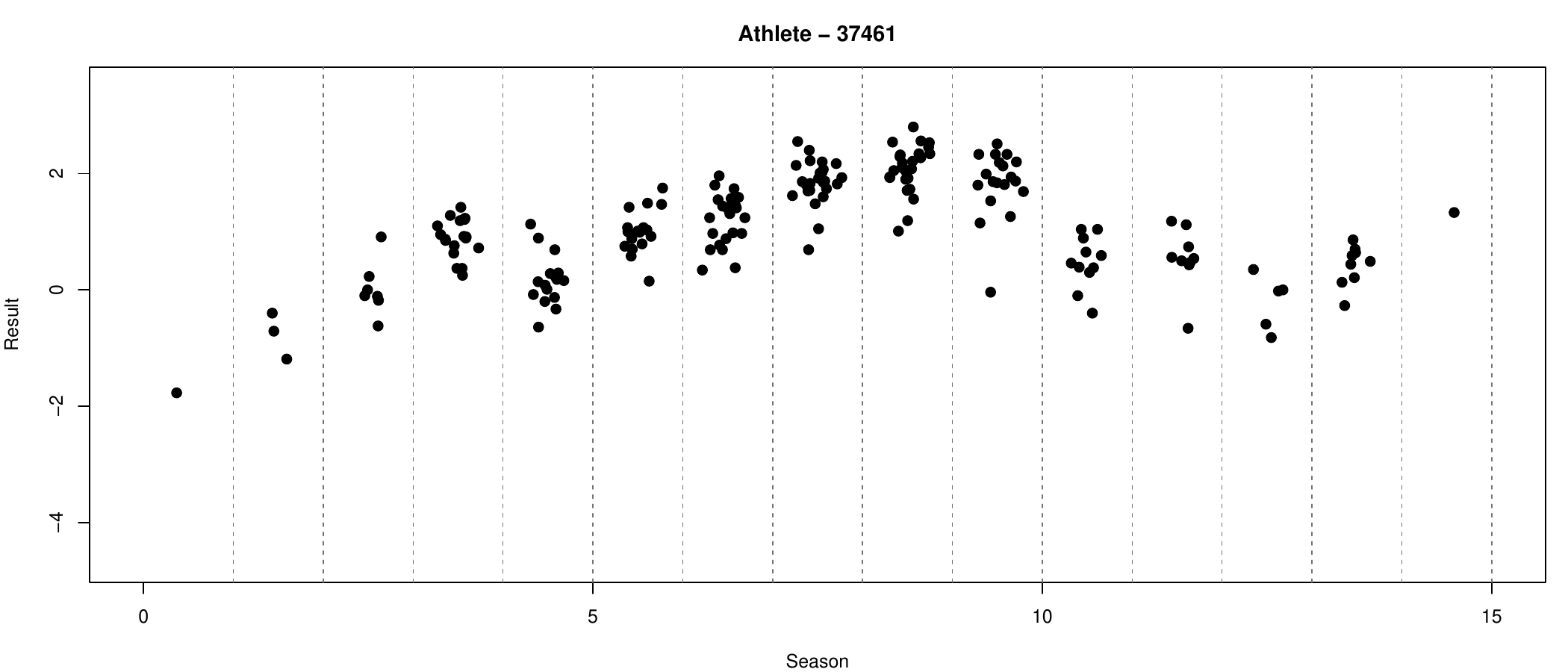}
\hfill
\includegraphics[width=0.48\linewidth]{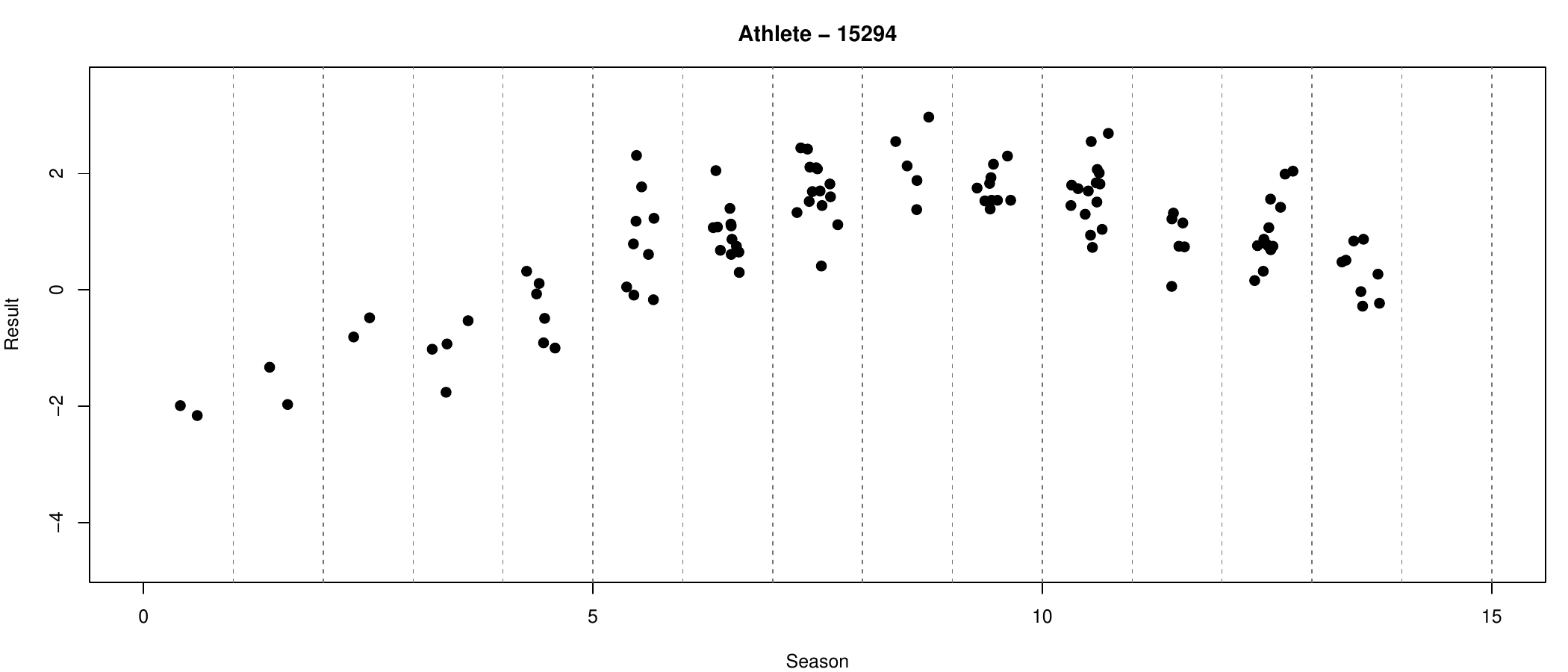}\\
\includegraphics[width=0.48\linewidth]{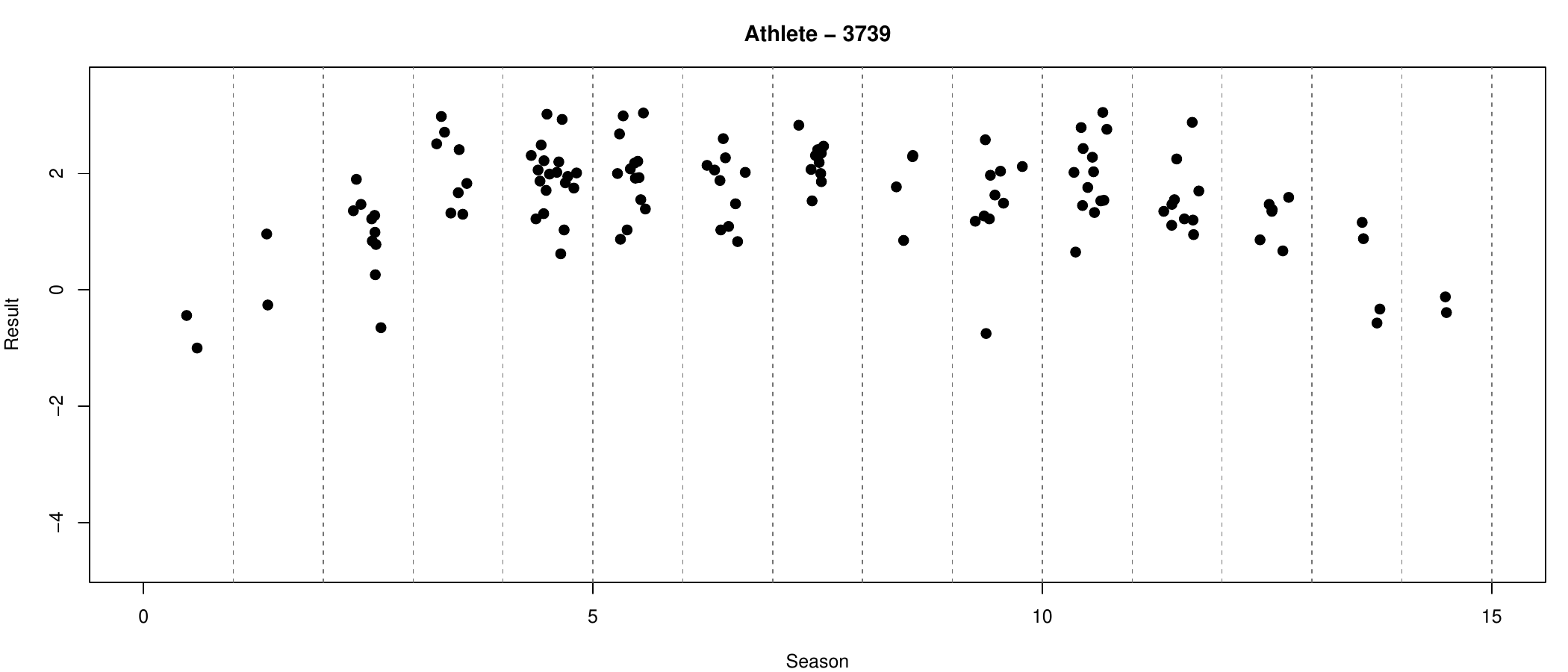}
\hfill
\includegraphics[width=0.48\linewidth]{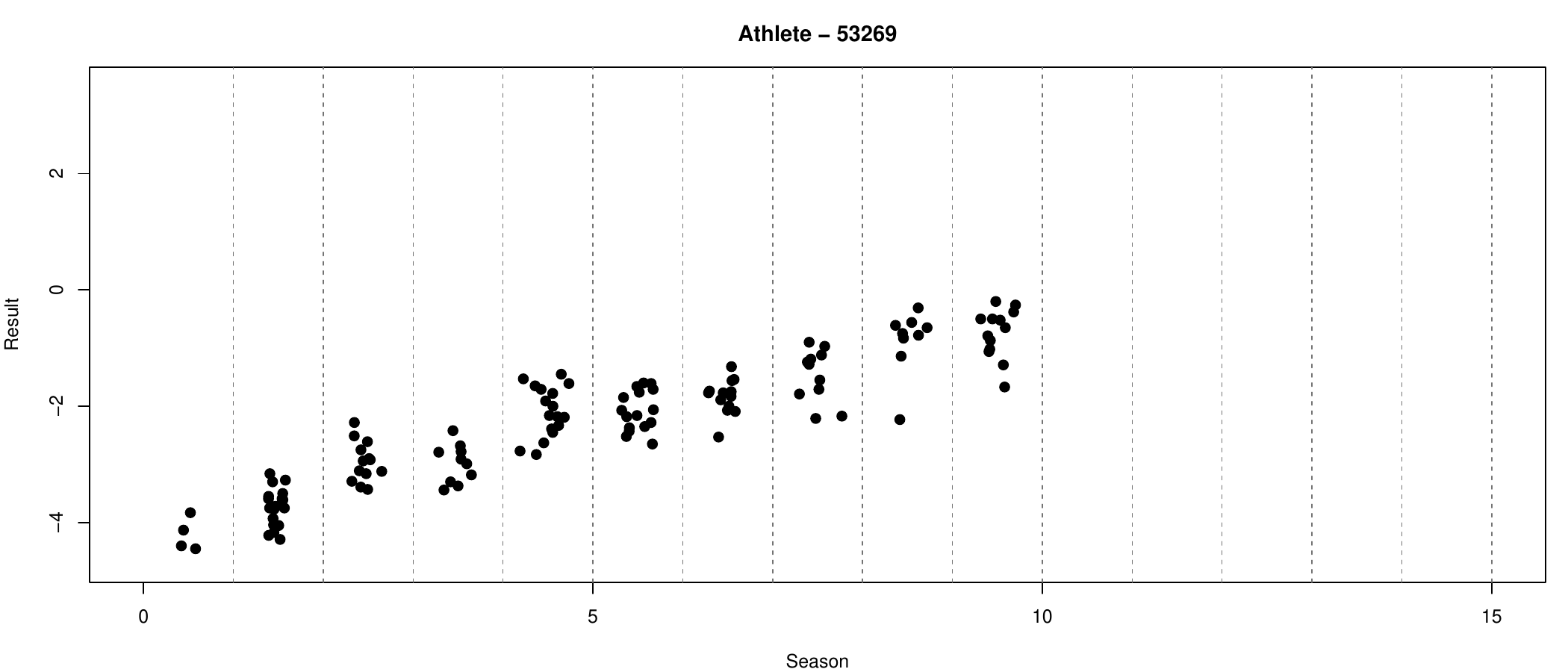}
\caption{Shot put measurements collected throughout an athlete's career for four randomly selected athletes. Vertical dotted lines delimit seasons. 
}
\label{fig:application_traiettorie}
\end{figure}

\Cref{tab:app_cluster_sizes} presents the local clustering sizes, for each season and for each cluster. Finally, \Cref{tab:my-table} shows some summary statistics for the global clusters and compares them with the overall dataset. 

\begin{table}[h]
\centering
\caption{Cluster cardinalities for each cluster (rows) and for each season (columns). }
\begin{tabular}{@{}c@{\hspace{0.5em}}*{15}{c@{\hspace{0.5em}}}@{}}
\hline \empty & \textbf{S1} & \textbf{S2} & \textbf{S3} & \textbf{S4} & \textbf{S5} & \textbf{S6} & \textbf{S7} & \textbf{S8} & \textbf{S9} & \textbf{S10} & \textbf{S11} & \textbf{S12} & \textbf{S13} & \textbf{S14} & \textbf{S15}             \\ \hline
Cluster 1  &  0   &  4   &  5   &  9   &  7  &  12 &  6  &  9  &  11 &  7  &  7  &  5  &  6  &  5  &  3 \\
Cluster 2  &  14  &  10  &  14  &  13  &  19 &  16 &  21 &  9  &  11 &  14 &  4  &  9  &  3  &  6  &  4 \\
Cluster 3  &  11  &  14  &  10  &  31  &  34 &  31 &  34 &  31 &  24 &  17 &  17 &  0  &  8  &  2  &  0 \\
Cluster 4  &  0   &  14  &  20  &  45  &  54 &  51 &  38 &  35 &  46 &  24 &  10 &  20 &  0  &  0  &  9 \\
Cluster 5  &  33  &  29  &  41  &  61  &  70 &  66 &  49 &  39 &  20 &  11 &  12 &  0  &  17 &  6  &  0 \\
Cluster 6  &  0   &  0   &  14  &  11  &  7  &  11 &  19 &  0  &  0  &  0  &  4  &  1  &  0  &  0  &  0 \\
Cluster 7  &  9   &  39  &  70  &  108 &  79 &  42 &  36 &  52 &  29 &  21 &  21 &  29 &  0  &  0  &  7 \\
Cluster 8  &  29  &  85  &  109 &  79  &  72 &  59 &  48 &  23 &  30 &  29 &  4  &  0  &  7  &  12 &  0 \\
Cluster 9 &  75  &  148 &  90  &  42  &  38 &  46 &  34 &  36 &  22 &  13 &  18 &  9  &  16 &  0  &  2 \\
Cluster 10 &  200 &  43  &  30  &  4   &  3  &  7  &  9  &  21 &  18 &  22 &  12 &  16 &  0  &  0  &  6 \\
Cluster 11 &  31  &  17  &  0   &  0   &  0  &  0  &  0  &  0  &  0  &  5  &  12 &  0  &  0  &  11 &  0 \\
Cluster 12 &  1   &  0   &  0   &  0   &  0  &  0  &  0  &  0  &  0  &  0  &  0  &  0  &  3  &  0  &  2 \\
Cluster 13  &  0   &  0   &  0   &  0   &  0  &  0  &  0  &  3  &  0  &  0  &  0  &  0  &  0  &  0  &  0 \\
Cluster 14 &  0   &  0   &  0   &  0   &  0  &  0  &  0  &  0  &  0  &  0  &  0  &  0  &  2  &  0  &  0 \\
Cluster 15 &  0   &  0   &  0   &  0   &  0  &  1  &  1  &  0  &  0  &  0  &  0  &  3  &  0  &  0  &  0 
\end{tabular}
\label{tab:app_cluster_sizes}
\end{table}

\begin{table}[h]
\caption{Cluster statistics for each global cluster. The first row represents the case where all data are pooled in a single cluster. Each subsequent row corresponds to a specific global cluster. The columns display the sample mean of various cluster statistics, differentiated by M (male) or F (female) athletes in the cluster.}
\centering
\renewcommand{\arraystretch}{1.0} 
\begin{tabular}{@{}c@{\hspace{0.5em}}c@{\hspace{0.5em}}c@{\hspace{0.5em}}c@{\hspace{0.5em}}c@{\hspace{0.5em}}c@{\hspace{0.5em}}c@{\hspace{0.5em}}c@{\hspace{0.5em}}c@{}}
\hline
\textbf{Cluster} & \textbf{SizeM} & \textbf{SizeF} & \textbf{MeanM} & \textbf{MeanF} & \textbf{VarM} & \textbf{VarF} &  \textbf{AgeM} &  \textbf{AgeF} \\ \hline
Pooled &  1925 &  1683 &  1.12 &  -1.29 &  1.51 &  2.16 &  24.47 &  22.77 \\
1 &  42 &  54 &  3.47 &  1.50 &  0.28 &  0.48 &  27.70 &  26.08 \\
2 &  77 &  90 &  2.85 &  0.72 &  0.28 &  0.38 &  27.52 &  24.91 \\
3 &  147 &  117 &  2.26 &  0.17 &  0.28 &  0.42 &  25.79 &  25.05 \\
4 &  234 &  132 &  1.70 &  -0.39 &  0.29 &  0.42 &  25.80 &  23.54 \\
5 &  260 &  194 &  1.18 &  -1.05 &  0.21 &  0.33 &  24.41 &  22.34 \\
6 &  42 &  25 &  1.14 &  -1.02 &  0.86 &  0.98 &  24.24 &  21.69 \\
7 &  312 &  228 &  0.72 &  -1.56 &  0.28 &  0.39 &  23.84 &  22.17 \\
8 &  316 &  270 &  0.24 &  -2.08 &  0.26 &  0.34 &  22.71 &  21.85 \\
9 &  277 &  309 &  -0.37 &  -2.68 &  0.28 &  0.34 &  22.18 &  21.38 \\
10 &  178 &  212 &  -1.07 &  -3.30 &  0.23 &  0.33 &  22.05 &  21.03 \\
11 &  35 &  41 &  -1.81 &  -3.56 &  0.29 &  1.20 &  21.40 &  25.78 \\
12 &  0 &  6 &  - &  2.63 &  - &  0.32 &  - &  28.05 \\
13 &  2 &  1 &  1.08 &  0.00 &  13.05 &  1.83 &  26.96 &  28.00 \\
14 &  1 &  1 &  -0.46 &  -4.19 &  2.19 &  0.47 &  30.80 &  31.00 \\
15 &  2 &  3 &  -4.76 &  -4.20 &  18.10 &  2.60 &  36.33 &  25.42 \\
\end{tabular}
\label{tab:my-table}
\end{table}





\clearpage
\bibliographystyle{jasa3}
\bibliography{references}

\end{document}